\documentclass[a4paper,reqno]{amsart}
\usepackage[a4paper]{geometry}
\usepackage[T1]{fontenc}
\usepackage{amssymb}
\usepackage[initials]{amsrefs}
\usepackage{mathtools}
\usepackage{microtype}
\usepackage{enumerate}
\numberwithin{equation}{section}
\theoremstyle{plain}
\newtheorem{theorem}{Theorem}[section]
\newtheorem{proposition}[theorem]{Proposition}
\newtheorem{lemma}[theorem]{Lemma}
\theoremstyle{remark}

\newtheorem*{remark*}{Remark}
\newtheorem*{remarks*}{Remarks}
\newtheorem*{stepone*}{Step 1}
\newtheorem*{steptwo*}{Step 2}
\newtheorem*{stepthree*}{Step 3}
\newtheorem*{stepfour*}{Step 4}
\newtheorem*{stepfive*}{Step 5}
\theoremstyle{definition}

\newtheorem*{definition*}{Definition}
\newtheorem*{assumptions*}{Assumptions}
\newtheorem*{notation*}{Notation}
\newtheorem*{notations*}{Notations}
\newtheorem*{properties*}{Properties}

\providecommand{\BS}[1]{\boldsymbol{#1}}
\providecommand{\C}[1]{\mathcal{#1}}
\providecommand{\D}[1]{\mathbb{#1}}
\providecommand{\F}[1]{\mathfrak{#1}} 
\providecommand{\R}[1]{\mathrm{#1}}
\providecommand{\abs}[1]{\lvert#1\rvert}
\providecommand{\accol}[1]{\lbrace#1\rbrace}
\providecommand{\av}[1]{\langle#1\rangle}
\providecommand{\croch}[1]{\lbrack#1\rbrack}
\providecommand{\norm}[1]{\lVert#1\rVert}
\providecommand{\scal}[1]{\langle#1\rangle}
\renewcommand{\Im}{\operatorname{Im}}
\renewcommand{\Re}{\operatorname{Re}}
\newcommand{\cercle}{\mathbb{T}}
\newcommand{\dd}{\mathrm{d}}
\newcommand{\eul}{\mathrm{e}}
\newcommand{\ii}{\mathrm{i}}
\newcommand{\vece}{\mathrm{e}}
\DeclareMathOperator{\card}{card}
\DeclareMathOperator{\class}{C}
\DeclareMathOperator{\const}{const}
\DeclareMathOperator{\dist}{dist}
\DeclareMathOperator{\id}{id}

\DeclareMathOperator{\ord}{O}
\DeclareMathOperator{\osmall}{o}
\DeclareMathOperator{\supp}{supp}
\DeclareMathOperator{\tr}{tr}
 
\newcommand{\spec}{\sigma} 
\newcommand{\uw}{\underline\omega} 
\newcommand\ut{\underline t} 
\begin{document} 
\title[Large eigenvalues of Jaynes--Cummings type models]{Asymptotic behavior of large eigenvalues\\of Jaynes--Cummings type models}
\author[A. Boutet de Monvel]{Anne Boutet de Monvel$^1$}
\address{$^1$Institut de Math\'ematiques de Jussieu\\
Universit\'e Paris Diderot\\
b\^atiment Sophie Germain, 
case 7012, 75205 Paris Cedex 13, France}
\email{anne.boutet-de-monvel@imj-prg.fr}
\author[L. Zielinski]{Lech Zielinski$^2$}
\address{$^2$LMPA, 
Universit\'e du Littoral,
50 rue Ferdinand Buisson,
CS 80699, 62228 Calais Cedex, France}
\email{Lech.Zielinski@lmpa.univ-littoral.fr}
\subjclass[2010]{Primary 47B36; Secondary 47A10, 47A75, 15A42, 47A55}
\keywords{Jacobi matrices, Jaynes--Cummings model, eigenvalue estimates}
\date{April 18, 2016} 
\begin{abstract} 
We consider a class of unbounded self-adjoint operators including the Hamiltonian of the Jaynes--Cummings model without the rotating-wave approximation (RWA). The corresponding operators are defined by infinite Jacobi matrices with discrete spectrum. Our purpose is to give the asymptotic behavior of large eigenvalues.   
\end{abstract}
\maketitle 
\section{Introduction}
\subsection{Jaynes--Cummings model}

We call ``Jaynes--Cummings model'' a self-adjoint operator $J$ defined in $l^2(\D{N}^*)$ by an infinite real Jacobi matrix  
\begin{equation} \label{11} 
J=\begin{pmatrix}  
d(1) & a(1) &0& 0 &\dots\\ 
a(1) & d(2) & a(2) &0& \dots\\ 
0 & a(2) & d(3) & a(3) &\dots\\ 
0 &0& a(3) & d(4) &\dots\\  
\vdots &\vdots&\vdots&\vdots&\ddots 
\end{pmatrix} 
\end{equation}
whose entries are of the form  
\begin{subequations} \label{12}
\begin{align} \label{12d}
d(k)&=k+(-1)^k\rho\\
\label{12a}  
a(k)&=a_1k^{1/2}    
\end{align} 
\end{subequations}
where $\rho$ and $a_1>0$ are real constants. The study of this kind of operators is motivated by the Hamiltonian of the Jaynes--Cummings model without the rotating-wave approximation (RWA) (see \`E.~A.~Tur \cite{Tur}).

The self-adjoint operator $J$ associated to the Jacobi matrix \eqref{11} acts on $l^2(\D{N}^*)$ by
\begin{equation}\label{1J}
(Jx)(k)=d(k)x(k)+a(k)x(k+1)+a(k-1)x(k-1) 
\end{equation} 
($x(0)=a(0)=0$). It is defined on $\C{D}\coloneqq\bigl\{x\in l^2(\D{N}^*):\sum_{k=1}^{\infty}d(k)^2\abs{x(k)}^2<\infty\bigr\}$. According to \eqref{12} the diagonal entries $d(k)$ are dominant and tend to $\infty$ with $k$. The self-adjoint operator $J$ is then bounded from below with compact resolvent (see \cite{CJ}), and we denote by 
\[
\lambda_1(J)\leq\dots\leq\lambda_n(J)\leq\lambda_{n+1}(J)\leq\dots 
\] 
its eigenvalues, enumerated in non-decreasing order, counting multiplicities. The aim of this paper is to describe the asymptotic behavior of $\lambda_n(J)$ when $n\to\infty$.

\begin{theorem}[Jaynes--Cummings model]  \label{thm:11} 
Let $J$ be the self-adjoint operator defined by \eqref{1J} with
\[
\begin{cases}
d(k)=k+(-1)^k\rho&\\
a(k)=a_1k^{1/2}&    
\end{cases} 
\]
where $\rho$ and $a_1>0$ are real constants. We assume $\abs{\rho}<1/2$. Then the $n$-th eigenvalue $\lambda_n(J)$ has the large $n$ asymptotics 
\begin{equation} \label{E0}
\lambda_n(J)=n-a_1^2+\ord(n^{-1/4}\ln n).
\end{equation} 
\end{theorem}
   
In Section~\ref{sec:12} we compare our results with other known results. In Section~\ref{sec:13} we state Theorem~\ref{thm:12} which is a generalization of Theorem~\ref{thm:11} motivated by the paper of A.~Boutet de Monvel, S.~Naboko, L.~O.~Silva \cite{BNS}. Theorem~\ref{thm:12} gives the large $n$ asymptotics of $\lambda_n(J)$ for Jacobi matrices \eqref{11} whose entries are of the form  
\begin{subequations} \label{JC}
\begin{align} \label{JCd}
d(k)&=k+v(k)\\ 
\label{JCa} 
a(k)&=a_1k^{\gamma}+a_1'k^{\gamma-1}+\ord(k^{\gamma-2})
\end{align} 
\end{subequations} 
where $v\colon\D{N}^*\to\D{R}$ is periodic and $a_1>0$, $a_1'$, and $0<\gamma\leq 1/2$ are real constants. Section~\ref{sec:14} gives the plan of the paper. Section~\ref{sec:15} lists the main notations.

\subsection{Modified Jaynes-Cummings models}  \label{sec:12}

In this section we recall known results about the asymptotic behavior of large eigenvalues for ``modified Jaynes--Cummings models'', i.e., for Jacobi matrices \eqref{11} with entries of the form 
\begin{equation} \label{mJC}
\begin{cases}    
d(k)=k^{\alpha}+v(k)&\\  
a(k)=a_1k^{\gamma}&    
\end{cases} 
\end{equation}
where $\alpha>\gamma>0$, $a_1>0$ are real constants, and $v\colon\D{N}^*\to\D{R}$ is periodic. It turns out that the large $n$ asymptotic behavior of $\lambda_n(J)$ strongly depends on whether $\alpha-\gamma>1$ or not. 

\subsubsection*{Asymptotics of large eigenvalues with persistent periodic oscillations}
In the easy case $\alpha-\gamma>1$ it is possible to apply approximation methods based on the idea of successive diagonalizations which was first applied to the problem of eigenvalue asymptotics of Jacobi matrices in the paper of J.~Janas and S.~Naboko \cite{JN}. The name ``modified Jaynes-Cummings models'' was then introduced in the paper of A.~Boutet de Monvel, S.~Naboko, L.~O.~Silva \cite{BNS} treating the case of entries of the form \eqref{mJC} with $\alpha=2$ and $\gamma=\frac{1}{2}$. The asymptotic behavior obtained in \cite{BNS} in that case was  
\[     
\lambda_n(J)=n^2+v(n)+\ord(n^{-1}). 
\]  
More general results of M.~Malejki \cite{Ma} and A.~Boutet de Monvel, L.~Zielinski \cite{BZ2} for the case of entries of the form \eqref{mJC} give as large $n$ asymptotics 
\[   
\lambda_n(J)=n^{\alpha}+v(n)+\ord(n^{\gamma-2\kappa}+n^{2\gamma-\alpha})
\] 
where $\kappa\coloneqq\alpha-1-\gamma>0$. Moreover, under the additional conditions $\alpha \leq 2$ and $\gamma<\frac{2}{3}(\alpha-1)$ we have $\alpha-2\gamma >0$ and $2\kappa-\gamma=2(\alpha-1)-3\gamma>0$, hence we obtain the asymptotic behavior
\begin{equation} \label{100}
\lambda_n(J)-n^{\alpha}=v(n)+\osmall(1) 
\end{equation} 
reflecting the oscillations determined by the periodic nature of $v$. 
 
\subsubsection*{Asymptotics of large eigenvalues without periodic oscillations}
The case $\alpha=1$ and $0<\gamma\leq 1/2$ investigated in this paper exhibits a radical change in the asymptotic behavior of $\lambda_n(J)$. The new phenomenon is the absence of periodic oscillations of large eigenvalues. This phenomenon was already described in our earlier paper \cite{BZ4} treating the case $\alpha=1$ and $0<\gamma<\frac{1}{2}$. In this paper we follow the general framework of \cite{BZ4} but in order to address the case $\gamma=\frac{1}{2}$ we need to improve the remainder estimates. To that end, we refine our approach constructing suitable approximations by means of truncated Fourier series. After submission of this paper we learned about \cite{Yan} where \eqref{E0} is proved, but with a weaker estimate. 

\subsection{Jaynes--Cummings type models} \label{sec:13}

In this paper we consider ``Jaynes--Cummings type models'', i.e., Jacobi matrices \eqref{11} with entries of type 
\begin{align*}
d(k)&=k+v(k),\\ 
\label{tJCa} 
a(k)&\asymp k^{\gamma},
\end{align*} 
where $v\colon\D{N}^*\to\D{R}$ is periodic of period $N\geq 1$ and $0<\gamma\leq 1/2$. Let us denote by 
\[
\av{v}\coloneqq\frac{1}{N}\sum_{1\leq k\leq N}v(k)
\]
the ``mean value'' of $v$ and by
\begin{equation}  \label{110} 
\rho_N=\rho_N(v)\coloneqq\max_{1\leq k\leq N}\abs{v(k)-\av{v}}.
\end{equation}
the maximum mean absolute deviation.

\begin{assumptions*}  
\begin{enumerate}[{(H}1)]
\item
$v$ is ``weakly dispersive'', in the sense that
\begin{equation}\label{H1}
\rho_N<
\begin{cases} 
\frac{1}{2}&\text{if }N=2,\\
\frac{1}{\pi\sqrt{N}}&\text{if }N\geq 3.
\end{cases}
\end{equation}
\item
$a(k)\asymp k^{\gamma}$ with $\class^2$ regularity, i.e.,
\begin{subequations}\label{H2}
\begin{align} \label{H2ak}
&ck^{\gamma}\leq a(k)\leq Ck^{\gamma}\\
\label{H2dak}
&\abs{\delta a(k)}\leq C'k^{\gamma-1}\\
\label{H2ddak}
&\abs{\delta^2a(k)}\leq C''k^{\gamma-2}
\end{align}
\end{subequations}
for some real constants $C$, $C'$, $C''$, $c>0$. Here
\[  
\delta a(k)\coloneqq a(k+1)-a(k)\text{ and }\delta^2a(k)\coloneqq a(k+2)-2a(k+1)+a(k).
\]
\end{enumerate}
\end{assumptions*}

\begin{remark*}
In particular, (H2) is satisfied if the large $k$ behavior of $a(k)$ is given by \eqref{JCa}.
\end{remark*}

\begin{theorem}[Jaynes--Cummings type model]  \label{thm:12} 
Let $J$ be the self-adjoint operator defined in $l^2(\D{N}^*)$ by \eqref{1J} where 
\begin{enumerate}[\rm(i)]
\item
$d(k)=k+v(k)$ with $v$ real-valued, $N$-periodic, and satisfying \emph{(H1)}, \emph{i.e.}, \eqref{H1}.
\item
$a(k)\asymp k^{\gamma}$ satisfies \emph{(H2)}, \emph{i.e.}, \eqref{H2} with $0<\gamma\leq 1/2$.
\end{enumerate}
Then its $n$-th eigenvalue $\lambda_n(J)$ has the large $n$ asymptotics 
\begin{equation} \label{E}
\lambda_n(J)=n+\av{v}+a(n-1)^2-a(n)^2+\ord(n^{-\gamma/2}\ln n).
\end{equation} 
\end{theorem}  

\begin{remark*}
Let us notice that hypotheses (H2), precisely \eqref{H2ak} and \eqref{H2dak}, imply
\[
a(n-1)^2-a(n)^2=-\bigl(a(n-1)+a(n)\bigr)\delta a(n-1)=\ord(n^{2\gamma-1})=\ord(1)\text{ as }n\to\infty.
\]
For the Jaynes--Cummings model, $a(k)=a_1k^{1/2}$, so we even have $a(n-1)^2-a(n)^2=-a_1^2=\const$.   
\end{remark*}

\begin{proof}[Proof of Theorem~\ref{thm:12}$\implies$Theorem~\ref{thm:11}]
The Jaynes--Cummings model satisfies assumption (H2) with $\gamma=1/2$. It satisfies also (H1) with $N=2$, $\av{v}=0$ and $\rho_2=\abs{\rho}$. Moreover, as noted above, $a(n-1)^2-a(n)^2=-a_1^2$, thus the asymptotic formula \eqref{E} becomes \eqref{E0}.  
\end{proof}

\subsection{Plan of the paper} \label{sec:14}

In Section \ref{sec:2} we define operators $J_n$ which are easier to investigate than $J$ and such that, by Proposition \ref{prop:2}, the $n$-th eigenvalue of $J$ is well approximated by a suitable eigenvalue of $J_n$. Thus, to get Theorem~\ref{thm:12} it remains to prove the asymptotic formula for $J_n$ stated in Theorem~\ref{thm:2}. To summarize:
\[
\begin{rcases} 
\text{Prop.~\ref{prop:2}}\\
\text{Thm.~\ref{thm:2}}
\end{rcases}\!\implies\!\text{Thm.~\ref{thm:12}}\!\implies\!\text{Thm.~\ref{thm:11}}.
\]
The proof of Theorem \ref{thm:2} is completed in Section \ref{sec:11} according to the schema 
\[
\begin{rcases} 
\text{Prop.~\ref{prop:3}}\Rightarrow\text{Prop.~\ref{prop:4}}\\
\text{Prop.~\ref{prop:5} \& Prop.~\ref{prop:11}}\\
\end{rcases}\!\implies\!\text{Thm.~\ref{thm:2}}.
\]
That corresponds to the following four steps: 

\begin{stepone*}
In Section \ref{sec:3} we prove Proposition \ref{prop:3} which is Theorem \ref{thm:2} in the case without periodic modulation, i.e., when $v\equiv 0$.  
\end{stepone*}

\begin{steptwo*}
In Section \ref{sec:4} we prove Proposition \ref{prop:4} which gives some preliminary information about the spectrum of $J_n$ obtained by the min-max principle.
\end{steptwo*}

\begin{stepthree*}
In Section \ref{sec:5} we replace the operators $J_n$ by operators $L_n$ obtained by conjugation with suitable unitary operators $\eul^{\ii B_n}$. Proposition \ref{prop:5} states a trace estimate for those operators $L_n$. Its proof is given in Sections \ref{sec:6}-\ref{sec:10}. 
\end{stepthree*}

\begin{stepfour*}
In Section \ref{sec:11} we prove Proposition \ref{prop:11} 
which is the final ingredient of the proof of Theorem \ref{thm:2}. 
\end{stepfour*}

To end this section we give some details about the proof of the trace estimate of Proposition \ref{prop:5} which is the central part of our approach. We start the proof of Proposition \ref{prop:5} in Section \ref{sec:6} by proving three lemmas that allow  us to replace Proposition \ref{prop:5} by Proposition \ref{prop:6}:
\[
\begin{rcases} 
\text{Lem.~\ref{lem:61}-\ref{lem:63}}\\
\text{Prop.~\ref{prop:6}}
\end{rcases}\!\implies\!\text{Prop.~\ref{prop:5}}.
\]
In Section \ref{sec:7} we introduce a class of operators defined by Fourier transform and used in Section \ref{sec:8} to construct an approximation of $\eul^{\ii B_n}$. This construction is used in Sections \ref{sec:9} \& \ref{sec:10} to give approximations of terms figuring in Proposition \ref{prop:6} by means of oscillatory integrals. That allow us to complete the proof of Proposition \ref{prop:6} by application of the stationary phase method. 

\subsection{Notations} \label{sec:15}

Let $\C{H}$ be a Hilbert space.
\begin{enumerate}[\textbullet]
\item
$\C{B}(\C{H})$ is the algebra of bounded operators on $\C{H}$ equipped with the operator norm $\norm{\,\cdot\,}_{\C{B}(\C{H})}$,
\item
If $Q\in\C{B}(\C{H})$ we also simply denote $\norm{Q}$. Moreover, $\Re Q\coloneqq\frac{1}{2}(Q+Q^*)$ and $\Im Q\coloneqq\frac{1}{2\ii}(Q-Q^*)$.
\item
$\C{B}_1(\C{H})\subset\C{B}(\C{H})$ is the ideal of trace class operators equipped with the norm $\norm{Q}_{\C{B}_1(\C{H})}\coloneqq\tr\sqrt{Q^*Q}$. 
\end{enumerate}
Throughout the paper, we also use the following notations:
\begin{enumerate}[\textbullet]
\item
$\D{N}=\accol{0,1,\dots}$ is the set of nonnegative integers, $\D{N}^*=\accol{1,2,\dots}$ is the set of positive integers.
\item
$l^2(\D{Z})$ is the Hilbert space of square-summable complex sequences $x\colon\D{Z}\to\D{C}$ with scalar product $\scal{x,y}\coloneqq\sum_{k\in\D{Z}}\overline{x(k)}y(k)$ and norm $\norm{x}_{l^2(\D{Z})}\coloneqq\sqrt{\scal{x,x}}$.
\item
$\accol{\vece_n}_{n\in\D{Z}}$ denotes the canonical basis of $l^2(\D{Z})$, i.e., $\vece_n(j)=\delta_{j,n}$.
\item
$H(j,k)\coloneqq\scal{\vece_j,H\vece_k}$, $j,k\in\D{Z}$ denote the matrix elements of an operator $H$ acting on $l^2(\D{Z})$ and defined on its canonical basis.
\item
$l^2(\D{N}^*)$ is the Hilbert space of square-summable sequences $x\colon\D{N}^*\to\D{C}$ equipped with the scalar product $\scal{x,y}\coloneqq\sum_{k=1}^{\infty}\overline{x(k)}y(k)$ and the norm $\norm{x}_{l^2(\D{N}^*)}\coloneqq\sqrt{\scal{x,x}}$. It is identified with the closed subspace of $l^2(\D{Z})$ generated by $\accol{\vece_n}_{n\in\D{N}^*}$, i.e.\ with $\accol{x\in l^2(\D{Z}):x(k)=0\text{ for any }k\leq 0}$. 
\end{enumerate}
We also define operators acting on $l^2(\D{Z})$ or $l^2(\D{N}^*)$:
\begin{enumerate}[\textbullet]
\item
The shift $S\in\C{B}(l^2(\D{Z}))$ is defined by $(Sx)(k)=x(k-1)$, $k\in\D{Z}$. In particular $S\vece_n=\vece_{n+1}$.
\item
$\Lambda$ acts on $l^2(\D{Z})$ by $(\Lambda x)(k)=kx(k)$, $k\in\D{Z}$ for any $x$ such that $(kx(k))_{k\in\D{Z}}\in l^2(\D{Z})$.
\item
For any $b\colon\D{Z}\to\D{C}$ we define $b(\Lambda)$ by functional calculus, i.e., $b(\Lambda)$ is closed in $l^2(\D{Z})$ and such that $b(\Lambda)\vece_k=b(k)\vece_k$, $k\in\D{Z}$. 
\item
$S^+$ and $\Lambda^+$ denote the respective restrictions of $S$ and $\Lambda$ to $l^2(\D{N}^*)$.
\item
If $L$ is a self-adjoint operator which is bounded from below with compact resolvent we denote 
\[
\lambda_1(L)\leq\dots\leq\lambda_n(L)\leq\lambda_{n+1}(L)\leq\dots
\]
its eigenvalues, enumerated in non-decreasing order, counting multiplicities.
\end{enumerate}
Throughout the paper $n\in\D{N}^*$ is the large parameter involved in the asymptotics \eqref{E0} or \eqref{E}. All error estimates are considered with respect to $n\geq 1$ and some statements will be established only for $n\geq n_0$, where $n_0$ is some large enough constant.

\section{Operators $J_n$} \label{sec:2}  
\subsection{Plan of Section~\ref{sec:2}} \label{sec:21}  

In Section \ref{sec:22} we define auxiliary operators $J_n$, $n\geq 1$. In Section \ref{sec:23} we state Theorem \ref{thm:2} which gives the asymptotic formula for the $n$th eigenvalue of $J_n$. We finally sketch a proof of Theorem \ref{thm:12} based on Theorem \ref{thm:2} and Proposition \ref{prop:2}. 

The operators $J_n$ act on $l^2(\D{Z})$ by Jacobi matrices with entries $\accol{d_n(k)}_{k\in\D{Z}}$, $\accol{a_n(k)}_{k\in\D{Z}}$ that are obtained from $\accol{d(k)}_{k=1}^{\infty}$, $\accol{a(k)}_{k=1}^{\infty}$ by cut-offs and linearizations, see \eqref{adn}.

\subsection{Definition of $\BS{J_n}$} \label{sec:22} 

It depends on the choice of a cut-off function $\theta_0\in\class^{\infty}(\D{R})$ such that 
\begin{subequations} \label{222}
\begin{equation} \label{220}
\begin{cases} 
\theta_0(t)=1&\text{if }\abs{t}\leq\frac{1}{6}\\ 
\theta_0(t)=0&\text{if }\abs{t}\geq\frac{1}{5}\\ 
0\leq\theta_0(t)\leq 1&\text{otherwise}.  
\end{cases}  
\end{equation}
From now on we fix such a cut-off function. Then, for $\tau>0$ we denote 
\begin{equation} \label{223}
\theta_{\tau,n}(s)\coloneqq\theta_0\!\left(\frac{s-n}\tau\right)
\end{equation}
\end{subequations} 
and define $d_n,\,a_n\colon\D{Z}\to\D{R}$ by 
\begin{subequations} \label{adn}
\begin{align}   \label{dn}
d_n(k)&\coloneqq k+v(k)\theta_{n,n}(k)^2,\\
\label{an}
a_n(k)&\coloneqq\left(a(n)+(k-n)\delta a(n)\right)\theta_{2n,n}(k). 
\end{align}
Let us notice that $d_n(n)=d(n)$, $a_n(n)=a(n)$, and
\begin{align}
\label{dn+}
d_n(k)&=
\begin{cases}
d(k)&\text{if }\abs{k-n}\leq\frac{n}{6}\\
k&\text{if }\abs{k-n}\geq\frac{n}{5}\,.
\end{cases}\\
\label{an+}
a_n(k)&=
\begin{cases}
a(n)+(k-n)\delta a(n)&\text{if }\abs{k-n}\leq\frac{n}{3}\\
0&\text{if }\abs{k-n}\geq\frac{2n}{5}\,,
\end{cases}
\end{align}
\end{subequations}
These modifications allow important simplifications. They ensure the large $n$ estimates \eqref{21de}, i.e.
\[
\sup_{k\in\D{Z}}\,\abs{\delta^ma_n(k)}=\ord(n^{\gamma-m}),\quad m=0,1,2
\]
which are useful to control errors with respect to the large parameter $n$. Moreover, the replacement of $a(k)$ by its linearization at $n$ for $k$ close to $n$ allows a very simple composition formula in Lemma~\ref{lem:74}, which is essential in the analysis developed in Sections~\ref{sec:8}-\ref{sec:10}.

With these $d_n(k)$'s and $a_n(k)$'s we consider the self-adjoint operator $J_n$ acting on $l^2(\D{Z})$ by 
\begin{equation} \label{Jn+}
(J_nx)(k)=d_n(k)x(k)+a_n(k)x(k+1)+a_n(k-1)x(k-1),
\end{equation} 
for $x$ such that $(kx(k))_{k\in\D{Z}}\in l^2(\D{Z})$. Its matrix in the canonical basis $(\vece_k)_{k\in\D{Z}}$ is of the form
\[
J_n=
\left(
\begin{array}{c@{}c}
\begin{array}{c@{}rrc|}
\ddots\,&\vdots&\vdots&\vdots\\
\cdots\,&-2&0&0\\
\cdots\,&0&-1&0\\
\cdots\,&0&0&0\\
\hline
\end{array}&
\begin{array}{ccc}
&&\\
&&\\
&\BS{0}&\\
&&\\[2mm]
\hline
\end{array}\\
\begin{array}{cccc}
&&&\\
&&\BS{0}&\\
&&&
\end{array}
&\begin{array}{|ccc}
&&\\
&J_n^+&\\
&&
\end{array}
\end{array}
\!\!\right)
\]
where the blocks $\BS{0}$ are identically zero and where the block 
\[
J_n^+=
\begin{pmatrix}  
d_n(1) & a_n(1) &0& 0 &\dots\\ 
a_n(1) & d_n(2) & a_n(2) &0& \dots\\ 
0 & a_n(2) & d_n(3) & a_n(3) &\dots\\ 
0 &0& a_n(3) & d_n(4) &\dots\\  
\vdots &\vdots&\vdots&\vdots&\ddots 
\end{pmatrix} 
\]   
is its restriction to $l^2(\D{N}^*)$. The spectrum of $J_n$ is clearly   
\[
\spec(J_n)=\spec(J_n^+)\cup\accol{k\in\D{Z}:k\leq 0}. 
\] 
Further on, we write $\spec(J_n)=\accol{\lambda_k(J_n)}_{k\in\D{Z}}$ with 
\[ 
\lambda_k(J_n)=
\begin{cases} 
\lambda_k(J_n^+)&\text{if }k\geq 1,\\ 
k&\text{if }k\leq 0,
\end{cases}
\] 
where $\lambda_1(J_n^+)\leq\dots\leq\lambda_k(J_n^+)\leq\lambda_{k+1}(J_n^+)\leq\dots$ denote the eigenvalues of $J_n^+$, enumerated in non-decreasing order, counting multiplicities. 

\subsection{Asymptotic behavior of $\BS{\lambda_n(J_n)}$}\label{sec:23} 
  
\begin{theorem} \label{thm:2} 
Let $(d(k))_{k\in\D{Z}}$ and $(a(k))_{k\in\D{Z}}$ be as in Theorem \ref{thm:12} with $\av{v}=0$, and $J_n$, $\accol{\lambda_k(J_n)}_{k\in\D{Z}}$ as above. Then one has the large $n$ estimate
\begin{subequations}
\begin{align} \label{En} 
\lambda_n(J_n)&=l(n)+\ord(n^{-\gamma/2}\ln n)\\
\label{lnk}
l(n)&\coloneqq n+a_n(n-1)^2-a_n(n)^2. 
\end{align}
\end{subequations}
\end{theorem}
  
\begin{proof} 
See Section \ref{sec:11}. 
\end{proof} 

\begin{proof}[Proof of Theorem~\ref{thm:2}$\implies$Theorem~\ref{thm:12}] 

(i) We have $\lambda_n(J-\scal{v})=\lambda_n(J)-\scal{v}$. Thus, to prove Theorem~\ref{thm:12} we can assume $\av{v}=0$. 

(ii) Proposition~\ref{prop:2} states the estimate
\[
\lambda_n(J)=\lambda_n(J_n)+\ord(n^{3\gamma-2}).
\]
In other words the left-hand sides of \eqref{E} and \eqref{En}, i.e.\ $\lambda_n(J)$ and $\lambda_n(J_n)$ are the same modulo $\ord(n^{3\gamma-2})$, a fortiori modulo $\ord(n^{-\gamma/2}\ln n)$ since $\gamma\leq 1/2$ implies $3\gamma-2<-\gamma/2$.

(iii) Lemma \ref{lem:23} for $k=n-1$ gives the large $n$ estimates $a(n-1)-a(n)=\ord(n^{\gamma-1})$ and $a(n-1)-a_n(n-1)=\ord(n^{\gamma-2})$, by \eqref{eq:216} and \eqref{eq:217}, respectively. Since $a(n)=\ord(n^{\gamma})$ we have the same estimate for $a(n-1)$ and $a_n(n-1)$. Thus, $a(n-1)^2-a_n(n-1)^2=\ord(n^{2\gamma-2})$ and 
\begin{equation}  \label{eq:121}
a_n(n-1)^2-a_n(n)^2=a(n-1)^2-a(n)^2+\ord(n^{2\gamma-2}).
\end{equation}
Since $\gamma\leq 1/2$ implies $2\gamma-2<-\gamma/2$ this relation holds a fortiori modulo $\ord(n^{-\gamma/2}\ln n)$. That proves that the right-hand sides of \eqref{E} and \eqref{En} are the same modulo $\ord(n^{-\gamma/2}\ln n)$.

(iv) By (ii) and (iii) \eqref{En}$\implies$\eqref{E}, i.e.\ Theorem~\ref{thm:2}$\implies$Theorem~\ref{thm:12} with $\av{v}=0$.
\end{proof} 

\section{Theorem~\ref{thm:2} when $v\equiv 0$} \label{sec:3}  
\subsection{Plan of Section~\ref{sec:3}}  \label{sec:31}

The aim of this section is to prove Proposition \ref{prop:3} which says that Theorem \ref{thm:2} holds when $v\equiv 0$. Since the proof is based on the min-max principle we consider operators acting on $l^2(\D{N}^*)$. In Section \ref{sec:32} we state Proposition \ref{prop:3} and we explain the idea of the proof. In Section \ref{sec:33} we show a useful property of 
\begin{equation} \label{300}
l_n(k)\coloneqq k+a_n(k-1)^2-a_n(k)^2,\quad k\geq 1.   
\end{equation} 
Note that $l_n(n)=l(n)$ where $l(n)$ is defined by \eqref{lnk}. The proof of Proposition \ref{prop:3} is completed in  Section \ref{sec:34}.   

\subsection{Main result} \label{sec:32}

We consider the operator $J_{0,n}^+\colon\C{D}\to l^2(\D{N}^*)$ defined by 
\[
(J_{0,n}^+x)(k)=kx(k)+a_n(k)x(k+1)+a_n(k-1)x(k-1),  
\] 
where $a_n(k)$ is given by \eqref{an}. Thus $J_{0,n}^+$ coincides with $J_n^+$ if $v\equiv 0$ and Theorem~\ref{thm:2} in that case follows from Proposition~\ref{prop:3}.

\begin{proposition} \label{prop:3} 
If $l_n(k)$ is given by \eqref{300}, then 
\begin{equation} \label{E3}   
\sup_{k\in\D{N}^*}\,\abs{\lambda_k(J_{0,n}^+)-l_n(k)}=\ord(n^{3\gamma-2}).
\end{equation}   
\end{proposition} 

\begin{proof}[Sketch of proof]
A complete proof is given in Section \ref{sec:34}. The proof is similar to the first step of the successive diagonalization method \cite{BZ1}. We observe that
\begin{equation} \label{216}
J_{0,n}^+=\Lambda^++A_n^+
\end{equation}
where $A_n^+$ is the finite rank operator defined by the matrix
\begin{equation} \label{2A}
A_n^+=
\begin{pmatrix}  
0 & a_n(1) &0& 0 &\dots\\ 
a_n(1) & 0 & a_n(2) &0& \dots\\ 
0 & a_n(2) & 0 & a_n(3) &\dots\\ 
0 &0& a_n(3) & 0 &\dots\\  
\vdots &\vdots&\vdots&\vdots&\ddots 
\end{pmatrix} 
\end{equation}  
In Section \ref{sec:34} we define self-adjoint operators $B_n^+$ such that the difference  
\begin{equation} \label{327}
R_n^+\coloneqq\eul^{\ii B_n^+}J_{0,n}^+\eul^{-\ii B_n^+}-l_n(\Lambda^+) 
\end{equation}    
can be estimated by 
\begin{equation} \label{327'}
\norm{R_n^+}_{\C{B}(l^2(\D{N}^*))}=\ord(n^{3\gamma-2}). 
\end{equation} 
By the min-max principle, 
\[
\abs{\lambda_k(l_n(\Lambda^+)+R_n^+)-\lambda_k(l_n(\Lambda^+))}\leq\norm{R_n^+}_{\C{B}(l^2(\D{N}^*))},
\] 
hence the estimate \eqref{E3} follows from \eqref{327'} and from the relations 
\begin{align}
&\lambda_k(l_n(\Lambda^+)+R_n^+)=\lambda_k(J_{0,n}^+)\notag\\
\label{lakln}
&\lambda_k(l_n(\Lambda^+))=l_n(k)\text{ for }n\geq n_1,  
\end{align} 
where $n_1$ is some large enough integer and $k\geq 1$. This equality \eqref{lakln} follows from $\spec(l_n(\Lambda^+))=\accol{l_n(k)}_{k=1}^{\infty}$ and from Lemma \ref{lem:32} below. By this lemma we can indeed find $n_1$ such that 
\[ 
n\geq n_1\implies l_n(k)<l_n(k+1)\text{ for all }k\in\D{N}^*.\qedhere 
\]  
\end{proof}

\subsection{The sequence $\BS{(l_n(k))_{k=1}^{\infty}}$ is increasing for large $\BS{n}$} \label{sec:33}  

\begin{lemma} \label{lem:32}   
Let $(l_n(k))_{k=1}^{\infty}$ be defined by \eqref{300}. For any $\varepsilon>0$ there exists $n(\varepsilon)$ such that  
\begin{equation} \label{E32} 
\abs{l_n(k+1)-l_n(k)-1}<\varepsilon  
\end{equation}   
holds for any $n\geq n(\varepsilon)$ and all $k\in\D{N}^*$. 
\end{lemma} 

\begin{proof} 
We write 
\begin{subequations}  \label{lna1nk}
\begin{align}  \label{lna1n}
l_n(k)&=k+a_{1,n}(k)\\
\label{a1nk} 
a_{1,n}(k)&\coloneqq a_n(k-1)^2-a_n(k)^2.
\end{align}
\end{subequations} 
Thus,
\begin{align*}
a_{1,n}(k)&=-\bigl(a_n(k-1)+a_n(k)\bigr)\,\delta a_n(k-1),\\
\delta a_{1,n}(k)&=-\bigl(\delta a_n(k-1)+\delta a_n(k)\bigr)\,\delta a_n(k)-\bigl(a_n(k-1)+a_n(k)\bigr)\,\delta^2a_n(k-1).
\end{align*} 
Under (H2) Lemma~\ref{lem:21} states estimates \eqref{21de}, i.e.\ $\sup_{k\geq 1}\abs{\delta^ma_n(k)}\leq\tilde Cn^{\gamma-m}$ for $m=0,1,2$. It follows that  
\[
\sup_{k\geq 1}\,\abs{a_{1,n}(k)}\leq Cn^{2\gamma-1}\ \text{ and }\ \sup_{k\geq 1}\,\abs{\delta a_{1,n}(k)}\leq C_0n^{2\gamma-2}
\] 
for some constants $C$, $C_0>0$. Therefore,  
\begin{equation} \label{lnk1} 
\sup_{k\geq 1}\,\abs{l_n(k+1)-l_n(k)-1}\leq C_0n^{2\gamma-2}.  
\end{equation}
We complete the proof choosing $n(\varepsilon)$ such that $C_0n(\varepsilon)^{2\gamma-2}<\varepsilon$.
\end{proof} 

\subsection{Proof of Proposition \ref{prop:3}}  \label{sec:34} 

We consider the operators  
\begin{equation} \label{2B}
B_n^+\coloneqq\begin{pmatrix}  
0 & \ii a_n(1) &0& 0 &\dots\\ 
-\ii a_n(1) & 0 & \ii a_n(2) &0& \dots\\ 
0 & -\ii a_n(2) & 0 & \ii a_n(3) &\dots\\ 
0 &0& -\ii a_n(3) & 0 &\dots\\  
\vdots &\vdots&\vdots&\vdots&\ddots 
\end{pmatrix} 
\end{equation} 

\begin{stepone*}\sl
We claim that $\ii(\Lambda^+B_n^+-B_n^+\Lambda^+)=\croch{\ii\Lambda^+,\,B_n^+}=A_n^+$ 
where $A_n^+$ is given by \eqref{2A}.
\end{stepone*}

\begin{proof}
Writing $A_n^+=2\Re(S^+a_n(\Lambda^+))$ and $B_n^+=2\Im(S^+a_n(\Lambda^+))$ we get 
\[
\croch{\ii B_n^+,\Lambda^+}=2\Re\croch{S^+a_n(\Lambda^+),\Lambda^+}=2\Re\croch{S^+,\Lambda^+}a_n(\Lambda^+).  
\] 
Now it suffices to observe that $\croch{S^+,\Lambda^+}=-S^+$. 
\end{proof}

\begin{steptwo*}\sl
We claim that $\croch{\ii B_n^+,\,A_n^+}=2\,a_{1,n}(\Lambda^+)$ where $a_{1,n}$ is as in \eqref{a1nk}.
\end{steptwo*}

\begin{proof}
We observe that $\croch{\ii B_n^+,\,A_n^+}=2\Re\croch{S^+a_n(\Lambda^+),A_n^+}$ and  
\begin{align*}
\croch{S^+a_n(\Lambda^+),A_n^+}&=\croch{S^+a_n(\Lambda^+),\,S^+a_n(\Lambda^+)+a_n(\Lambda^+)(S^+)^*}\\ 
&= S^+a_n(\Lambda^+)^2(S^+)^*-a_n(\Lambda^+)(S^+)^*S^+a_n(\Lambda^+)\\
&=a_n(\Lambda^+-I)^2-a_n(\Lambda^+)^2.
\end{align*} 
\renewcommand{\qed}{} 
\end{proof}

\begin{stepthree*}\sl
As explained at the end of Section \ref{sec:32} it remains to prove \eqref{327'}. 
\end{stepthree*}

\begin{proof}
Recall \eqref{327'} is the estimate $\norm{R_n^+}_{\C{B}(l^2(\D{N}^*))}=\ord(n^{3\gamma-2})$ where, according to \eqref{327} and \eqref{216},
\[
R_n^+=\eul^{\ii B_n^+}J_{0,n}^+\eul^{-\ii B_n^+}-l_n(\Lambda^+)\text{ with }J_{0,n}^+=\Lambda^++A_n^+. 
\]
In order to prove \eqref{327'} we denote $J_{0,n}(t)\coloneqq\Lambda^++tA_n^+$ for $t\in\D{R}$ and introduce  
\[ 
G_n(t)\coloneqq\eul^{\ii tB_n^+}J_{0,n}(t)\eul^{-\ii tB_n^+}. 
\]
Then we get 
\[
\partial_tG_n(t)=\eul^{\ii tB_n^+}\bigl(\F{d}_{B_n^+}J_{0,n}(t)\,\bigr) \eul^{-\ii tB_n^+},
\]
Using Steps 1 and 2 for the last equality below we find that
\begin{align*}
\F{d}_{B_n^+}J_{0,n}(t)&\coloneqq\partial_tJ_{0,n}(t)\,+\croch{\ii B_n^+,J_{0,n}(t)}\\
&=A_n^++\croch{\ii B_n^+,\,\Lambda^+}+t\,\croch{\ii B_n^+,\,A_n^+}\\
&=2t\,a_{1,n}(\Lambda^+). 
\end{align*}
Hence we can write
\[   
\partial_tG_n(t)=2t\,\bigl(a_{1,n}(\Lambda^+)+R_n(t)\bigr)
\] 
with  
\begin{align*} 
R_n(t)&\coloneqq\eul^{\ii tB_n^+}a_{1,n}(\Lambda^+)\eul^{-\ii tB_n^+}-a_{1,n}(\Lambda^+)\\ 
&=\int_0^1\eul^{\ii stB_n^+}[\ii B_n^+,\,a_{1,n}(\Lambda^+)]\eul^{-\ii st B_n^+}\dd s.
\end{align*} 
Since $\norm{\croch{S^+,a_{1,n}(\Lambda^+)}}_{\C{B}(l^2(\D{N}^*))}=\norm{(\delta a_{1,n})(\Lambda^+)]}_{\C{B}(l^2(\D{N}^*))}=\ord(n^{2\gamma-2})$, it is clear that 
\begin{equation} \label{E22''}
\norm{R_n(t)}_{\C{B}(l^2(\D{N}^*))}\leq\norm{\croch{B_n^+,\,a_{1,n}(\Lambda^+)}}_{\C{B}(l^2(\D{N}^*))}=\ord(n^{3\gamma-2}). 
\end{equation}
Using \eqref{lna1n}, i.e.\ $l_n(k)=k+a_{1,n}(k)$, we find that $R_n^+\coloneqq\eul^{\ii B_n^+}J_{0,n}^+\eul^{-\ii tB_n^+}-l_n(\Lambda^+)$ can be written
\begin{align*}
R_n^+&=\eul^{\ii B_n^+}J_{0,n}^+\eul^{-\ii B_n^+}-\Lambda^+-a_{1,n}(\Lambda^+)\\
&=G_n(1)-G_n(0)-a_{1,n}(\Lambda^+)\\ 
&=\int_0^12t\,\bigl(a_{1,n}(\Lambda^+)+R_n(t)\bigr)\dd t-a_{1,n}(\Lambda^+)\\
&=\int_0^12tR_n(t)\dd t. 
\end{align*}
Hence \eqref{327'} follows from \eqref{E22''}. 
\renewcommand{\qed}{} 
\end{proof}

\section{Properties of the spectrum of $J_n$} \label{sec:4} 
\subsection{Plan of Section \ref{sec:4}} \label{sec:41} 

The purpose of this section is to prove two properties of the spectrum of $J_n$ given in Proposition \ref{prop:4} which is stated in Section~\ref{sec:42}. The proof of the first property is given in Section \ref{sec:43} and the proof of the second one is given in Section \ref{sec:45}.

\subsection{Main result} \label{sec:42} 

\begin{proposition}[estimates for eigenvalues of $J_n$] \label{prop:4} 
Assume that the operators $J_n$ are as in Theorem \ref{thm:2} and $\rho_N$ is given by \eqref{110} with $\av{v}=0$, \emph{i.e.}
\begin{equation}  \label{42} 
\rho_N=\max_{1\leq k\leq N}\,\abs{v(k)}.
\end{equation} 
\begin{enumerate}[\rm(a)] 
\item
If $C_0$ is large enough, then 
\begin{equation} \label{E4a}   
\sup_{k\geq 1}\abs{\lambda_k(J_n)-l_n(k)}\leq\rho_N+C_0n^{3\gamma-2}.
\end{equation}   
\item
If $n_1$ is large enough, then for $n\geq n_1$ one has
\begin{equation}  \label{E4b}
\sup_{k\geq 1}\abs{\lambda_{k+N}(J_n)-\lambda_k(J_n)-N}=\ord(n^{\gamma-1}).
\end{equation} 
\end{enumerate}
\end{proposition} 

\begin{remarks*} 
(i) 
We will deduce \eqref{E4a} from Proposition \ref{prop:3} by application of the min-max principle. 

(ii)  
Let $k\geq 1$. For $C>0$ we define the intervals 
\begin{equation} \label{De}
\Delta_{k,n}^C\coloneqq[l_n(k)-\rho_N-Cn^{\gamma-1},\,l_n(k)+\rho_N+Cn^{\gamma-1}].
\end{equation}  
Since the hypothesis $\rho_N<\frac{1}{2}$ allows us to use \eqref{E32} from Lemma \ref{lem:32} with $0<\varepsilon<\frac{1}{2}-\rho_N$ we find   
\[
l_n(k+1)-l_n(k)\geq 2\rho_N+\varepsilon
\] 
for $n\geq n(\varepsilon)$. Therefore choosing $n_C$ large enough to ensure $2Cn_C^{\gamma-1}<\varepsilon$ we obtain  
\[
n\geq n_C\implies\Delta_{k,n}^C\cap\Delta_{k+1,n}^C=\varnothing.  
\]  
Since $3\gamma-2\leq\gamma-1$, \eqref{E4a} implies that there exists $C_0>0$ such that $\lambda_k(J_n)\in\Delta_{k,n}^{C_0}$, hence
\[
n\geq n_{C_0}\implies\spec(J_n)\cap\Delta_{k,n}^{C_0}=\accol{\lambda_k(J_n)}.
\] 
This localisation of $\lambda_k(J_n)$ is crucial for the proof of \eqref{E4b} given in Section~\ref{sec:45}.
\end{remarks*} 

\subsection{Proof of Proposition \ref{prop:4} (a)}  \label{sec:43} 

Let us note that $\lambda_k(J_n)=\lambda_k(J_n^+)$ for $k\geq 1$. Moreover, $J_n^+=J_{0,n}^++v_n(\Lambda^+)$ with
\begin{equation} \label{vn}  
v_n(k)\coloneqq v(k)\theta_{n,n}^2(k).
\end{equation}  
Then, by the min-max principle and \eqref{42} we get 
\[
\abs{\lambda_k(J_n^+)-\lambda_k(J_{0,n}^+)}\leq\norm{v_n(\Lambda^+)}_{\C{B}(l^2(\D{N}^*))}\leq\rho_N,
\]
and \eqref{E4a} follows using estimate \eqref{E3} from Proposition~\ref{prop:3}.

\subsection{Proof of Proposition \ref{prop:4} (b)}  \label{sec:45} 

\begin{stepone*}\sl
Let $C'$ be large enough. Then there is $n_{C'}$ such that
\[ 
n\geq n_{C'}\implies\lambda_k(J_n)+N\in\Delta_{k+N,n}^{C'}.
\] 
\end{stepone*}

\begin{proof}
By definition \eqref{De} of $\Delta_{j,n}^{C}$ it suffices to show the estimate
\begin{equation} \label{716} 
\abs{\lambda_k(J_n)+N-l_n(k+N)}\leq C'n^{\gamma-1}+\rho_N.
\end{equation} 
The left-hand side of \eqref{716} can be estimated by 
\begin{equation} \label{717} 
\abs{\lambda_k(J_n)-l_n(k)}+\abs{l_n(k)+N-l_n(k+N)}. 
\end{equation} 
It remains to observe that the first term of \eqref{717} can be estimated by $\rho_N+C_0n^{3\gamma-2}$ due to Proposition \ref{prop:4}~(a) and the second term of \eqref{717} can be estimated by $C_0' n^{2\gamma-2}$ due to \eqref{lnk1}. 
\renewcommand{\qed}{}
\end{proof}

\begin{steptwo*}\sl
We claim that 
\[ 
\norm{S^{-N}J_n S^N-J_n-N}\leq C''n^{\gamma-1}.
\]
\end{steptwo*}
  
\begin{proof}
Using $S^{-N}a_n(\Lambda)S^N=a_n(\Lambda+N)$ we get 
\[
\norm{S^{-N}a_n(\Lambda)S^N-a_n(\Lambda)}=\ord(n^{\gamma-1})
\]    
from $\abs{a_n(\lambda+N)-a_n(\lambda)}\leq Cn^{\gamma-1}$, and using $S^{-N}v(\Lambda)S^N=v(\Lambda)$ we get 
\[
\norm{S^{-N}v_n(\Lambda)S^N-v_n(\Lambda)}=\ord(n^{-1}).
\]  
\renewcommand{\qed}{}
\end{proof}

\begin{stepthree*}\sl
We finally check that
\begin{equation} \label{7175} 
\lambda_{k+N}(J_n)\in\Delta_{k,n}^\circ\coloneqq\bigl\lbrack\lambda_k(J_n)+N-C''n^{\gamma-1},\,\lambda_k(J_n)+N+C''n^{\gamma-1}\bigr\rbrack
\end{equation}
holds for $n\geq n_0$ if $n_0$ and $C''$ are large enough.  
\end{stepthree*}
 
\begin{proof}
Let $C''$ be as in Step 2. If $R_n\coloneqq S^{-N}J_n S^N-J_n-N$, then by the min-max principle,
\[ 
\spec(J_n)=\spec(S^{-N}J_nS^N)=\spec(J_n+N+R_n)\subset\bigcup_{j\in\D{Z}}\Delta_{j,n}^{\circ}. 
\] 
Let $C'$ be as in Step 1 and $\hat C>C'+C''$. Then
\begin{equation} \label{718} 
n\geq n_{\hat C}\implies\Delta_{k,n}^\circ\subset\Delta_{k+N,n}^{\hat C}
\end{equation}
if $n_{\hat C}$ is large enough. If moreover $\hat C\geq C_0$ with $C_0$ as in \eqref{E4a}, then  
\begin{equation} \label{719} 
n\geq n_{\hat C}\implies\left(\lambda_j(J_n)\in\Delta_{j,n}^{\hat C}\quad\&\quad\Delta_{k,n}^{\hat C}\cap\Delta_{k+1,n}^{\hat C}=\varnothing\right)
\end{equation}
for every $j,k\in\D{Z}$. Using \eqref{719} with $j=k+N$ and \eqref{718} we obtain \eqref{7175} writing
\[
\lambda_{k+N}(J_n)\in\,\spec(J_n) \cap \Delta_{k+N,n}^{\hat C}\subset\bigcup_{j\in\D{Z}}\,\Delta_{j,n}^\circ\cap \Delta_{k+N,n}^{\hat C}=\Delta_{k,n}^\circ\text{ for }n\geq n_{\hat C}. 
\]
\renewcommand{\qed}{}
\end{proof}

\section{Operators $L_n$} \label{sec:5} 
\subsection{Plan of Section~\ref{sec:5}} \label{sec:51} 

In Section~\ref{sec:3} we obtained asymptotic estimates of eigenvalues when $v\equiv 0$ by reducing the off-diagonal entries through suitable conjugations with $\eul^{\ii B_n^+}$. We follow the same method to manage the general case. 

In Section \ref{sec:52} we use $\eul^{\ii B_n^+}$ from Section \ref{sec:34} to replace $J_n$ by $L_n$. In Section~\ref{sec:53} we state properties of the spectrum of $L_n$. In Section~\ref{sec:54} we state Proposition~\ref{prop:5} which is the most important ingredient of the proof of Theorem~\ref{thm:2}. The proof of Proposition~\ref{prop:5} begins in Section~\ref{sec:6} and ends in Section~\ref{sec:10}.    

\subsection{Definition of $\BS{L_n}$} \label{sec:52}

We define the operator $L_n$ acting on $l^2(\D{Z})$ by
\[
L_n\coloneqq l_n(\Lambda)+\tilde V_n
\]
where
\begin{equation} \label{lnj}
l_n(k)=
\begin{cases} 
k+a_n(k-1)^2-a_n(k)^2&\text{if }k\geq 1,\\
k&\text{if }k\leq 0
\end{cases}
\end{equation} 
with $a_n(k)$ defined in \eqref{an} and
\begin{align} \label{Vn'}
\tilde V_n&\coloneqq\eul^{\ii B_n}v_n(\Lambda)\eul^{-\ii B_n}\\
\label{Bn''}
B_n&\coloneqq\ii\left(a_n(\Lambda)S^{-1}-Sa_n(\Lambda)\right)=\begin{pmatrix}\BS{0}&\BS{0}\\\BS{0}&B_n^+\end{pmatrix}.
\end{align}
The restriction $B_n^+$ to $l^2(\D{N}^*)$ was already defined by \eqref{2B} in Section~\ref{sec:34}. Similarly,
\[
L_n=
\left(
\begin{array}{c@{}c}
\begin{array}{c@{}rrc|}
\ddots\,&\vdots&\vdots&\vdots\\
\cdots\,&-2&0&0\\
\cdots\,&0&-1&0\\
\cdots\,&0&0&0\\
\hline
\end{array}&
\begin{array}{ccc}
&&\\
&&\\
&\BS{0}&\\
&&\\[2mm]
\hline
\end{array}\\
\begin{array}{cccc}
&&&\\
&&\BS{0}&\\
&&&
\end{array}
&\begin{array}{|ccc}
&&\\
&L_n^+&\\
&&
\end{array}
\end{array}
\!\!\right)
\] 
The restriction $L_n^+$ to $l^2(\D{N}^*)$ is given by 
\begin{subequations} \label{lnplustvn}
\begin{align} \label{lnplus}
L_n^+&\coloneqq l_n(\Lambda^+)+\tilde V_n^+ \\
\label{tvn}
\tilde V_n^+&\coloneqq\eul^{\ii B_n^+}v_n(\Lambda^+)\eul^{-\ii B_n^+}. 
\end{align}
\end{subequations}
The spectrum of $L_n$ is clearly   
\[  
\spec(L_n)=\spec(L_n^+)\cup\accol{k\in\D{Z}:k\leq 0}. 
\]   
Further on, we write $\spec(L_n)=\accol{\lambda_k(L_n)}_{k\in\D{Z}}$ with 
\[ 
\lambda_k(L_n)=
\begin{cases} 
\lambda_k(L_n^+)&\text{if }k\geq 1,\\ 
k&\text{if }k\leq 0.
\end{cases}
\] 

\subsection{Properties of the spectrum of $\BS{L_n}$} \label{sec:53}

\begin{proposition}[estimates for $\lambda_k(L_n)$] \label{cor5} 
Let $L_n$ and $\accol{\lambda_k(L_n)}_{k\in\D{Z}}$ be as in Section~\ref{sec:52}.
\begin{enumerate}[\rm(a)] 
\item
Estimate \eqref{En} from Theorem \ref{thm:2} is equivalent to  
\begin{equation} \label{cor5a} 
\lambda_n(L_n)=l_n(n)+\ord(n^{-\gamma/2}\ln n). 
\end{equation} 
\item
If $C$ is large enough, then 
\begin{equation} \label{cor5b}   
\sup_{k\geq 1}\abs{\lambda_k(L_n)-l_n(k)}\leq\rho_N+Cn^{3\gamma-2}.  
\end{equation}
\item
If $n_1$ is large enough, then for $n\geq n_1$ one has
\begin{equation}  \label{cor5c}
\sup_{k\geq 1}\abs{\lambda_{k+N}(L_n)-\lambda_k(L_n)-N}=\ord(n^{\gamma-1}). 
\end{equation} 
\end{enumerate}
\end{proposition}  

\begin{proof}
This proposition translates estimates for $J_n$ into estimates for $L_n$ through the key estimate
\begin{equation} \label{E5}
\sup_{k\geq 1}\abs{\lambda_k(J_n)-\lambda_k(L_n)}=\ord(n^{3\gamma-2}).
\end{equation} 
Apply \eqref{E5} to translate each of the three estimates \eqref{En}, \eqref{E4a}, and \eqref{E4b}, the first one from Theorem \ref{thm:2} and the other two from Proposition \ref{prop:4}. Statements (b) and (c) are thus corollaries of Proposition \ref{prop:4}. 

It remains to prove \eqref{E5}. Let $k\geq 1$. We have $\lambda_k(L_n)=\lambda_k(L_n^+)$ and $\lambda_k(J_n)=\lambda_k(J_n^+)=\lambda_k(\eul^{\ii B_n^+}J_n^+\eul^{-\ii B_n^+})$. Moreover, using $J_n^+=J_{0,n}^++v_n(\Lambda^+)$ together with \eqref{327}, \eqref{tvn}, and \eqref{lnplus} that define $R_n^+$, $\tilde V_n^+$, and $L_n^+$ we find 
\[
\eul^{\ii B_n^+}J_n^+\eul^{-\ii B_n^+}=l_n(\Lambda^+)+R_n^++\tilde V_n^+=L_n^++R_n^+.
\]
Finally, by the min-max principle and estimate \eqref{327'} of $\norm{R_n^+}_{\C{B}(l^2(\D{N}^*))}$:
\[
\sup_{k\geq 1}\abs{\lambda_k(J_n)-\lambda_k(L_n)}=\sup_{k\geq 1}\abs{\lambda_k(L_n^++R_n^+)-\lambda_k(L_n^+)}\leq\norm{R_n^+}_{\C{B}(l^2(\D{N}^*))}=\ord(n^{3\gamma-2}).\qedhere
\] 
\end{proof}

\subsection{A trace estimate} \label{sec:54}   

We denote $L_{0,n}\coloneqq l_n(\Lambda)$ and we want to compare the spectrum of 
\begin{equation} \label{Ln=}
L_n\coloneqq L_{0,n}+\tilde V_n
\end{equation} 
with that of $L_{0,n}$ which is $\accol{l_n(k)}_{k\in\D{Z}}$ for $n\geq n_0$. For this purpose we consider the expression 
\begin{subequations} \label{G0ntr0}
\begin{equation} \label{G0n}
\C{G}_n^0\coloneqq\sum_{k\in\D{Z}}\bigl(\chi(\lambda_k(L_n)-l_n(n))-\chi(l_n(k)-l_n(n))\bigr), 
\end{equation} 
with $\chi\in\C{S}(\D{R})$, where $\C{S}(\D{R})$ denotes the Schwartz class of rapidly decreasing functions on $\D{R}$. Let us observe that $\C{G}_n^0$ can be written as a trace:
\begin{equation} \label{tr0'}
\C{G}_n^0=\tr\bigl(\chi(L_n-l(n))-\chi(L_{0,n}-l(n))\bigr), 
\end{equation}
\end{subequations} 
where, as already noted, $l(n)=l_n(n)=n+a_n(n-1)^2-a_n(n)^2$. 

\begin{proposition}[trace estimate] \label{prop:5}   
Let $\chi\in\C{S}(\D{R})$ be such that its Fourier transform 
\begin{equation} \label{F} 
\hat\chi(t)\coloneqq\int_{-\infty}^{\infty}\chi(\lambda)\eul^{-\ii t\lambda}\frac{\dd\lambda}{2\pi}
\end{equation} 
has compact support. If $\C{G}_n^0$ is given by \eqref{G0ntr0}, then one has the large $n$ estimate  
\begin{equation} \label{tr0}
\C{G}_n^0=\ord(n^{-\gamma/2}\ln n). 
\end{equation} 
\end{proposition} 

\begin{proof} 
See Section \ref{sec:10}. 
\end{proof}  

\section{Reformulation of Proposition \ref{prop:5}} \label{sec:6}
\subsection{Plan of Section \ref{sec:6}} \label{sec:61} 
Let $\C{G}_n^0$ be given by~\eqref{G0n}. In this section we show that the trace estimate $\C{G}_n^0=\ord(n^{-\gamma/2}\ln n)$ in Proposition~\ref{prop:5} is a consequence of Proposition~\ref{prop:6}, whose proof will be given in Sections~\ref{sec:7}-\ref{sec:10}. 

We explain now the idea of obtaining the trace estimate \eqref{tr0} from Proposition~\ref{prop:6}. To begin with, we observe that the trace formulation \eqref{tr0'} express $\C{G}_n^0$ as a function of $L_n$ and $L_{0,n}$, and such a function can be expressed by means of the evolutions $\eul^{\ii tL_n}$ and $\eul^{\ii tL_{0,n}}$ ($t\in\D{R}$) via the standard representation formula based on the Fourier transform. Next we write   
\[
\eul^{\ii tL_n}-\eul^{\ii tL_{0,n}}=\eul^{\ii tL_{0,n}}(U_n(t)-I)
\]
where   
\[
U_n(t)\coloneqq\eul^{-\ii tL_{0,n}}\eul^{\ii tL_n},\quad t\in\D{R}
\]
and use the Neumann series to express $U_n(t)-I$. Then to obtain information about traces it suffices to consider estimates of the diagonal entries for every term in the Neumann series. In Proposition~\ref{prop:6} we state estimates which ensure the estimate $\C{G}_n^0=\ord(n^{-\gamma/2}\ln n)$. The same approach was used in our previous paper \cite{BZ4} where we considered weaker remainder estimates and the stronger assumption $\gamma<1/2$. In the framework of \cite{BZ4} we show estimates similar to the estimates of Proposition~\ref{prop:6} in a very short way as all involved operators are functions of $S$ and their matrix elements can be directly expressed by means of oscillatory integrals.
 
In Section~\ref{sec:62} we prove Lemma~\ref{lem:61} which says that modulo $\ord(n^{-\gamma})$ we can modify the trace~\eqref{G0n} by using an auxiliary cut-off. In Section~\ref{sec:63} we prove Lemma~\ref{lem:62} which shows that the trace estimate~\eqref{tr0} follows from condition~\eqref{es0} on the evolution $U_n(t)$. In Section~\ref{sec:64} we prove Lemma~\ref{lem:63} which shows that this condition results from estimates \eqref{es} on the coefficients of the Neumann series for $U_n(t)$. In Section~\ref{sec:65} we state Proposition~\ref{prop:6} which shows that these estimates are valid. 

\subsection{An auxiliary cut-off} \label{sec:62} 

The aim of this section is to check that the trace estimate \eqref{tr0} in Proposition~\ref{prop:5} is equivalent to the estimate 
\begin{equation} \label{tr}
\C{G}_n=\ord(n^{-\gamma/2}\ln n)
\end{equation} 
where 
\begin{equation} \label{Gn}
\C{G}_n\coloneqq\tr\Bigl(\theta_{n^{\gamma},n}(L_{0,n})\bigl(\chi(L_n-l(n))-\chi(L_{0,n}-l(n))\bigr)\Bigr).
\end{equation} 
The cut-off $\theta_{n^{\gamma},n}$ is defined by \eqref{222}. 

\begin{lemma} \label{lem:61} 
If $\C{G}_n^0$ is given by \eqref{G0n} and $\C{G}_n$ by \eqref{Gn}, then 
\begin{equation} \label{Gn-Gn0}
\C{G}_n-\C{G}_n^0=\ord(n^{-\gamma}).
\end{equation}
\end{lemma} 

\begin{proof} 
First of all we observe that there is a constant $C>0$ such that 
\[
\norm{(I+(L_{0,n}-l(n))^2)^{-1}}_{{\C{B}}_1(l^2(\D{Z}))}=\sum_{j\in\D{Z}}\frac{1}{1+(l_n(j)-l(n))^2}\leq C
\] 
and a similar estimate holds for $L_n$:  
\begin{equation} \label{p331}
\norm{(I+(L_n-l(n))^2)^{-1}}_{{\C{B}}_1(l^2(\D{Z}))}=\sum_{j\in\D{Z}} \frac{1}{1+(\lambda_j(L_n)-l(n))^2}\leq C.
\end{equation} 
Next we claim that for every $\mu>0$ we can estimate
\begin{equation} \label{p332}
\norm{(I-\theta_{n^{\gamma},n}(L_n))\chi(L_n-l(n))}_{{\C{B}}_1(l^2(\D{Z}))}=\ord(n^{-\mu}).
\end{equation} 
Indeed, if $\chi_0(s)\coloneqq (1+s^2)\chi(s)$ then for every $\mu >0$ we have  
\[ 
\sup_{s\in\D{R}}\abs{(1-\theta_{n^{\gamma},n}(s))\chi_0(s-l(n))}=\ord(n^{-\mu}). 
\]
Hence,  
\begin{equation} \label{p332'}
\norm{(I-\theta_{n^{\gamma},n}(L_n))\chi_0(L_n-l(n))}=\ord(n^{-\mu}).
\end{equation} 
Since the left-hand side of \eqref{p332} can be estimated by 
\[  
\norm{(I-\theta_{n^{\gamma},n}(L_n))\chi_0(L_n-l(n))}\times\norm{(1+(L_n-l(n))^2)^{-1}}_{{\C{B}}_1(l^2(\D{Z}))}  
\] 
we deduce \eqref{p332} from \eqref{p331} and  \eqref{p332'}. Reasoning similarly with $L_{0,n}$ instead of $L_n$ we obtain 
\begin{equation}  \label{l33b}
\norm{(I-\theta_{n^{\gamma},n}(L_{0,n}))\chi(L_{0,n}-l(n))}_{{\C{B}}_1(l^2(\D{Z}))}=\ord(n^{-\mu}).
\end{equation} 
If the operator $T$ is self-adjoint, the operator $R$ is bounded and $\theta\in C_0^{\infty}(\D{R})$, then there exists a constant $C=C(\theta)$ such that  
\begin{equation} \label{HS} 
\norm{\theta(T+R)-\theta(T)}\leq C\norm{R}. 
\end{equation}  
Thus, using \eqref{HS} with $T=n^{-\gamma}(L_{0,n}-n)$ and $R=n^{-\gamma}\tilde V_n$ we can estimate 
\[
\norm{\theta_{n^{\gamma},n}(L_n)-\theta_{n^{\gamma},n}(L_{0,n})}\leq   C_0\norm{n^{-\gamma}\tilde V_n}=\ord(n^{-\gamma})
\] 
and combining this last estimate with \eqref{p331} we obtain  
\begin{equation}  \label{l33a}
\norm{(\theta_{n^{\gamma},n}(L_n)-\theta_{n^{\gamma},n}(L_{0,n}))\chi(L_n-l(n))}_{{\C{B}}_1(l^2(\D{Z}))}=\ord(n^{-\gamma}).   
\end{equation}  
However, using \eqref{l33a} and \eqref{p332} with $\mu=\gamma$ we obtain
\begin{equation}  \label{l33c}
\norm{(I-\theta_{n^{\gamma},n}(L_{0,n}))\chi(L_n-l(n))}_{{\C{B}}_1(l^2(\D{Z}))}=\ord(n^{-\gamma}).
\end{equation} 
It is now clear that \eqref{Gn-Gn0} follows from \eqref{l33c} and \eqref{l33b} with $\mu=\gamma$. 
\end{proof}

\subsection{Use of the Fourier transform} \label{sec:63} 

For $t\in\D{R}$ we denote 
\begin{equation}  \label{unj}
u_{n,j}(t)\coloneqq\croch{U_n(t)}(j,j) 
\end{equation}
the diagonal entries of the evolution $U_n(t)=\eul^{-\ii tL_{0,n}}\eul^{\ii tL_n}$ introduced in Section \ref{sec:61}.

\begin{lemma}  \label{lem:62} 
If for every $t_0>0$ we have the estimate
\begin{equation}  \label{es0}
\sup_{\substack{\abs{t}\leq t_0\\\abs{j-n}\leq n^{\gamma}}}\abs{\partial_tu_{n,j}(t)}=\ord(n^{-\gamma/2}),
\end{equation} 
then we have the trace estimate \eqref{tr}, \emph{i.e.}, $\C{G}_n=\ord(n^{-\gamma/2}\ln n)$.
\end{lemma}  

\begin{proof} 
Let $\chi\in\C{S}(\D{R})$ be such that $\supp\hat\chi\subset\croch{-t_0,\,t_0}$ with $\hat\chi$ as in \eqref{F}. Hence,
\[ 
\chi(\lambda)=\int_{-\infty}^{\infty}\hat\chi(t)\eul^{\ii\lambda t}\dd t=\int_{-t_0}^{t_0}\hat\chi(t)\eul^{\ii\lambda t}\dd t. 
\] 
Using $L_n-l(n)$ and $L_{0,n}-l(n)$ in place of $\lambda$ we then obtain  
\[ 
\chi\bigl(L_n-l(n)\bigr)-\chi\bigl(L_{0,n}-l(n)\bigr)=\int\hat\chi(t)\eul^{-\ii tl(n)}\left(\eul^{\ii tL_n}-\eul^{\ii tL_{0,n}}\right)\dd t,  
\] 
hence  
\[ 
\C{G}_n=\int\hat\chi(t)\eul^{-\ii tl(n)}\tr\bigl(\theta_{n^{\gamma},n}(L_{0,n})\eul^{\ii tL_{0,n}}(U_n(t)-I)\bigr)\dd t. 
\] 
However, 
\[
\tr\bigl(\theta_{n^{\gamma},n}(L_{0,n})\eul^{\ii tL_{0,n}}(U_n(t)-I)\bigr)=\sum_{j\in\D{Z}}\,\scal{\eul^{-\ii tL_{0,n}}\theta_{n^{\gamma},n}(L_{0,n})\vece_j,(U_n(t)-I)\vece_j}
\] 
and  
\begin{equation} \label{335}
\theta_{n^{\gamma},n}(l_n(j))\neq 0\implies\abs{j-n}\leq n^{\gamma}.
\end{equation} 
Then, for any $j\in\D{Z}$,
\[ 
\eul^{-\ii tL_{0,n}}\theta_{n^{\gamma},n}(L_{0,n})\vece_j=\eul^{-\ii tl_n(j)}\theta_{n^{\gamma},n}(l_n(j))\vece_j
\]
and we can expand $\C{G}_n$ as
\[
\C{G}_n=\sum_{j\in\D{Z}}\C{G}_n(j)\text{ with } 
\C{G}_n(j)\coloneqq\int\hat\chi(t)\,\eul^{\ii t/2}\,\eul^{\ii t(l_n(j)-l(n)-1/2)}\theta_{n^{\gamma},n}(l_n(j)) (u_{n,j}(t)-1)\,\dd t.
\]
Due to Lemma \ref{lem:32} we can find $n_0,c_0>0$ such that 
\[
n\geq n_0\implies\abs{l_n(j)-l(n)-\tfrac{1}{2}}\geq c_0(1+\abs{j-n})
\]
and we can express 
\[
\eul^{\ii t(l_n(j)-l(n)-1/2)}=\frac{-\ii}{l_n(j)-l(n)-1/2}\,\partial_t\eul^{\ii t(l_n(j)-l(n)-1/2)}.
\]
Hence, integrating by parts we obtain $\C{G}_n(j)=\ii\C{G}_{1,n}(j)+\ii\C{G}_{2,n}(j)$ with 
\begin{align*}
\C{G}_{1,n}(j)&=\int\hat\chi(t)\,\eul^{\ii t(l_n(j)-l(n))}\frac{\theta_{n^{\gamma},n}(l_n(j))}{l_n(j)-l(n)-1/2}\,\partial_tu_{n,j}(t)\,\dd t,\\
\C{G}_{2,n}(j)&=\int\partial_t\bigl(\hat\chi(t)\,\eul^{\ii t/2}\bigr)\,\eul^{\ii t(l_n(j)-l(n)-1/2)}\frac{\theta_{n^{\gamma},n}(l_n(j))}{l_n(j)-l(n)-1/2}\,(u_{n,j}(t)-1)\,\dd t.
\end{align*}
Since $\supp\hat\chi\subset\croch{-t_0,\,t_0}$ we have the estimates 
\begin{align*}
\abs{\C{G}_{1,n}(j)}&\leq C\,\frac{\theta_{n^{\gamma},n}(l_n(j))}{1+\abs{j-n}}\,\sup_{\abs{t}\leq t_0}\abs{\partial_tu_{n,j}(t)},\\ 
\abs{\C{G}_{2,n}(j)}&\leq C\,\frac{\theta_{n^{\gamma},n}(l_n(j))}{1+\abs{j-n}}\,\sup_{\abs{t}\leq t_0}\abs{u_{n,j}(t)-1}.
\end{align*}
Combining \eqref{335} with $\sup_{\abs{t}\leq t_0}\abs{u_{n,j}(t)-1}\leq t_0\sup_{\abs{t}\leq t_0}\left\lvert\partial_tu_{n,j}(t)\right\rvert$ we find that the estimate 
\[
\abs{\C{G}_n}\leq\sum_{\abs{j-n}\leq n^{\gamma}}\frac{C_0}{1+\abs{j-n}}\,n^{-\gamma/2}  
\]
holds under assumption \eqref{es0}. To complete the proof we observe that 
\[ 
\sum_{\abs{k}\leq n^{\gamma}}\frac{1}{1+\abs{k}}\leq 1+2\ln n.\qedhere
\]
\end{proof} 

\subsection{Expansion of $\BS{U_n(t)}$} \label{sec:64} 

Since $-\ii\,\partial_tU_n(t)=\eul^{-\ii tL_{0,n}}(L_n -L_{0,n})\eul^{\ii tL_n}$ we can write 
\[ 
-\ii\,\partial_tU_n(t)=H_n(t)U_n(t), 
\] 
with 
\begin{equation}  \label{Hnt}
H_n(t)\coloneqq\eul^{-\ii tL_{0,n}}(L_n-L_{0,n})\eul^{\ii tL_{0,n}}.
\end{equation} 
Since $U_n(0)=I$ we then have the following expansion formula: 
\begin{equation}  \label{U=}
U_n(t)=I+\ii\int_0^tH_n(t_1)\dd t_1+\sum_{\nu=2}^{\infty}\,\ii^{\nu}\int_0^t\!\dd t_1\dots\int_0^{t_{\nu-1}}\!\!H_n(t_1)\dots H_n(t_{\nu})\,\dd t_{\nu}.
\end{equation} 
For $\nu\geq 1$ and $(t_1,\dots,t_{\nu})\in\D{R}^{\nu}$ we denote the diagonal entries of $\ii^{\nu}H_n(t_1)\dots H_n(t_{\nu})$ by
\begin{equation}  \label{gnu}
g_{\nu,n,j}(t_1,\dots,t_{\nu})=\ii^{\nu}\croch{H_n(t_1)\dots H_n(t_{\nu})}(j,j).
\end{equation}  

\begin{lemma}  \label{lem:63} 
We make the following two assumptions:
\begin{enumerate}[\rm(i)]
\item
For any $t_0>0$ we can find $C>0$ such that 
\begin{subequations} \label{es}
\begin{equation} \label{es1} 
\sup_{\substack{\abs{t_1}\leq t_0\\\abs{j-n}\leq n^{\gamma}}}\abs{g_{1,n,j}(t_1)}\leq Cn^{-\gamma/2}.
\end{equation} 
\item
For some $\varepsilon>0$ and for any $t_0>0$ we can find $C>0$ such that the estimates
\begin{equation} \label{es2} 
\sup_{\substack{\abs{t_1},\dots,\abs{t_{\nu-1}}\leq t_0\\\abs{j-n}\leq n^{\gamma}}}\int_{-t_0}^{t_0}\abs{g_{\nu,n,j}(t_1,\dots,t_{\nu})}\,\dd t_{\nu}\leq Cn^{-\gamma/2}
\end{equation}
\end{subequations} 
hold for $\nu\leq n^{\varepsilon}$. 
\end{enumerate}  
Then assumption \eqref{es0} of Lemma~\ref{lem:62} is satisfied, \emph{i.e.}, for any $t_0>0$
\[
\sup_{\substack{\abs{t}\leq t_0\\\abs{j-n}\leq n^{\gamma}}}\abs{\partial_tu_{n,j}(t)}=\ord(n^{-\gamma/2}).
\] 
\end{lemma} 

\begin{proof} 
$g_{\nu,n,j}(t_1,\dots,t_{\nu})$ and $u_{n,j}(t)$ are the $(j,j)$ coefficients of $\ii^{\nu}H_n(t_1)\dots H_n(t_{\nu})$ and $U_n(t)$, see \eqref{unj} and \eqref{gnu}. Thus, the expansion \eqref{U=} of $U_n(t)$ gives for its $(j,j)$ coefficient
\[ 
u_{n,j}(t)=1+\sum_{\nu=1}^{\infty}\int_0^t\dd t_1\dots\int_0^{t_{\nu-1}}\!\!g_{\nu,n,j}(t_1,\dots,t_{\nu})\,\dd t_{\nu},
\] 
and for its derivative
\[
\partial_tu_{n,j}(t)=g_{1,n,j}(t)+\sum_{\nu=2}^{\infty}u_{\nu,n,j}(t),
\]
where the terms $u_{\nu,n,j}$ are defined by
\begin{align*}
u_{2,n,j}(t)&=\int_0^tg_{2,n,j}(t,t_2)\,\dd t_2,\\
u_{\nu,n,j}(t)&=\int_0^t\!\dd t_2\dots\int_0^{t_{\nu-1}}\!\!g_{\nu,n,j}(t,t_2,\dots,t_{\nu})\,\dd t_{\nu},\quad\nu\geq 3.
\end{align*}
In what follows, $\abs{t}\leq t_0$ and $\abs{j-n}\leq n^{\gamma}$.
The term $g_{1,n,j}(t)$ is $\ord(n^{-\gamma/2})$ by \eqref{es1}. The terms of index $\nu<n^{\varepsilon}$ in the sum are estimated using \eqref{es2}: 
\[
\sum_{2\leq\nu<n^{\varepsilon}}\abs{u_{\nu,n,j}(t)}\leq\sum_{2\leq\nu<n^{\varepsilon}}\frac{Ct_0^{\nu-1}n^{-\gamma/2}}{(\nu-1)!}\,\leq C\eul^{t_0}n^{-\gamma/2}=\ord(n^{-\gamma/2}).
\]
To complete the proof it remains to observe that 
\begin{equation} \label{4end} 
\sum_{\nu>n^{\varepsilon}}\abs{u_{\nu,n,j}(t)}\leq\sum_{\nu>n^{\varepsilon}}\frac{Ct_0^{\nu-1}}{(\nu-1)!}\leq\frac{C\eul^{t_0}}{(\lfloor n^{\varepsilon}\rfloor-1)!}\,,
\end{equation} 
where $\lfloor s\rfloor\coloneqq\max\accol{k\in\D{Z}:k\leq s}$. Since $k!\sim(k/\eul)^k$ it is clear that the right-hand side of \eqref{4end} is rapidly decreasing when $n\to\infty$. 
\end{proof} 

\subsection{Summary} \label{sec:65} 

Lemmas \ref{lem:61}, \ref{lem:62}, and \ref{lem:63} reduce Proposition \ref{prop:5} to the proof of assumptions \eqref{es1} and  \eqref{es2} of Lemma \ref{lem:63} for some $\varepsilon>0$.

\begin{proposition} \label{prop:6}  
\begin{enumerate}[\rm(i)]
\item
For any $t_0>0$ we can find $C>0$ such that \eqref{es1} holds.
\item
For any $t_0>0$ and any $0<\varepsilon\leq\gamma/16$ we can find $C>0$ such that estimates \eqref{es2} hold for $\nu\leq n^{\varepsilon}$.  
\end{enumerate}
\end{proposition}

\begin{proof} 
See Sections \ref{sec:9} and \ref{sec:10}. 
\end{proof}

\begin{proof}[Proof of Proposition~\ref{prop:6}$\implies$Proposition~\ref{prop:5}]
By Proposition~\ref{prop:6}, both assumptions of Lemma \ref{lem:63} are satisfied for $0<\varepsilon\leq\gamma/16$. Hence Lemma \ref{lem:63} applies and assumption \eqref{es0} of Lemma \ref{lem:62} is satisfied. Thus, Lemma \ref{lem:62} also applies and estimate \eqref{tr} holds. Finally, estimates \eqref{tr} and \eqref{Gn-Gn0} from Lemma \ref{lem:61} imply estimate \eqref{tr0} in Proposition \ref{prop:5}.
\end{proof}

\section{The class of operators $q(\Lambda,S)$}  \label{sec:7}
\subsection{Plan of Section \ref{sec:7}} \label{sec:71} 

The aim of this section is to describe a class of operators in $l^2(\D{Z})$ which are needed in Sections \ref{sec:8}-\ref{sec:10}. These operators are denoted by $q(\Lambda,S)$ and defined in Section \ref{sec:72} by Fourier transform. In Section \ref{sec:72} we prove Lemma \ref{lem:71} which computes $q_1(\Lambda,S)q_2(\Lambda,S)^*$. In Section \ref{sec:73} we prove Lemma \ref{lem:72} which computes the conjugate $\eul^{-\ii s\Lambda}q(\Lambda,S)\eul^{\ii s\Lambda}$. In Section \ref{sec:75} we prove Lemma \ref{lem:74} which gives a specific composition formula. In Section \ref{sec:76} we prove Lemma \ref{lem:75} which gives a norm estimate used in Sections \ref{sec:8}-\ref{sec:10}. Finally, in Section \ref{sec:77} we prove Lemma \ref{lem:76} which estimates the norm of the commutator of $q(\Lambda,S)$ with diagonal operators.

\subsection{Notations} \label{sec:711}

Further on, we denote
\begin{enumerate}[\textbullet]
\item
$\cercle\coloneqq\accol{z\in\D{C}:\abs{z}=1}=\D{R}/2\pi\D{Z}$ the unit circle.
\item
$\R{L}^2(\cercle)$ the Hilbert space of classes of square integrable functions $f\colon\cercle\to\D{C}$ equipped with the scalar product $\scal{f,g}=\int_0^{2\pi}\overline{f(\eul^{\ii\xi})}g(\eul^{\ii\xi})\frac{\dd\xi}{2\pi}$.
\item
$\accol{f_j}_{j\in\D{Z}}$ the orthonormal basis defined by $f_j(\eul^{\ii\xi})=\eul^{\ii j\xi}$ for $\xi\in\D{R}$.
\item
$\C{F}_0\colon\R{L}^2(\cercle)\to l^2(\D{Z})$ the Fourier transform which is a unitary isomorphism such that $\C{F}_0f_n=\vece_n$:
\[
(\C{F}_0f)(j)=\scal{f_j,f}_{\R{L}^2(\cercle)}=\int_0^{2\pi}f(\eul^{\ii\xi})\,\eul^{-\ii j\xi}\,\frac{\dd\xi}{2\pi}.  
\] 
\item
$\norm{p}_{\class^m(\cercle)}\coloneqq\max_{0\leq i\leq m}\sup_{\xi\in\D{R}}\bigl\lvert\partial_{\xi}^ip(\eul^{\ii\xi})\bigr\rvert$ the $\class^m$-norm of $p\in\class^{\infty}(\cercle)$.
\item
$\tau_s\colon\cercle\to\cercle$, $s\in\D{R}$ the translation $\eul^{\ii\xi}\to\eul^{\ii(\xi-s)}$.
\item
$\tilde\tau_s\colon\D{Z}\times\cercle\to\D{Z}\times\cercle$ its extension $(j,\eul^{\ii\xi})\to(j,\eul^{\ii(\xi-s)})$.
\end{enumerate}
    
\subsection{Operators $\BS{p(S)}$ and $\BS{q(\Lambda,S)}$} \label{sec:72} 

\subsubsection{Operators $p(S)$}  \label{sec:721}

If $p\in\class^{\infty}(\cercle)$ we define $p(S)\in\C{B}(l^2(\D{Z}))$ by functional calculus. Since $\C{F}^{-1}_0S\C{F}_0f_n=f_{n+1}$ we have $(\C{F}^{-1}_0S\C{F}_0f)(\eul^{\ii\xi})=\eul^{\ii\xi}f(\eul^{\ii\xi})$. Thus, by Fourier transform $p(S)$ is the operator of multiplication by $p$, i.e., $(\C{F}^{-1}_0p(S)\C{F}_0f)(\eul^{\ii\xi})=p(\eul^{\ii\xi})f(\eul^{\ii\xi})$, so that
\begin{equation}  \label{pscoeff}
p(S)(j,k)\coloneqq\scal{\vece_j,\,p(S)\vece_k}=\scal{f_j,\,pf_k}_{\R{L}^2(\cercle)}=\int_0^{2\pi}p(\eul^{\ii\xi})\,\eul^{\ii(k-j)\xi}\,\frac{\dd\xi}{2\pi}.
\end{equation}

\begin{properties*}[of $p(S)$]
Let $p,p_1,p_2\in\class^{\infty}(\cercle)$.
\begin{enumerate}[1)]
\item \label{cfpstar}
$p(S)^*=\bar p(S)$.
\item  \label{cfp}
$(p_1p_2)(S)=p_1(S)p_2(S)$.
\end{enumerate}
\end{properties*}

\subsubsection{Operators $q(\Lambda,S)$}  \label{sec:722}

We consider two classes $\C{Q}^0$, $\C{Q}$ of functions $q^0,q\colon\D{Z}\times\cercle\to\D{C}$: 
\begin{enumerate}[a)]
\item
$q^0\in\C{Q}^0$ if $q^0(j,\,\cdot\,)\in\class^{\infty}(\cercle)$ for each $j\in\D{Z}$ and $q^0(j,\,\cdot\,)=0$ for large $\abs{j}$.
\item
$q\in\C{Q}$ if there exist $p\in\class^{\infty}(\cercle)$ and $q^0\in\C{Q}^0$ such that
\begin{equation}  \label{q}
q(j,\eul^{\ii\xi})=p(\eul^{\ii\xi})+q^0(j,\eul^{\ii\xi})=p(\eul^{\ii\xi})+p_j(\eul^{\ii\xi}),
\end{equation}
where $p_j(\eul^{\ii\xi})\coloneqq q^0(j,\eul^{\ii\xi})$. Let us note that $p_j\equiv 0$ for large $\abs{j}$. Moreover, $p$, $q^0$, and the $p_j$'s are uniquely determined by $q$ since $p(\eul^{\ii\xi})=q(j,\eul^{\ii\xi})$ for $\abs{j}\gg 0$.
\end{enumerate}

\begin{remark*}
If $q(j,\eul^{\ii\xi})=\eul^{\ii\tilde\psi(j,\eul^{\ii\xi})}$ with $\tilde\psi\in\C{Q}^0$, then $q-1\in\C{Q}^0$ and $q\in\C{Q}$.   
\end{remark*}

\begin{definition*}[of $q(\Lambda,S)$]
Let $q^0\in\C{Q}^0$. Let also $q\in\C{Q}$ be as in \eqref{q}.
\begin{enumerate}[(a)]
\item
The operator $q^0(\Lambda,S)\in\C{B}\bigl(l^2(\D{Z})\bigr)$ is the finite rank operator defined by
\begin{equation} \label{opq0}
q^0(\Lambda,S)=\sum_{j\in\D{Z}}\varPi_jp_j(S)
\end{equation}
where $\varPi_j=\scal{\vece_j,\,\cdot\,}\vece_j$ is the orthogonal projection on $\vece_j$ and $p_j(\eul^{\ii\xi})\coloneqq q^0(j,\eul^{\ii\xi})$.
\item
The operator $q(\Lambda,S)\in\C{B}\bigl(l^2(\D{Z})\bigr)$ is defined by
\begin{equation} \label{opq}
q(\Lambda,S)=p(S)+q^0(\Lambda,S).
\end{equation}
\end{enumerate}
\end{definition*}

\begin{properties*}[of $q(\Lambda,S)$]
We assume $q\in\C{Q}$.
\begin{enumerate}[(i)]
\item \label{cfq}
It follows from \eqref{pscoeff} using \eqref{opq} and \eqref{opq0} that the matrix elements of $q(\Lambda,S)$ are given by
\begin{equation} \label{qls}
q(\Lambda,S)(j,k)=\int_0^{2\pi}\,q(j,\eul^{\ii\xi})\,\eul^{\ii(k-j)\xi}\,\frac{\dd\xi}{2\pi}\,.
\end{equation}
\item  \label{cf0}
If $\tilde q(j,\eul^{\ii\xi})=q(j,\eul^{\ii\xi})\tilde p(\eul^{\ii\xi})$ with $\tilde p\in\class^{\infty}(\cercle)$, then $\tilde q(\Lambda,S)=q(\Lambda,S)\tilde p(S)$. Indeed, by \eqref{opq} and \eqref{opq0},
\begin{align*}
\tilde q(\Lambda,S)&=(p\tilde p)(S)+\sum_j\varPi_j(p_j\tilde p)(S)\\
&=\Bigl(p(S)+\sum_j\varPi_jp_j(S)\Bigr)\tilde p(S)=q(\Lambda,S)\tilde p(S).
\end{align*}
\item  \label{cf01}
Let $\theta\colon\D{Z}\to\D{C}$ be of finite support. If $\tilde q(j,\eul^{\ii\xi})=\theta(j)q(j,\eul^{\ii\xi})$, then $\tilde q(\Lambda,S)=\theta(\Lambda)q(\Lambda,S)$. In particular, if $\tilde q(j,\eul^{\ii\xi})=\theta(j)$, then $\tilde q(\Lambda,S)=\theta(\Lambda)$. By \eqref{cfq} we indeed have
\begin{align*}
\tilde q(\Lambda,S)(j,k)&=\theta(j)\int_0^{2\pi}q(j,\eul^{\ii\xi})\,\eul^{\ii(k-j)\xi}\,\frac{\dd\xi}{2\pi}\\
&=\theta(j)q(\Lambda,S)(j,k)=\bigl(\theta(\Lambda)q(\Lambda,S)\bigr)(j,k).
\end{align*}
\end{enumerate}
\end{properties*}

\begin{lemma} \label{lem:71} 
If $q_1,\,q_2\in\C{Q}^0$, then the matrix elements of $q_1(\Lambda,S){q_2(\Lambda,S)}^*$ are given by
\begin{equation} \label{cf}
\bigl(q_1(\Lambda,S)\,q_2(\Lambda,S)^*\bigr)(j,k)=\int_0^{2\pi}q_1(j,\eul^{\ii\xi})\overline{q_2(k,\eul^{\ii\xi})}\,\eul^{\ii(k-j)\xi}\,\frac{\dd\xi}{2\pi}\,. 
\end{equation} 
\end{lemma} 

\begin{proof} 
Let $p_{i,j}\coloneqq q_i(j,\,\cdot\,)$, $i=1,2$. By \eqref{opq0}, $q_i(\Lambda,S)=\sum_{j\in\D{Z}}\varPi_jp_{i,j}(S)$. Hence
\[
q_1(\Lambda,S)\,q_2(\Lambda,S)^*=\sum_{l,m\in\D{Z}}\varPi_lp_{1,l}(S)\bar p_{2,m}(S)\varPi_m.
\]
Using \eqref{pscoeff} at last line below we then have
\begin{align*}
\bigl(q_1(\Lambda,S)\,q_2(\Lambda,S)^*\bigr)(j,k)
&=\sum_{l,m\in\D{Z}}\scal{\varPi_l\vece_j,p_{1,l}(S)\bar p_{2,m}(S)\varPi_m\vece_k}\\
&=\scal{\vece_j,(p_{1,j}\bar p_{2,k})(S)\vece_k}\\
&=\int_0^{2\pi}p_{1,j}(\eul^{\ii\xi})\overline{p_{2,k}(\eul^{\ii\xi})}\,\eul^{\ii(k-j)\xi}\,\frac{\dd\xi}{2\pi}\,.\qedhere
\end{align*}
\end{proof} 

\subsection{Conjugation of $\BS{q(\Lambda,S)}$ by $\BS{\eul^{\ii s\Lambda}}$} \label{sec:73} 

For $s\in\D{R}$ and $p\in\class^{\infty}(\cercle)$ we have the formula
\[
\eul^{-\ii s\Lambda}p(S)\eul^{\ii s\Lambda}=(p\circ\tau_s)(S).
\]
Indeed, $\eul^{-\ii s\Lambda}S\eul^{\ii s\Lambda}=\eul^{-\ii s}S$, hence $\eul^{-\ii s\Lambda}p(S)\eul^{\ii s\Lambda}=p(\eul^{-\ii s}S)=(p\circ\tau_s)(S)$. More generally:

\begin{lemma}[unitary conjugation] \label{lem:72} 
If $q\in\C{Q}$ then for any $s\in\D{R}$,
\begin{equation} \label{eQe} 
\eul^{-\ii s\Lambda}q(\Lambda,S)\eul^{\ii s\Lambda}=(q\circ\tilde\tau_s)(\Lambda,S).
\end{equation} 
\end{lemma} 

\begin{proof} 
It suffices to check that both sides of \eqref{eQe} have the same matrix elements. Using that $\eul^{\ii s\Lambda}\vece_m=\eul^{\ii sm}\vece_m$ together with \eqref{qls}, we express the $(j,k)$ coefficient of the left-hand side as 
\begin{equation} \label{eQe'} 
\eul^{\ii s(k-j)}\scal{\vece_j,q(\Lambda,S)\vece_k}=\int_0^{2\pi}\!\eul^{\ii(k-j)(\xi+s)}q(j,\eul^{\ii\xi})\,\frac{\dd\xi}{2\pi}\,.
\end{equation}  
The change of variable $\eta=\xi+s$ allows us to express the right-hand side of \eqref{eQe'} as 
\[  
\int_s^{s +2\pi}\eul^{\ii(k-j)\eta}q(j,\eul^{\ii(\eta-s)})\frac{\dd\eta}{2\pi}=\int_0^{2\pi}\eul^{\ii(k-j)\eta}q(j,\eul^{\ii(\eta-s)})\frac{\dd\eta}{2\pi}\,,  
\] 
where the right-hand side is now the $(j,k)$ coefficient of the right-hand side of \eqref{eQe} and where we used that the integral of a $2\pi$-periodic function on $\croch{s,s+2\pi}$ is the same as on $\croch{0,2\pi}$. 
\end{proof}

\subsection{A composition formula} \label{sec:75} 

To state the composition formula we first describe preliminary constructions. 

\subsubsection{Framework}  \label{sec:751}

It involves a sequence of functions $\tilde\psi_n\in\C{Q}^0$, $n\geq 1$ with the following properties: 
\begin{subequations} \label{pf'Cmf}
\begin{equation} \label{pf'}
\tilde\psi_n(j,\eul^{\ii\xi})=\psi_n(\eul^{\ii\xi})+(j-n)\varphi_n(\eul^{\ii\xi})\text{ for }\abs{j-n}\leq n/3
\end{equation}  
with $\psi_n$, $\varphi_n\in\class^{\infty}(\cercle)$ real-valued and such that 
\begin{equation} \label{phin.asymptot}
\norm{\varphi_n}_{\class^m(\cercle)}=\ord(n^{\gamma-1}) 
\end{equation}
for every integer $m\geq 0$. In particular, for some $n_0$ depending on $\accol{\varphi_n}$ we have 
\begin{equation} \label{Cmf}
\sup_{n\geq n_0}\norm{\varphi_n}_{\class^1(\cercle)}\leq 1/2.
\end{equation} 
\end{subequations}
To such data we attach auxiliary functions $\eta_n$, $\xi_n$, $\tilde\xi_n$, $\F{p}_n$, $\vartheta_n$, and $\tilde\vartheta_n$. We define $\eta_n\colon\D{R}\to\D{R}$ by
\begin{equation} \label{eta} 
\eta_n(\xi)\coloneqq\xi-\varphi_n(\eul^{\ii\xi}). 
\end{equation}
Then $\eta_n(\xi+2\pi)=\eta_n(\xi)+2\pi$ and due to property \eqref{Cmf} its derivative satisfies
\[ 
\partial_{\xi}\eta_n(\xi)=1-\partial_{\xi}\varphi_n(\eul^{\ii\xi})\geq 1/2\ \text{ for }n\geq n_0.
\]  
Therefore $\eta_n\colon\D{R}\to\D{R}$ is bijective for $n\geq n_0$. Let $\xi_n\colon\D{R}\to\D{R}$ denote its inverse. It satisfies
\begin{equation} \label{xin}
\xi_n(\eta)-\varphi_n(\eul^{\ii\xi_n(\eta)})=\eta.
\end{equation}
Since $\eta\to\xi_n(\eta)-\eta$ is $2\pi$-periodic, we can then define $\tilde\xi_n\colon\cercle\to\D{R}$ by 
\[
\tilde\xi_n(\eul^{\ii\eta})=\xi_n(\eta)-\eta=\varphi_n(\eul^{\ii\xi_n(\eta)}). 
\] 
By derivation we also introduce $\F{p}_n\colon\cercle\to\D{R}$ defined by
\begin{equation} \label{compo}
\F{p}_n(\eul^{\ii\eta})\coloneqq 1+\partial_{\eta}\tilde\xi_n(\eul^{\ii\eta})=\partial_{\eta}\xi_n(\eta).
\end{equation}
Finally, we consider $\vartheta_n\colon\cercle\to\cercle$ and its extension $\tilde\vartheta_n\coloneqq\id_{\D{Z}}\times\vartheta_n\colon\D{Z}\times\cercle\to\D{Z}\times\cercle$ defined by
\begin{subequations} \label{tildevtt}
\begin{align} \label{vt}
\vartheta_n(\eul^{\ii\eta})&\coloneqq\eul^{\ii\eta}\eul^{\ii\tilde\xi_n(\eul^{\ii\eta})}=\eul^{\ii\xi_n(\eta)}\\
\label{tildevt}
\tilde\vartheta_n(j,\eul^{\ii\eta})&\coloneqq(j,\vartheta_n(\eul^{\ii\eta})).
\end{align}
\end{subequations}

\begin{lemma} \label{lem:73} 
Under assumption \eqref{phin.asymptot} we have the estimates 
\begin{equation} \label{normtildexin}
\norm{\tilde\xi_n}_{\R{C}^{m}(\cercle)}=\ord(n^{\gamma-1})
\end{equation} 
for any integer $m\geq 0$. Moreover,
\begin{equation} \label{xin'}
\norm{\F{p}_n-1}_{\class^0(\cercle)}=\ord(n^{\gamma-1}).
\end{equation}
\end{lemma}

\begin{proof} 
For $m=0$ \eqref{normtildexin} follows from the relation $\tilde\xi_n(\eul^{\ii\eta})=\varphi_n(\eul^{\ii\xi_n(\eta)})$, using \eqref{phin.asymptot} for $m=0$. Let $\varphi_n^{(m)}(\eul^{\ii\eta})\coloneqq\partial_{\eta}^m\varphi_n(\eul^{\ii\eta})$ and $\xi_n^{(m)}(\eta)\coloneqq\partial_{\eta}^m\xi_n(\eta)$. For $m=1$, differentiating \eqref{xin} we obtain
\begin{equation} \label{xindif}
\xi_n^{(1)}(\eta)\bigl(1-\varphi_n^{(1)}(\eul^{\ii\xi_n(\eta)})\bigr)=1. 
\end{equation}
Using \eqref{phin.asymptot} for $m=1$ we get \eqref{xin'}:
\[
\sup_{\eta\in\D{R}}\,\abs{\xi_n^{(1)}(\eta)-1}=\sup_{\eta\in\D{R}}\left\lvert\frac{\varphi_n^{(1)}(\eul^{\ii\xi_n(\eta)})}{1-\varphi_n^{(1)}(\eul^{\ii\xi_n(\eta)})}\right\rvert=\ord(n^{\gamma-1}).
\]
For $m\geq 2$, the proof of \eqref{normtildexin} is by induction on $m$. By successive differentiations of \eqref{xindif} we can express $\xi_n^{(m)}(\eta)\bigl(1-\varphi_n^{(1)}(\eul^{\ii\xi_n(\eta)})\bigr)$ as a linear combination of products of factors of the form $\xi_n^{(m')}$, $m'<m$ with some factor $\varphi_n^{(m'')}$, $m''\leq m$, and we get \eqref{normtildexin}, again using \eqref{phin.asymptot}. 
\end{proof} 

\subsubsection{Composition formula}  \label{sec:752}

\begin{lemma}[composition formula] \label{lem:74} 
Let $\tilde\psi_n^0$, $\tilde\psi_n\in\C{Q}^0$. We assume $\tilde\psi_n$ satisfies \eqref{pf'Cmf} for some $n_0$. Let also $\theta_n^0,\,\theta_n\in\class_0^{\infty}(\D{R})$ be real-valued, vanishing outside the interval $\croch{2n/3,4n/3}$. 

If $Q_n^0=\bigl(\theta_n^0\eul^{\ii\tilde\psi_n^{\,0}}\bigr)(\Lambda,S)$ and $Q_n=\bigl(\theta_n\eul^{\ii\tilde\psi_n}\bigr)(\Lambda,S)$, then for $n\geq n_0$,
\begin{equation}  \label{com}
Q_n^0Q_n^*=\bigl(\theta_n^0\eul^{\ii(\tilde\psi_n^0-\tilde\psi_n)\circ\tilde\vartheta_n}\bigr)(\Lambda,S)\,\F{p}_n(S)\,\widetilde\Theta_n,
\end{equation} 
where $\tilde\vartheta_n$ is given by \eqref{tildevtt}, $\F{p}_n$ by \eqref{compo}, and $\widetilde\Theta_n\coloneqq\theta_n(\Lambda)$. 
\end{lemma} 

\begin{proof} 
We have $Q_n^0=q_n^0(\Lambda,S)$ with $q_n^0(j,\eul^{\ii\xi})=\theta_n^0(j)\eul^{\ii\tilde\psi_n^0(j,\eul^{\ii\xi})}$, and similarly $Q_n=q_n(\Lambda,S)$ with $q_n(j,\eul^{\ii\xi})=\theta_n(j)\eul^{\ii\tilde\psi_n(j,\eul^{\ii\xi})}$. Hence $Q_n^0,\,Q_n\in\C{B}\bigl(l^2(\D{Z})\bigr)$ since $q_n^0,q_n\in\C{Q}^0$. To prove \eqref{com} it suffices to prove that both sides have the same $(j,k)$ coefficient. If $K_n\coloneqq Q_n^0Q_n^*$, then Lemma \ref{lem:71} gives
\[
K_n(j,k)=\theta_n^0(j)\theta_n(k)\int_0^{2\pi}\eul^{\ii\tilde\psi_n^0(j,\eul^{\ii\xi})-\ii\tilde\psi_n(k,\eul^{\ii\xi})+\ii(k-j)\xi}\,\frac{\dd\xi}{2\pi}\,.  
\]
Thus $K_n(j,k)=0$ either if $\abs{j-n}>n/3$ or if $\abs{k-n}>n/3$. Assume now that $\abs{j-n}\leq n/3$ and $\abs{k-n}\leq n/3$. By assumption \eqref{pf'}, $\tilde\psi_n(j,\eul^{\ii\xi})-\tilde\psi_n(k,\eul^{\ii\xi})=(j-k)\varphi_n(\eul^{\ii\xi})$ and we find 
\[ 
K_n(j,k)=\theta_n^0(j)\theta_n(k)\int_0^{2\pi}\eul^{\ii(\tilde\psi_n^0-\tilde\psi_n)(j,\eul^{\ii\xi})+\ii(k-j)(\xi-\varphi_n(\eul^{\ii\xi}))}\,\frac{\dd\xi}{2\pi}\,.
\]
As above, let $\xi_n\colon\D{R}\to\D{R}$ denote the inverse of $\eta_n$ for $n\geq n_0$ with $\eta_n$ defined by \eqref{eta}, i.e.\ $\eta_n(\xi)=\xi-\varphi_n(\eul^{\ii\xi})$. Due to \eqref{xin} and $\xi_n(\croch{0,\,2\pi})=\croch{\xi_n(0),\,\xi_n(0)+2\pi}$, the change of variable $\xi=\xi_n(\eta)$ gives
\[
K_n(j,k)=\theta_n^0(j)\theta_n(k)\int_{\xi_n(0)}^{\xi_n(0)+2\pi}\eul^{\ii(\tilde\psi_n^0-\tilde\psi_n)(j,\eul^{\ii\xi_n(\eta)})+\ii(k-j)\eta}\,\partial_{\eta}\xi_n(\eta)\,\frac{\dd\eta}{2\pi}\,.
\] 
Using \eqref{compo}, i.e.\ $\partial_{\eta}\xi_n(\eta)\eqqcolon\F{p}_n(\eul^{\ii\eta})$ and $(j,\eul^{\ii\xi_n(\eta)})=\tilde\vartheta_n(j,\eul^{\ii\eta})$ we get
\[
K_n(j,k)=\theta_n^0(j)\theta_n(k)\int_{\xi_n(0)}^{\xi_n(0)+2\pi}\eul^{\ii(\tilde\psi_n^0-\tilde\psi_n)\circ\tilde\vartheta_n(j,\eul^{\ii\eta})}\,\F{p}_n(\eul^{\ii\eta})\,\eul^{\ii(k-j)\eta}\,\frac{\dd\eta}{2\pi}\,.
\]
The function we integrate is $2\pi$-periodic, hence its integral above is the same as over $\croch{0,2\pi}$. Thus,
\[
K_n(j,k)=\theta_n(k)\int_0^{2\pi}\tilde q_n(j,\eul^{\ii\eta})\,\eul^{\ii(k-j)\eta}\,\frac{\dd\eta}{2\pi}=\theta_n(k)\,\tilde q_n(\Lambda,S)(j,k),
\]
where $\tilde q_n(j,\eul^{\ii\eta})\coloneqq\tilde q_n^0(j,\eul^{\ii\eta})\F{p}_n(\eul^{\ii\eta})$ with $\tilde q_n^0(j,\eul^{\ii\eta})\coloneqq\theta_n^0(j)\eul^{\ii(\tilde\psi_n^0-\tilde\psi_n)\circ\tilde\vartheta_n(j,\eul^{\ii\eta})}$. Let us observe that $\tilde q_n(\Lambda,S)=\tilde q_n^0(\Lambda,S)\,\F{p}_n(S)$ by property \eqref{cf0}. Moreover,
\[
\theta_n(k)\,\tilde q_n(\Lambda,S)(j,k)=\bigl(\tilde q_n(\Lambda,S)\,\theta_n(\Lambda)\bigr)(j,k).
\]
Thus, $K_n$ and $\tilde q_n^0(\Lambda,S)\,\F{p}_n(S)\,\widetilde\Theta_n$ have the same $(j,k)$ coefficient if $\abs{j-n}\leq n/3$ and $\abs{k-n}\leq n/3$. Otherwise, the $(j,k)$ coefficients of both sides of \eqref{com} vanish as multiples of $\theta_n^0(j)\theta_n(k)$. 
\end{proof}

\subsection{A norm estimate} \label{sec:76} 

\begin{lemma} \label{lem:75} 
Let $\widetilde Q_n\coloneqq\tilde q_n(\Lambda,S)$ be defined by $\tilde q_n\coloneqq\theta_n\eul^{\ii\tilde\psi_n}q_n$ with the following assumptions:
\begin{enumerate}[\rm(i)]
\item
$\theta_n\in\class_0^{\infty}(\D{R})$ is real-valued, $0\leq\theta_n\leq 1$, and $\theta_n(s)=0$ for $\abs{s-n}>n/3$.
\item
$\tilde\psi_n\colon\D{Z}\times\cercle\to\D{R}$ is of the form $\tilde\psi_n(j,\eul^{\ii\xi})=\psi_n(\eul^{\ii\xi})+(j-n)\varphi_n(\eul^{\ii\xi})$ for $\abs{j-n}\leq n/3$, with $\psi_n,\varphi_n\in\class^{\infty}(\cercle)$ real-valued; moreover, for some $n_0$:
\begin{equation} \label{assphin}
\sup_{n\geq n_0}\norm{\varphi_n}_{\class^2(\cercle)}\leq 1/2. 
\end{equation}
\item
$q_n\in\C{Q}$.
\end{enumerate}
Then, for $n\geq n_0$,  
\begin{equation} \label{eL42}
\norm{\widetilde Q_n}\leq 4\sqrt{\ln n}\,\sup_{\abs{j-n}\leq n/3}\norm{q_n(j,\,\cdot\,)}_{\class^1(\cercle)}.
\end{equation} 
\end{lemma} 

\begin{remark*} 
This lemma will be applied for $\theta_n=\theta_{n,n}$ or $\theta_{3n/2,n}$ defined according to \eqref{222}.
\end{remark*}

\begin{proof} 
Further on, we assume $n\geq n_0$. By assumption (i), $\tilde q_n\in\C{Q}^0$, hence $\widetilde Q_n\in\C{B}\bigl(l^2(\D{Z})\bigr)$. By the Schur test of boundedness in $l^2(\D{Z})$ applied to $K_n\coloneqq\widetilde Q_n\widetilde Q_n^*$ we get 
\begin{equation} \label{schur}
\norm{\widetilde Q_n}^2=\norm{K_n}\leq\sup_{j\in\D{Z}}\,\sum_{k\in\D{Z}}\abs{K_n(j,k)}.
\end{equation}
We first observe that $K_n(j,k)=0$ if $\abs{j-n}>n/3$, and also if $\abs{k-n}>n/3$. It is a consequence of \eqref{cf} since $\tilde q_n(j,\eul^{\ii\xi})=0$ for $\abs{j-n}>n/3$. Thus we can assume $\abs{j-n}\leq n/3$ and $\abs{k-n}\leq n/3$. Let $\eta_n(\xi)=\xi-\varphi_n(\eul^{\ii\xi})$ be as in \eqref{eta}. Using $\tilde\psi_n(j,\eul^{\ii\xi})-\tilde\psi_n(k,\eul^{\ii\xi})=(j-k)\varphi_n(\eul^{\ii\xi})$, Lemma \ref{lem:71} gives
\[
K_n(j,k)=\int_0^{2\pi}\eul^{\ii(k-j)\eta_n(\xi)}\theta_n(j)q_n(j,\eul^{\ii\xi})\,\overline{q_n(k,\eul^{\ii\xi})}\,\theta_n(k)\,\frac{\dd\xi}{2\pi}\,.
\] 
Moreover,
\begin{equation} \label{knjj}
\abs{K_n(j,j)}\leq\norm{q_n(j,\,\cdot\,)}_{\class^1(\cercle)}^2.
\end{equation}
Since $\abs{\partial_{\xi}\varphi_n(\eul^{\ii\xi})}\leq 1/2$ by \eqref{assphin}, then $\abs{\partial_{\xi}\eta_n(\xi)}=\abs{1-\partial_{\xi}\varphi_n(\eul^{\ii\xi})}\geq 1/2$, and we can introduce 
\[
b_n(j,\eul^{\ii\xi},k)\coloneqq\frac{\theta_n(j)q_n(j,\eul^{\ii\xi})\,\overline{q_n(k,\eul^{\ii\xi})}\,\theta_n(k)}{\partial_{\xi}\eta_n(\xi)}\,. 
\]
Thus, for $k\neq j$, 
\[
K_n(j,k)=-\frac{\ii}{k-j}\int_0^{2\pi}\partial_{\xi}\bigl(\eul^{\ii(k-j)\eta_n(\xi)}\bigr)b_n(j,\eul^{\ii\xi},k)\,\frac{\dd\xi}{2\pi}\,,  
\]
then by integration by parts 
\[
K_n(j,k)=\frac{\ii}{k-j}\int_0^{2\pi}\eul^{\ii(k-j)\eta_n(\xi)}\partial_{\xi}b_n(j,\eul^{\ii\xi},k)\,\frac{\dd\xi}{2\pi}\,,   
\]
which gives the estimate
\begin{equation} \label{knjk}
\abs{K_n(j,k)}\leq\frac{\norm{b_n(j,\,\cdot\,,k)}_{\class^1(\cercle)}}{\abs{k-j}}\,.
\end{equation}
Let us note that $b_n(j,\,\cdot\,,k)\neq 0$ implies $\abs{j-n}\leq n/3$ and $\abs{k-n}\leq n/3$. We then denote
\[
M\coloneqq\sup_{\abs{j-n}\leq n/3}\norm{q_n(j,\,\cdot\,)}_{\class^1(\cercle)}.
\]
By assumption \eqref{assphin} we have $\abs{\partial_{\xi}\eta_n(\xi)}\geq 1/2$ and $\abs{\partial_{\xi}^2\eta_n(\xi)}=\abs{\partial_{\xi}^2\varphi_n(\eul^{\ii\xi}))}\leq 1/2$, hence we get
\[ 
\sup_{j,k\in\D{Z}}\norm{b_n(j,\,\cdot\,,k)}_{\class^1(\cercle)}=\sup_{\substack{\abs{j-n}\leq n/3\\\abs{k-n}\leq n/3}}\norm{b_n(j,\,\cdot\,,k)}_{\class^1(\cercle)}\leq 2M^2+2M^2+4M^2\times\frac{1}{2}=6M^2.
\] 
Thus, using \eqref{knjj} and \eqref{knjk},
\begin{align*}
\sup_{j\in\D{Z}}\sum_{k\in\D{Z}}\abs{K_n(j,k)}&=\sup_{\abs{j-n}\leq n/3}\Bigl(\abs{K_n(j,j)}+\sum_{\substack{\abs{k-n}\leq n/3\\k\neq j}}\abs{K_n(j,k)}\Bigr)\\
&\leq\Bigl(1+12\sum_{1\leq m\leq n/3}\frac{1}{m}\Bigr)M^2\\
&\leq 16M^2\ln n,
\end{align*}
with $n>1$ for the last inequality. The proof is completed due to \eqref{schur}.  
\end{proof}

\begin{remark*}
The norm estimate of Lemma \ref{lem:75} is not optimal. The logarithmic factor in the right-hand side of \eqref{eL42} can be replaced by a suitable estimate of $q_n(j+1,\,\cdot\,)-q_n(j,\,\cdot\,)$. Since the presence of logarithmic factors makes no difference for the remainder estimates we consider, our choice is to use the simplest assumptions and a non-optimal norm estimate. 
\end{remark*}

\subsection{A commutator estimate} \label{sec:77} 

Further on, $\Theta_n$ is the operator defined by
\begin{equation}   \label{op.theta.n}
\Theta_n\coloneqq\theta_{n,n}(\Lambda)=\theta_0\Bigl(\tfrac{1}{n}\Lambda-I\Bigr)
\end{equation}
where $\theta_{n,n}$ and $\theta_0$ are as in \eqref{222}.

\begin{lemma} \label{lem:76} 
Let $Q_n\coloneqq q_n(\Lambda,S)$ be defined by $q_n\coloneqq\theta_{n,n}\eul^{\ii\tilde\psi_n}$ with the following assumptions: 
\begin{enumerate}[\rm(i)]
\item
$\tilde\psi_n(j,\eul^{\ii\xi})=\psi_n(\eul^{\ii\xi})+(j-n)\varphi_n(\eul^{\ii\xi})$ for $\abs{j-n}\leq n/3$, with $\psi_n,\varphi_n\in\class^{\infty}(\cercle)$ real-valued,
\item
$\sup_{n\geq n_0}\norm{\varphi_n}_{\class^2(\cercle)}\leq 1/2$. 
\end{enumerate}
We then have the estimate
\begin{equation}  \label{commutator.estimate}
\norm{\croch{\Theta_n,Q_n}}\leq C\,\frac{\sqrt{\ln n}}{n}\sup_{\abs{j-n}\leq n/3}\norm{\tilde\psi_n(j,\,\cdot\,)}_{\class^2(\cercle)}
\end{equation}
where $C$ is some positive constant.
\end{lemma} 

\begin{proof} 
The inverse Fourier formula allows us to express
\[
\Theta_n=\int_{-\infty}^{\infty}\hat\theta_0(t)\eul^{-\ii t}\eul^{\ii t\Lambda/n}\dd t,
\]
where $\hat\theta_0\in\C{S}(\D{R})$. Introducing $P_n^s\coloneqq\eul^{\ii s\Lambda}Q_n\eul^{-\ii s\Lambda}$ we observe that $P_n^0=Q_n$. Then we can write
\begin{equation}  \label{comm.thetan.qn}
\croch{\Theta_n,Q_n}=\int_{-\infty}^{\infty}\hat\theta_0(t)\eul^{-\ii t}\bigl(P_n^{t/n}-P_n^0\bigr)\eul^{\ii t\Lambda/n}\dd t.
\end{equation}
To estimate the norm of this commutator we use the estimate 
\begin{equation}  \label{estim.pntn}
\norm{P_n^{t/n}-P_n^0}\leq\frac{\abs{t}}{n}\sup_{s\in\D{R}}\,\norm{\partial_sP_n^s}
\end{equation}
and now estimate $\norm{\partial_sP_n^s}$. By \eqref{eQe} from Lemma \ref{lem:72} applied to $Q_n=q_n(\Lambda,S)$ with $q_n=\theta_{n,n}\eul^{\ii\tilde\psi_n}$ we get $P_n^s=(q_n\circ\tilde\tau_{-s})(\Lambda,S)=(\theta_{n,n}\eul^{\ii\tilde\psi_n\circ\tilde\tau_{-s}})(\Lambda,S)$, hence
\[
\partial_sP_n^s=q_n^s(\Lambda,S)
\]
with
\[
q_n^s(j,\eul^{\ii\xi})\coloneqq\ii\,\theta_{n,n}(j)\eul^{\ii\tilde\psi_n\circ\tilde\tau_{-s}}\partial_{\xi}\tilde\psi_n(j,\eul^{\ii(\xi+s)}).
\]
By assumptions (i) and (ii), Lemma \ref{lem:75} applies to $q_n^s(\Lambda,S)$. By estimate \eqref{eL42} we get
\begin{equation}  \label{estim.pns.der}
\sup_{s\in\D{R}}\,\bigl\lVert\partial_sP_n^s\bigr\rVert\leq 4\sqrt{\ln n}\,\sup_{\abs{j-n}\leq n/3}\norm{\tilde\psi_n(j,\,\cdot\,)}_{\class^2(\cercle)}.
\end{equation}
It suffices now to apply estimates \eqref{estim.pntn} and \eqref{estim.pns.der} in the integral representation \eqref{comm.thetan.qn}.
\end{proof}

\section{Approximation of $\eul^{\ii B_n}$} \label{sec:8}
\subsection{Plan of Section~\ref{sec:8}}  \label{sec:81}

Proposition~\ref{prop:8} shows one can construct a good approximation of $\eul^{\ii B_n}\Theta_n$ by an operator of the form $q_n(\Lambda,S)$. This proposition is stated in Section~\ref{sec:82}. Its proof is given in Section~\ref{sec:84} and uses an auxiliary computation developed in Section~\ref{sec:83}.
 
\subsection{Main result}  \label{sec:82}

Let $n\geq 1$. Recall that $B_n=\ii\left({a_n(\Lambda)S^{-1}-Sa_n(\Lambda)}\right)\in\C{B}(l^2(\D{Z}))$, see \eqref{Bn''} and \eqref{2B}. $\Theta_n\coloneqq\theta_{n,n}(\Lambda)$ is still as in \eqref{op.theta.n}. Then we introduce $Q_n\in\C{B}(l^2(\D{Z}))$ and $\tilde\psi_n\in\C{Q}^0$ defined by
\begin{subequations} \label{pr41}
\begin{align} \label{pr41a} 
Q_n&\coloneqq\bigl(\theta_{n,n}\eul^{\ii\tilde\psi_n}\bigr)(\Lambda,S)=\Theta_n\,\eul^{\ii\tilde\psi_n}(\Lambda,S)\\
\label{tpf0}
\tilde\psi_n(j,\eul^{\ii\xi})&\coloneqq 2a_n(j)\sin\xi\,\bigl(1-\delta a(n)\cos\xi\bigr).
\end{align}
\end{subequations} 
The operators $Q_n$ are of finite rank. Moreover, by \eqref{an+}, \begin{equation} \label{41an}
a_n(j)=a(n)+(j-n)\delta a(n) 
\end{equation} 
for $\abs{j-n}\leq n/3$. Then we can write
\begin{subequations} \label{tpff'}
\begin{equation} \label{tpf}
\tilde\psi_n(j,\eul^{\ii\xi})\coloneqq\psi_n(\eul^{\ii\xi})+(j-n)\varphi_n(\eul^{\ii\xi})\text{ for }\abs{j-n}\leq n/3
\end{equation} 
with 
\begin{equation}  \label{tpf'}
\begin{cases} 
\psi_n(\eul^{\ii\xi})\coloneqq 2a(n)\sin\xi\,\bigl( 1-\delta a(n)\cos\xi\bigr),&\\
\varphi_n(\eul^{\ii\xi})\coloneqq 2\delta a(n)\sin\xi\,\bigl( 1-\delta a(n)\cos\xi\bigr).&  
\end{cases} 
\end{equation}
\end{subequations} 
By (H2), $a(n)=\ord(n^{\gamma})$ and $\delta a(n)=\ord(n^{\gamma-1})$ with $0<\gamma\leq 1/2$. Thus, for any $m\in\D{N}$ we have
\begin{subequations}  \label{tpsin.phin}
\begin{align}   \label{psin.estim}
\norm{\psi_n}_{\class^m(\cercle)}&=\ord(n^{\gamma})\\
\label{phin.estim}
\norm{\varphi_n}_{\class^m(\cercle)}&=\ord(n^{\gamma-1})\\
\label{tpsin.estim}
\sup_{\abs{j-n}\leq n/3}\norm{\tilde\psi_n(j,\,\cdot\,)}_{\class^m(\cercle)}&=\ord(n^{\gamma}).
\end{align}
\end{subequations}
Let us note that these $\tilde\psi_n$ satisfy properties \eqref{pf'Cmf} from Section~\ref{sec:751}: see indeed \eqref{tpff'} and \eqref{phin.estim}.

\begin{proposition}[approximation of $\eul^{\ii B_n}\Theta_n$ by $Q_n$]\label{prop:8} 
Let $B_n$ be given by \eqref{Bn''} and let $Q_n$ be defined by \eqref{pr41}. Then the difference $R_n\coloneqq\eul^{\ii B_n}\Theta_n-Q_n$ satisfies
\[ 
\norm{R_n}=\ord(n^{\gamma-1}\sqrt{\ln n}).
\] 
\end{proposition} 

\begin{proof} 
See Section \ref{sec:84}.     
\end{proof}

\subsection{An auxiliary computation}  \label{sec:83}

For $0\leq t\leq 1$ we define $\tilde\psi_n^t$ by
\begin{subequations} \label{pf12}
\begin{equation} \label{tpft}
\tilde\psi_n^t(j,\eul^{\ii\xi})\coloneqq 2a_n(j)t\sin\xi\,\bigl(1-t\delta a(n)\cos\xi\bigr).
\end{equation} 
By \eqref{41an}, for $\abs{j-n}\leq\frac{n}{3}$, we can also write\begin{equation} \label{pf1}
\tilde\psi_n^{\,t}(j,\eul^{\ii\xi})=\psi_n^{\,t}(\eul^{\ii\xi})+(j-n)\varphi_n^t(\eul^{\ii\xi})
\end{equation} 
with 
\begin{align} \label{pf2a}  
\psi_n^{\,t}(\eul^{\ii\xi})&\coloneqq 2t\,a(n)\sin\xi\,(1-t\delta a(n)\cos\xi),\\ 
\label{pf2b}
\varphi_n^t(\eul^{\ii\xi})&\coloneqq 2t\,\delta a(n)\sin\xi\,(1-t\delta a(n)\cos\xi).
\end{align}
\end{subequations} 
Thus, if $j,\,j+1\in\croch{2n/3,4n/3}$ we have the relation
\begin{equation} \label{pf3} 
\varphi_n^t(\eul^{\ii\xi})=\tilde\psi_n^{\,t}(j+1,\eul^{\ii\xi})-\tilde\psi_n^{\,t}(j,\eul^{\ii\xi}).
\end{equation} 
Using $a(n)=\ord(n^{\gamma})$ and $\delta a(n)=\ord(n^{\gamma-1})$ we find that for $m\in\D{N}$ there exists $C_m>0$ such that 
\begin{subequations} \label{Cmm'}
\begin{align} \label{Cm} 
&\sup_{0\leq t\leq 1}\norm{\psi_n^t}_{\class^m(\cercle)}\leq C_mn^{\gamma},\\
\label{Cm'} 
&\sup_{0\leq t\leq 1}\norm{\varphi_n^t}_{\class^m(\cercle)}\leq C_mn^{\gamma-1}.
\end{align}
\end{subequations}

\begin{lemma} \label{lem:83} 
Let $\tilde\psi_n^t$ and $\varphi_n^t$ be as in \eqref{pf12} for $0\leq t\leq 1$. Then we can write 
\begin{equation}  \label{eq:rntd}
a_n(j)\Im\bigl(2\eul^{\ii\varphi_n^t(\eul^{\ii\xi})-\ii\xi}\bigr)+\partial_t\tilde\psi_n^{\,t}(j,\eul^{\ii\xi})=a_n(j)\,r_n^{\,t}(\eul^{\ii\xi})
\end{equation} 
with $r_n^{\,t}\colon\cercle\to\D{R}$ satisfying $\sup_{0\leq t\leq 1}\norm{r_n^{\,t}}_{\class^0(\cercle)}=\ord(n^{2(\gamma-1)})$.
\end{lemma}

\begin{proof} 
By differentiation of \eqref{tpft} we get 
\[ 
\partial_t\tilde\psi_n^{\,t}(j,\eul^{\ii\xi})=2a_n(j)\sin\xi\left(1-2t\delta a(n)\cos\xi\right)
\] 
for $j\in\D{Z}$ and $\xi\in\D{R}$. So we can actually write
\[
a_n(j)\Im\bigl(2\eul^{\ii\varphi_n^t(\eul^{\ii\xi})-\ii\xi}\bigr)+\partial_t\tilde\psi_n^{\,t}(j,\eul^{\ii\xi})=a_n(j)r_n^{\,t}(\eul^{\ii\xi})
\] 
with
\begin{equation} \label{rnt}
r_n^t(\eul^{\ii\xi})\coloneqq\Im\bigl(2\eul^{\ii\varphi_n^t(\eul^{\ii\xi})}\eul^{-\ii\xi}\bigr)+2\sin\xi\bigl(1-2t\delta a(n)\cos\xi\bigr).
\end{equation}
It remains to estimate $\norm{r_n^t}_{\class^0(\cercle)}$. Using \eqref{Cm'} for $m=0$, we have
\[ 
\left\lvert\eul^{\ii\varphi_n^t(\eul^{\ii\xi})}-1-\ii\varphi_n^t(\eul^{\ii\xi})\right\rvert\leq 2\bigl\lvert\varphi_n^t(\eul^{\ii\xi})\bigr\rvert^2=\ord(n^{2(\gamma-1)}),
\]
uniformly in $t\in\croch{0,1}$ and $\xi\in\D{R}$. Hence,
\[
\eul^{\ii\varphi_n^t(\eul^{\ii\xi})}\eul^{-\ii\xi}=(1+\ii\varphi_n^t(\eul^{\ii\xi}))\eul^{-\ii\xi}+\ord(n^{2(\gamma-1)}). 
\]
By \eqref{pf2b} and assumption $\delta a(n)=\ord(n^{\gamma-1})$ from (H2) we have
\[
\varphi_n^t(\eul^{\ii\xi})=2t\delta a(n)\sin\xi+\ord(n^{2(\gamma-1)}),
\]
hence $\eul^{\ii\varphi_n^t(\eul^{\ii\xi})}\eul^{-\ii\xi}=(1+2\ii t\delta a(n)\sin\xi)\eul^{-\ii\xi}+\ord(n^{2(\gamma-1)})$. Thus, 
\[
\Im\bigl(2\eul^{\ii\varphi_n^t(\eul^{\ii\xi})}\eul^{-\ii\xi}\bigr)=-2\sin\xi\bigl(1-2t\delta a(n)\cos\xi\bigr)+\ord(n^{2(\gamma-1)}),
\]
i.e., $r_n^t(\eul^{\ii\xi})=\ord(n^{2(\gamma-1)})$.
\end{proof} 

\subsection{Proof of Proposition \ref{prop:8}} \label{sec:84}

We consider the operators $Q_n^{\,t}\in\C{B}\bigl(l^2(\D{Z})\bigr)$ defined by
\[ 
Q_n^{\,t}\coloneqq q_n^{\,t}(\Lambda,S),
\] 
where $q_n^{\,t}(j,\eul^{\ii\xi})\coloneqq\theta_{n,n}(j)\eul^{\ii\tilde\psi_n^{\,t}(j,\eul^{\ii\xi})}$ with $\tilde\psi_n^{\,t}$ as in \eqref{tpft}. The matrix coefficients of $Q_n^{\,t}$ are given by
\[
Q_n^{\,t}(j,k)=\theta_{n,n}(j)\int_0^{2\pi}x_n^{\,t,\xi}(j)\,\eul^{\ii k\xi}\,\frac{\dd\xi}{2\pi}
\]
with
\begin{equation}  \label{xntxi}
x_n^{\,t,\xi}(j)\coloneqq\eul^{\ii\tilde\psi_n^{\,t}(j,\eul^{\ii\xi})-\ii j\xi}.
\end{equation}
By \eqref{pf12} and \eqref{Cm'} for $m=2$, Lemma \ref{lem:75} applies and gives  
\begin{equation}  \label{normQnt}
\sup_{0\leq t\leq 1}\norm{Q_n^{\,t}}=\ord(\sqrt{\ln n}).
\end{equation}    
Since $Q_n^0=\Theta_n$, we can express
\[ 
Q_n^1-\eul^{\ii B_n}\Theta_n=\int_0^1\partial_t\bigl(\eul^{\ii(1-t)B_n}Q_n^{\,t}\bigr)\,\dd t=\int_0^1\eul^{\ii(1-t)B_n}\left(\partial_t-\ii B_n\right)Q_n^{\,t}\,\dd t,
\]
and it remains to prove  
\begin{equation} \label{normQn}
\sup_{0\leq t\leq 1}\norm{\left(\partial_t-\ii B_n\right)Q_n^{\,t}}=\ord(n^{\gamma-1}\sqrt{\ln n}).
\end{equation}
To prove \eqref{normQn} we first show that $B_n\coloneqq\ii\bigl(a_n(\Lambda)S^{-1}-Sa_n(\Lambda)\bigr)$ can be replaced by 
\[ 
B_n'\coloneqq\ii\,a_n(\Lambda)(S^{-1}-S)=B_n+\ii\croch{S,\,a_n(\Lambda)}. 
\] 
For this purpose we observe that the estimates 
\begin{align*}
&\norm{\croch{S,\,a_n(\Lambda)}}=\norm{\delta a_n(\Lambda)}=\ord(n^{\gamma-1})\\
&\norm{\croch{S^{\pm 1},\,\Theta_n}a_n(\Lambda)}=\ord(n^{\gamma-1})
\end{align*}
imply 
\begin{equation} \label{Bn'}
\norm{B_n\Theta_n-\Theta_nB_n'}=\ord(n^{\gamma-1}).
\end{equation}  
We introduce the operators $\hat Q_n^{\,t}\in\C{B}\bigl(l^2(\D{Z})\bigr)$ defined by   
\[
\hat Q_n^{\,t}=\hat q_n^{\,t}(\Lambda,S)
\]
with $\hat q_n^{\,t}(j,\eul^{\ii\xi})\coloneqq\theta_{3n/2,n}(j)\,\eul^{\ii\hat\psi_n^{\,t}(j,\eul^{\ii\xi})}$. The matrix coefficients of $\hat Q_n^{\,t}$ are given by
\[
\hat Q_n^{\,t}(j,k)=\theta_{3n/2,n}(j)\int_0^{2\pi}x_n^{\,t,\xi}(j)\,\eul^{\ii k\xi}\,\frac{\dd\xi}{2\pi}\,,
\]
with $x_n^{\,t,\xi}(j)$ still given by \eqref{xntxi}. If $\theta_{n,n}(j)\neq 0$ then $\theta_{3n/2,n}(j)=1$, and thus $\theta_{n,n}\theta_{3n/2,n}=\theta_{n,n}$, hence $Q_n^{\,t}=\Theta_n\hat Q_n^{\,t}$ and $B_nQ_n^{\,t}-\Theta_nB_n'\hat Q_n^{\,t}=(B_n\Theta_n-\Theta_nB_n')\hat Q_n^{\,t}$. Lemma \ref{lem:75} applies and gives
\[
\sup_{0\leq t\leq 1}\norm{\hat Q_n^{\,t}}=\ord(\sqrt{\ln n}).
\]
Using this estimate and \eqref{Bn'} we get 
\[
\sup_{0\leq t\leq 1}\norm{B_nQ_n^{\,t}-\Theta_nB_n'\hat Q_n^{\,t}}=\ord(n^{\gamma-1}\sqrt{\ln n}).
\] 
Let us denote $P_n^{\,t}\coloneqq\Theta_nB_n'\hat Q_n^{\,t}$. Thus, instead of \eqref{normQn} it suffices to show the estimate 
\begin{equation} \label{normQn'} 
\sup_{0\leq t\leq 1}\norm{\partial_tQ_n^{\,t}-\ii P_n^{\,t}}=\ord(n^{\gamma-1}\sqrt{\ln n}).
\end{equation}  
Since $P_n^{\,t}=b_n(\Lambda)(S^{-1}-S)\hat Q_n^{\,t}$ with $b_n(j)\coloneqq\ii\,\theta_{n,n}(j)a_n(j)$, we have
\begin{align*}
&P_n^{\,t}(j,k)=b_n(j)\bigl(\hat Q_n^{\,t}(j+1,k)-\hat Q_n^{\,t}(j-1,k)\bigr)\\
&=\ii\,\theta_{n,n}(j)a_n(j)\int_0^{2\pi}\bigl(\theta_{3n/2,n}(j+1)x_n^{\,t,\xi}(j+1)-\theta_{3n/2,n}(j-1)x_n^{\,t,\xi}(j-1)\bigr)\,\eul^{\ii k\xi}\,\frac{\dd\xi}{2\pi}.
\end{align*}
Further on, we assume $n\geq 20$. We then have $\frac{n}{4}-1\geq\frac{n}{5}$, hence $\theta_{3n/2,n}(j\pm 1)=1$ if $\theta_{n,n}(j)\neq 0$. Thus, $\theta_{n,n}(j)\theta_{3n/2,n}(j\pm1)=\theta_{n,n}(j)$ and we can write 
\begin{equation}  \label{Pn}
P_n^{\,t}(j,k)=\ii\,\theta_{n,n}(j)a_n(j)\int_0^{2\pi}\bigl(x_n^{\,t,\xi}(j+1)-x_n^{\,t,\xi}(j-1)\bigr)\,\eul^{\ii k\xi}\,\frac{\dd\xi}{2\pi}\,.
\end{equation}
For $\abs{j-n}\leq\frac{n}{5}$ we have $\abs{j\pm1-n}\leq\frac{n}{3}$, and \eqref{pf3} applies, $x_n^{\,t,\xi}(j\pm1)=x_n^{\,t,\xi}(j)\,\eul^{\pm\ii\varphi_n^t(\eul^{\ii\xi})\mp\ii\xi}$ and
\[
x_n^{\,t,\xi}(j+1)-x_n^{\,t,\xi}(j-1)=\ii\,x_n^{\,t,\xi}(j)\,\Im\bigl(2\eul^{\ii\varphi_n^t(\eul^{\ii\xi})-\ii\xi}\bigr).
\]
Thus, for $\abs{j-n}\leq n/5$, using \eqref{Pn} we can express  
\begin{equation} \label{reste}
(\partial_tQ_n^{\,t}-\ii P_n^{\,t})(j,k)=\theta_{n,n}(j)\int_0^{2\pi}y_n^{\,t,\xi}(j)\,\eul^{\ii k\xi}\,\frac{\dd\xi}{2\pi}
\end{equation} 
with 
\[
y_n^{\,t,\xi}(j)\coloneqq\partial_tx_n^{\,t,\xi}(j)+\ii\,a_n(j)x_n^{\,t,\xi}(j)\,\Im\bigl(2\eul^{\ii\varphi_n^t(\eul^{\ii\xi})-\ii\xi}\bigr).
\]
Using \eqref{eq:rntd} from Lemma~\ref{lem:83} and
\[ 
\partial_tx_n^{\,t,\xi}(j)=\ii\,x_n^{\,t,\xi}(j)\,\partial_t\tilde\psi_n^{\,t}(j,\eul^{\ii\xi})
\] 
we obtain
\[
y_n^{\,t,\xi}(j)=\ii\,x_n^{\,t,\xi}(j)a_n(j)r_n^t(\eul^{\ii\xi})
\]
with $r_n^t$ given by \eqref{rnt}. Let us note that both sides of \eqref{reste} vanish for $\abs{j-n}\geq\frac{n}{5}$. Thus, \eqref{reste} is valid for any $j,k$ and can be written
\begin{align*}
(\partial_tQ_n^{\,t}-\ii P_n^{\,t})(j,k)&=\theta_{n,n}(j)\int_0^{2\pi}\ii\,x_n^{\,t,\xi}(j)a_n(j)\,r_n^t(\eul^{\ii\xi})\,\eul^{\ii k\xi}\,\frac{\dd\xi}{2\pi}\\
&=\ii\,a_n(j)\int_0^{2\pi}q_n^{\,t}(j,\eul^{\ii\xi})\,r_n^t(\eul^{\ii\xi})\,\eul^{\ii(k-j)\xi}\,\frac{\dd\xi}{2\pi}\,.
\end{align*}
By properties \eqref{cf0} and \eqref{cf01} from Section \ref{sec:72} these relations mean that
\begin{equation} \label{fin4}
\partial_tQ_n^{\,t}-\ii P_n^{\,t}=\ii\,a_n(\Lambda)\,Q_n^{\,t}\,r_n^{\,t}(S). 
\end{equation}
Since Lemma \ref{lem:83} ensures $\norm{r_n^{\,t}(S)}=\ord(n^{2(\gamma-1)})$, uniformly in $t$, using \eqref{normQnt} and $\norm{a_n(\Lambda)}=\ord(n^{\gamma})$ we conclude that the norm of  \eqref{fin4} is $\ord(n^{3\gamma-2}\sqrt{\ln n})$, uniformly in $t$. We thus get \eqref{normQn'} since $\gamma\leq\frac{1}{2}$ implies $3\gamma-2\leq\gamma-1$. The proof of Proposition \ref{prop:8} is completed.

\section{Proof of Proposition \ref{prop:6} \textup{(i)}}  \label{sec:9} 
\subsection{Plan of Section~\ref{sec:9}} \label{sec:91} 

To prove the estimate $H_n(t)(j,j)=\ord(n^{-\gamma/2})$, uniformly for $\abs{t}\leq t_0$ and $\abs{j-n}\leq n^{\gamma}$ we first decompose $H_n(t)$ into a sum of components $H_n^{\,\omega,t}$ and prove Lemma~\ref{lem:91} allowing us to replace them by simpler operators $Q_n^{\,\omega,t}$. Then each diagonal entry $Q_n^{\,\omega,t}(j,j)$ can be expressed by an oscillatory integral \eqref{Q1} whose phase $\tilde\psi_n^{\,\omega,t}$ is investigated in Section~\ref{sec:93}. In Section~\ref{sec:94} we estimate this integral through a suitable version of the method of stationary phase. The proof of Proposition \ref{prop:6} (i) is completed in Section~\ref{sec:95}.

\subsection{Approximation of $\BS{H_n(t)}$}  \label{sec:92}
\subsubsection{Decomposition of $H_n(t)$ into components $H_n^{\,\omega,t}$}

Since $v$ is periodic of period $N$ we have 
\[ 
v(k)=\sum_{\omega\in\Omega}c_{\omega}\eul^{\ii\omega k}  
\] 
where $c_{\omega}\in\D{C}$ are constants, $\Omega=\frac{1}{N}2\pi\D{Z}/2\pi\D{Z}=\accol{2\pi j/N:j=0,1,\dots,N-1}$. Since $\av{v}=0$ implies $c_0=0$ we have the decomposition 
\[ 
v(\Lambda)=\sum_{\omega\in\Omega^*}c_{\omega}\eul^{\ii\omega\Lambda}    
\] 
where $\Omega^*=\Omega\setminus\accol{0}$. Let us note that $\eul^{2\pi\ii\Lambda}=I$. Let $v_n=v\,\theta_{n,n}^2$ be as in \eqref{vn}. Thus, $v_n(\Lambda)=(\theta_{n,n}^2v)(\Lambda)=\Theta_n^2\,v(\Lambda)$. Recall that by \eqref{Hnt} we have $H_n(t)\coloneqq\eul^{-\ii tL_{0,n}}(L_n-L_{0,n})\eul^{\ii tL_{0,n}}$. By \eqref{Ln=} and \eqref{Vn'},
\[
L_n-L_{0,n}=\tilde V_n=\eul^{\ii B_n}v_n(\Lambda)\eul^{-\ii B_n},
\]
so we can write and expand $H_n(t)$ as follows: 
\[ 
H_n(t)=\eul^{-\ii tL_{0,n}}\eul^{\ii B_n}\,\Theta_n^2v(\Lambda)\,\eul^{-\ii B_n}\eul^{\ii tL_{0,n}}=\sum_{\omega\in\Omega^*}c_{\omega} H_n^{\,\omega,t}     
\] 
with 
\begin{equation}  \label{hnomt}
H_n^{\,\omega,t}\coloneqq\eul^{-\ii tL_{0,n}}\eul^{\ii B_n}\,\Theta_n^2\eul^{\ii\omega\Lambda}\,\eul^{-\ii B_n}\eul^{\ii tL_{0,n}}. 
\end{equation}

\subsubsection{Approximants $Q_n^{\,\omega,t}$}

We approximate $H_n^{\,\omega,t}$ for large $n$ by
\begin{subequations}  \label{Qtpsiwt}
\begin{equation} \label{Qwt}
Q_n^{\,\omega,t}\coloneqq\eul^{\ii\omega\Lambda}\,\bigl(\theta_{n,n}^2\eul^{\ii\tilde\psi_n^{\,\omega,t}}\bigr)(\Lambda,S),
\end{equation}
where the phase $\tilde\psi_n^{\,\omega,t}\in\C{Q}^0$ is chosen as follows:
\begin{equation} \label{tpsiwt}
\tilde\psi_n^{\,\omega,t}\coloneqq(\tilde\psi_n\circ\tilde\tau_{\omega}-\tilde\psi_n)\circ\tilde\vartheta_n\circ\tilde\tau_t   
\end{equation} 
\end{subequations}
with $\tilde\psi_n$ as in \eqref{tpf0}. We noticed in Section~\ref{sec:82} that these $\tilde\psi_n$ satisfy \eqref{pf'Cmf}. Thus, all constructions and results of Section~\ref{sec:751} apply. In particular $\tilde\vartheta_n$ is defined by \eqref{tildevt} for $n\geq n_0=n_0(\accol{\varphi_n})$.

\subsubsection{Approximation of $H_n^{\,\omega,t}$}

\begin{lemma}[approximation of $H_n^{\,\omega,t}$ by $Q_n^{\,\omega,t}$]\label{lem:91} 
Let $H_n^{\,\omega,t}$ be as in \eqref{hnomt}. If $Q_n^{\,\omega,t}$ is defined for large $n$ by \eqref{Qtpsiwt}, then the difference $R_n^{\,\omega,t}\coloneqq H_n^{\,\omega,t}-Q_n^{\,\omega,t}$ satisfies
\[
\sup_{\abs{t}\leq t_0}\norm{R_n^{\,\omega,t}}=\ord(n^{\gamma-1}\ln n).
\]
\end{lemma}  

\begin{proof} 
We first treat the case $t=0$ in Steps 1--3. The general case is treated in Step 4. 

\begin{stepone*}\sl
Estimate of $R_{n,1}^{\,\omega,0}\coloneqq H_n^{\omega,0}-Q_n\eul^{\ii\omega\Lambda}Q_n^*$.  
\end{stepone*}

Let $Q_n=\bigl(\theta_{n,n}\eul^{\ii\tilde\psi_n}\bigr)(\Lambda,S)$ and $\eul^{\ii B_n}\Theta_n=Q_n+R_n$ be as in Proposition \ref{prop:8}. By \eqref{hnomt},
\begin{align*}
H_n^{\omega,0}&=\eul^{\ii B_n}\Theta_n\eul^{\ii\omega\Lambda}\Theta_n\eul^{-\ii B_n}\\
&=\eul^{\ii B_n}\Theta_n\eul^{\ii\omega\Lambda}(Q_n^*+R_n^*) \\ 
&=(Q_n+R_n)\eul^{\ii\omega\Lambda}Q_n^*+\eul^{\ii B_n}\Theta_n\eul^{\ii\omega\Lambda}R_n^*.
\end{align*}
Thus the difference $R_{n,1}^{\,\omega,0}\coloneqq H_n^{\omega,0}-Q_n\eul^{\ii\omega\Lambda}Q_n^*$ can be written
\[
R_{n,1}^{\,\omega,0}=R_n\eul^{\ii\omega\Lambda}Q_n^*+\eul^{\ii B_n}\Theta_n\eul^{\ii\omega\Lambda}R_n^*.
\] 
Using estimates $\norm{R_n}=\ord(n^{\gamma-1}\sqrt{\ln n})$ from Proposition \ref{prop:8}, $\norm{Q_n}=\ord(\sqrt{\ln n})$ from Lemma \ref{lem:75}, and $\norm{\Theta_n}\leq 1$, we finally get
\[
\norm{R_{n,1}^{\,\omega,0}}\leq\norm{R_n}\,(\norm{Q_n}+1)=\ord(n^{\gamma-1}\ln n). 
\] 

\begin{steptwo*}\sl
Estimate of the difference $R_{n,2}^{\,\omega,0}\coloneqq Q_n\eul^{\ii\omega\Lambda}Q_n^*-\widetilde Q_n^{\,\omega,0}$ where $\widetilde Q_n^{\,\omega,0}\coloneqq\eul^{\ii\omega\Lambda}\,\bigl(\theta_{n,n}\eul^{\ii\tilde\psi_n^{\,\omega,0}}\bigr)(\Lambda,S)\,\Theta_n$.
\end{steptwo*}
 
By Lemma \ref{lem:72},
\[
\eul^{-\ii\omega\Lambda}Q_n\eul^{\ii\omega\Lambda}=\eul^{-\ii\omega\Lambda}\,\bigl(\theta_{n,n}\eul^{\ii\tilde\psi_n}\bigr)(\Lambda,S)\,\eul^{\ii\omega\Lambda}=\bigl(\theta_{n,n}\eul^{\ii\tilde\psi_n\circ\tilde\tau_{\omega}}\bigr)(\Lambda,S).
\] 
Hence,
\[
Q_n\eul^{\ii\omega\Lambda}Q_n^*=\eul^{\ii\omega\Lambda}\,\bigl(\theta_{n,n}\eul^{\ii\tilde\psi_n\circ\tilde\tau_{\omega}}\bigr)(\Lambda,S)\,\bigl((\theta_{n,n}\eul^{\ii\tilde\psi_n})(\Lambda,S)\bigr)^*.  
\] 
Then the composition formula \eqref{com} from Lemma \ref{lem:74} gives 
\begin{equation} \label{com.appl}
Q_n\eul^{\ii\omega\Lambda}Q_n^*=\eul^{\ii\omega\Lambda}\,\bigl(\theta_{n,n}\eul^{\ii\tilde\psi_n^{\,\omega,0}}\bigr)(\Lambda,S)\,\F{p}_n(S)\,\Theta_n
\end{equation}
with $\F{p}_n$ as in \eqref{compo} and $\tilde\psi_n^{\,\omega,0}=(\tilde\psi_n\circ\tilde\tau_{\omega}-\tilde\psi_n)\circ\tilde\vartheta_n$. Using \eqref{com.appl} we find that
\[
R_{n,2}^{\,\omega,0}\coloneqq Q_n\eul^{\ii\omega\Lambda}Q_n^*-\widetilde Q_n^{\,\omega,0}=\eul^{\ii\omega\Lambda}\,\bigl(\theta_{n,n}\eul^{\ii\tilde\psi_n^{\,\omega,0}}\bigr)(\Lambda,S)(\F{p}_n(S)-I)\Theta_n.
\]
Lemma \ref{lem:75} gives the estimate $\norm{\bigl(\theta_{n,n}\eul^{\ii\tilde\psi_n^{\,\omega,0}}\bigr)(\Lambda,S)}=\ord(\sqrt{\ln n})$. Moreover, $\F{p}_n-1=\ord(n^{\gamma-1})$ by \eqref{xin'}. Using also $\norm{\Theta_n}\leq 1$, we finally get 
\[
\norm{R_{n,2}^{\,\omega,0}}=\ord(n^{\gamma-1}\sqrt{\ln n}).
\] 
 
\begin{stepthree*}\sl
Estimate of $R_{n,3}^{\,\omega,0}\coloneqq\widetilde Q_n^{\,\omega,0}-Q_n^{\,\omega,0}$. End of proof of Lemma \ref{lem:91} for $t=0$.  
\end{stepthree*}

We have $R_n^{\,\omega,0}=R_{n,1}^{\,\omega,0}+R_{n,2}^{\,\omega,0}+R_{n,3}^{\,\omega,0}$. To prove Lemma \ref{lem:91} for $t=0$ it remains to estimate
\[
R_{n,3}^{\,\omega,0}\coloneqq\widetilde Q_n^{\,\omega,0}-Q_n^{\,\omega,0}=\eul^{\ii\omega\Lambda}\bigl\lbrack\bigl(\theta_{n,n}\eul^{\ii\tilde\psi_n^{\,\omega,0}}\bigr)(\Lambda,S),\,\Theta_n\bigr\rbrack.
\]
To estimate the commutator we can apply Lemma \ref{lem:76} since
\[
\sup_{\abs{j-n}\leq n/3}\norm{\tilde\psi_n^{\omega,0}(j,\,\cdot\,)}_{\class^2(\cercle)}=\ord(n^{\gamma}).
\]   
Hence, estimate \eqref{commutator.estimate} gives $\norm{R_{n,3}^{\,\omega,0}}=\bigl\lVert\bigl\lbrack\bigl(\theta_{n,n}\eul^{\ii\tilde\psi_n^{\,\omega,0}}\bigr)(\Lambda,S),\,\Theta_n\bigr\rbrack\bigr\rVert=\ord(n^{\gamma-1}\ln n)$.

\begin{stepfour*}\sl
End of proof of Lemma \ref{lem:91} for arbitrary $t$.  
\end{stepfour*}

For this purpose we introduce for $s,t\in\D{R}$
\[
\tilde H_n^{\,\omega,t}(s)\coloneqq\eul^{-\ii t\Lambda}\eul^{-\ii sta_{1,n}(\Lambda)}H_n^{\omega,0}\eul^{\ii sta_{1,n}(\Lambda)}\eul^{\ii t\Lambda}.  
\] 
Since $H_n^{\,\omega,t}=\eul^{-\ii tL_{0,n}}H_n^{\omega,0}\eul^{\ii tL_{0,n}}$ with $L_{0,n}=l_n(\Lambda)=\Lambda+a_{1,n}(\Lambda)$ we find that 
\[
H_n^{\,\omega,t}=\tilde H_n^{\,\omega,t}(1)=\tilde H_n^{\,\omega,t}(0)+\tilde R_n^{\,\omega,t}.
\] 
We first claim that the remainder $\tilde R_n^{\,\omega,t}$ satisfies
\begin{equation} \label{Hn-Hn0}
\sup_{\abs{t}\leq t_0}\norm{\tilde R_n^{\,\omega,t}}=\ord(n^{3\gamma-2}).
\end{equation} 
Indeed, since  
\[
\partial_s\tilde H_n^{\,\omega,t}(s)=\eul^{-\ii t\Lambda}\eul^{-\ii sta_{1,n}(\Lambda)}\croch{\ii H_n^{\omega,0},\,ta_{1,n}(\Lambda)}\,\eul^{\ii sta_{1,n}(\Lambda)}\eul^{\ii t\Lambda},
\] 
it suffices to show 
\[ 
\norm{\croch{H_n^{\omega,0},\,a_{1,n}(\Lambda)}}=\ord(n^{3\gamma-2}). 
\] 
However, $\norm{\croch{S,\,a_{1,n}(\Lambda)}}=\ord(n^{2\gamma-2})$ implies $\norm{\croch{B_n,\,a_{1,n}(\Lambda)}}=\ord(n^{3\gamma-2})$, hence the norm of  
\[
\croch{\eul^{\ii B_n},\,a_{1,n}(\Lambda)}=\int_0^1\eul^{\ii tB_n}\,\croch{\ii B_n,\,a_{1,n}(\Lambda)}\,\eul^{\ii(1-t)B_n}\,\dd t
\] 
is $\ord(n^{3\gamma-2})$ and \eqref{Hn-Hn0} follows. To complete the proof we express 
\[
\tilde H_n^{\,\omega,t}(0)=\eul^{-\ii t\Lambda}Q_n^{\,\omega,0} \eul^{-\ii t\Lambda}+\tilde R_{0,n}^{\,\omega,t}=\eul^{-\ii t\Lambda}\eul^{\ii\omega\Lambda}\bigl(\theta_{n,n}^2\,\eul^{\ii\tilde\psi_n^{\,\omega,0}}\bigr)(\Lambda,S)\,\eul^{-\ii t\Lambda}+\tilde R_{0,n}^{\,\omega,t}
\]
where the norm of $\tilde R_{0,n}^{\,\omega,t}\coloneqq\eul^{-\ii t\Lambda}R_n^{\,\omega,0}\eul^{\ii t\Lambda}$ is $\ord(n^{\gamma-1}\ln n)$ by Steps 1-3. It suffices to note that 
\[
\eul^{-\ii t\Lambda}\,\eul^{\ii\tilde\psi_n^{\,\omega,0}}(\Lambda,S)\,\eul^{\ii t\Lambda}=\eul^{\ii\tilde\psi_n^{\,\omega,0}\circ\tilde\tau_t}(\Lambda,S)
\] 
by Lemma \ref{lem:72} and that $\tilde\psi_n^{\,\omega,0}\circ\tilde\tau_t=\tilde\psi_n^{\,\omega,t}$ in view of the definition \eqref{tpsiwt}.
\end{proof} 

\subsection{Decomposition of the phase $\BS{\tilde\psi_n^{\,\omega,t}}$}\label{sec:93}

If $\abs{j-n}\leq n/3$ then combining \eqref{tpsiwt} with \eqref{tpf} we can write
\begin{subequations} \label{psiphiomt} 
\begin{equation} \label{psinfwt} 
\tilde\psi_n^{\,\omega,t}(j,\eul^{\ii\eta})=\psi_n^{\,\omega,t}(\eul^{\ii\eta})+(j-n)\varphi_n^{\,\omega,t}(\eul^{\ii\eta})
\end{equation}
with 
\begin{align} \label{psiomt} 
\psi_n^{\,\omega,t}&\coloneqq (\psi_n\circ\tau_{\omega}-\psi_n)\circ\vartheta_n\circ\tau_t=\psi_n^{\,\omega,0}\circ\tau_t\\ 
\label{phiomt}
\varphi_n^{\,\omega,t}&\coloneqq(\varphi_n\circ\tau_{\omega}-\varphi_n)\circ\vartheta_n\circ\tau_t=\varphi_n^{\,\omega,0}\circ\tau_t
\end{align}
\end{subequations} 
where $\psi_n$ and $\varphi_n$ are given by \eqref{tpf'}. In order to estimate more easily the terms
\begin{equation} \label{Q1}
Q_n^{\,\omega,t}(j,j)=\eul^{\ii\omega j}\theta_{n,n}(j)^2\int_0^{2\pi}\eul^{\ii\tilde\psi_n^{\,\omega,t}(j,\eul^{\ii\eta})}\,\frac{\dd\eta}{2\pi} 
\end{equation} 
we consider a special decomposition $\psi_n^{\,\omega,t}=\psi_{n,1}^{\,\omega,t}+\psi_{n,2}^{\,\omega,t}$ whose description is given below.

\subsubsection{Decomposition of $\tilde\psi_n$}

Further on $\psi_n$ and $\varphi_n$ are given by \eqref{tpf'}. We have
\begin{subequations}  \label{psin12}
\begin{align}  \label{psin} 
\psi_n&=\psi_{n,1}+\psi_{n,2}\\
\label{psin1}
\psi_{n,1}(\eul^{\ii\xi})&\coloneqq 2a(n)\sin\xi\\ 
\label{psin2}
\psi_{n,2}(\eul^{\ii\xi})&\coloneqq-2a(n)\delta a(n)\sin\xi\cos\xi.
\end{align}
\end{subequations}        
This decomposition allows us to write \eqref{tpf} as
\begin{equation}  \label{tpsi} 
\tilde\psi_n(j,\eul^{\ii\xi})=\psi_{n,1}(\eul^{\ii\xi})+\psi_{n,2}(\eul^{\ii\xi})+(j-n)\varphi_n(\eul^{\ii\xi})\text{ for }\abs{j-n}\leq\frac{n}{3}\,.
\end{equation} 
Let us note that $\tilde\psi_n$ reduces to $\psi_{n,1}$ if $\delta a(n)=0$. Hence we call $\psi_{n,1}$ the ``principal part'' of $\tilde\psi_n$. Moreover, assumptions $a(n)=\ord(n^{\gamma})$ and $\delta a(n)=\ord(n^{\gamma-1})$ imply that for every $m\in\D{N}$ we have
\begin{subequations} \label{53ab}
\begin{align} \label{53a}
\norm{\psi_{n,1}}_{\class^m(\cercle)}&=\ord(n^{\gamma})\\
\label{53b}
\norm{\psi_{n,2}}_{\class^m(\cercle)}&=\ord(n^{2\gamma-1}).
\end{align}
\end{subequations} 

\subsubsection{Decomposition of $\psi_n^{\,\omega,t}$}

We define the ``principal part'' $\psi_{n,1}^{\,\omega,t}$ of the phase $\tilde\psi_n^{\,\omega,t}$ by
\begin{subequations} \label{psi1wtt'} 
\begin{equation} \label{psi1wt}  
\psi_{n,1}^{\,\omega,t}\coloneqq(\psi_{n,1}\circ\tau_{\omega}-\psi_{n,1})\circ\tau_t.
\end{equation} 
Using \eqref{psi1wt} with \eqref{psin1} we find
\begin{equation} \label{psi1wt'}   
\psi_{n,1}^{\,\omega,t}(\eul^{\ii\xi})=2a(n)\bigl(\sin(\xi-t-\omega)-\sin(\xi-t)\bigr)=-4a(n)\sin\tfrac{\omega}{2}\cos(\xi-t-\tfrac{\omega}{2}).
\end{equation}
\end{subequations} 
If $\delta a(n)=0$ then $\tilde\vartheta=\id_{\D{Z}\times\cercle}$ and $\tilde\psi_n^{\,\omega,t}$ reduces to its principal part $\psi_{n,1}^{\,\omega,t}$. To estimate more easily $\psi_{n,2}^{\,\omega,t}=\psi_n^{\,\omega,t}-\psi_{n,1}^{\,\omega,t}$ we write $\psi_{n,1}^{\,\omega,t}=\psi_{n,1}^{\,\omega,0}\circ\tau_t$ with 
\[
\psi_{n,1}^{\omega,0}\coloneqq\psi_{n,1}\circ\tau_{\omega}-\psi_{n,1}.
\] 
Thus we can decompose $\psi_n^{\omega,0}$ as 
\begin{equation} \label{psiwt} 
\psi_n^{\omega,0}=\psi_{n,1}^{\omega,0}+\psi_{n,2}^{\omega,0}
\end{equation}
where the remaining part is
\begin{equation} \label{psiwt2} 
\psi_{n,2}^{\omega,0}=\psi_{n,1}^{\omega,0}\circ\vartheta_n-\psi_{n,1}^{\omega,0}+(\psi_{n,2}\circ\tau_{\omega}-\psi_{n,2})\circ\vartheta_n. 
\end{equation} 
This decomposition \eqref{psiwt} combined with \eqref{psiomt} allows us to decompose $\psi_n^{\,\omega,t}$ for arbitrary $t$: 
\begin{equation} \label{psinwt} 
\psi_n^{\,\omega,t}=(\psi_{n,1}^{\omega,0}+\psi_{n,2}^{\omega,0})\circ\tau_t=\psi_{n,1}^{\,\omega,t}+\psi_{n,2}^{\,\omega,t}.    
\end{equation}
The principal part $\psi_{n,1}^{\,\omega,t}$ is given by \eqref{psi1wtt'} and $\psi_{n,2}^{\,\omega,t}\coloneqq\psi_{n,2}^{\omega,0}\circ\tau_t$ with $\psi_{n,2}^{\omega,0}$ as in \eqref{psiwt2}. By \eqref{psinfwt} and \eqref{psinwt} we finally get the decomposition
\begin{equation}   \label{tpsinomt}
\tilde\psi_n^{\,\omega,t}(j,\eul^{\ii\xi})=\psi_{n,1}^{\,\omega,t}(\eul^{\ii\xi})+\psi_{n,2}^{\,\omega,t}(\eul^{\ii\xi})+(j-n)\varphi_n^{\,\omega,t}(\eul^{\ii\xi}).
\end{equation} 

\subsubsection{Estimates of $\psi_{n,1}^{\,\omega,t}$, $\psi_{n,2}^{\,\omega,t}$, and $\varphi_n^{\,\omega,t}$}

\begin{lemma} \label{lem:92} 
\emph{(a)}
For every integer $m\geq 1$ there exists a constant $C_m'$ such that 
\begin{equation} \label{dcirc}
\norm{f\circ\vartheta_n-f}_{\class^{m-1}(\cercle)}\leq C_m' n^{\gamma-1}\norm{f}_{\class^m(\cercle)}
\end{equation} 
holds for any $f\in\class^m(\cercle)$. 

\emph{(b)} 
For every $m\geq 0$ there is a constant $C_m$ such that  
\begin{subequations}  \label{L53}
\begin{align} \label{L53a}
&\norm{\psi^{\,\omega,t}_{n,1}}_{\class^m(\cercle)}=\norm{\psi^{\,\omega,0}_{n,1}}_{\class^m(\cercle)}\leq C_mn^{\gamma}\\
\label{L53b}
&\norm{\psi^{\,\omega,t}_{n,2}}_{\class^m(\cercle)}=\norm{\psi^{\,\omega,0}_{n,2}}_{\class^m(\cercle)}\leq C_mn^{2\gamma-1}\leq C_m\\
\label{L53c}
&\norm{\varphi^{\,\omega,t}_n}_{\class^m(\cercle)}=\norm{\varphi^{\,\omega,0}_n}_{\class^m(\cercle)}\leq C_mn^{\gamma-1}.
\end{align} 
\end{subequations} 
\end{lemma} 

\begin{proof} 
(a) 
For $s\in\D{R}$ we define $\vartheta_n^s(\eul^{\ii\eta})\coloneqq\eul^{\ii\eta}\eul^{\ii s\tilde\xi_n(\eul^{\ii\eta})}$, so that $f\circ\vartheta_n-f=f\circ\vartheta_n^1-f\circ\vartheta_n^0$. If $m\geq 1$, there exists a constant $C_m'$ such that, for every $g\in\class^{m-1}(\cercle)$,
\begin{equation} \label{circ}
\sup_{0\leq s\leq 1}\norm{g\circ\vartheta_n^s}_{\class^{m-1}(\cercle)}\leq\tilde C_m\norm{g}_{\class^{m-1}(\cercle)}.
\end{equation}  
Using the chain rule we easily get \eqref{circ} by induction with respect to $m$. Next we introduce $g(\eul^{\ii\eta})\coloneqq\partial_\eta f(\eul^{\ii\eta})$ and observe that 
\begin{equation} \label{ddcirc}
\partial_sf\bigl(\eul^{\ii\eta}\eul^{\ii s\tilde\xi_n(\eul^{\ii\eta})}\bigr)=\tilde\xi_n(\eul^{\ii\eta})g(\eul^{\ii\eta}\eul^{\ii s\tilde\xi_n(\eul^{\ii\eta})}\bigr). 
\end{equation} 
The $\norm{\,\cdot\,}_{\class^{m-1}(\cercle)}$-norm of \eqref{ddcirc} can be estimated by $C_m''\norm{\tilde\xi_n}_{\class^{m-1}(\cercle)}\norm{g}_{\class^{m-1}(\cercle)}$, as follows from \eqref{circ}. The proof is completed using the estimate $\norm{\tilde\xi_n}_{\class^{m-1}(\cercle)}=\ord(n^{\gamma-1})$, which is proven in Lemma~\ref{lem:73} under assumption \eqref{phin.asymptot}.

(b) 
It is clear that \eqref{L53a} follows from the estimate \eqref{53a} of $\psi_{n,1}$ and \eqref{L53c} follows from estimates \eqref{phin.estim} and \eqref{circ}. Then using $a(n)\delta a(n)=\ord(n^{2\gamma-1})$ and the definition of $\psi_{n,2}$ we obtain $\norm{\psi_{n,2}}_{\class^m(\cercle)}=\ord(n^{2\gamma-1})$ and \eqref{circ} ensures 
\begin{equation} \label{L53d} 
\norm{(\psi_{n,2}\circ\tau_{\omega}-\psi_{n,2})\circ\vartheta_n}_{\class^m(\cercle)}\leq C_m'n^{2\gamma-1}. 
\end{equation} 
Moreover, estimate \eqref{dcirc} from Lemma \ref{lem:92} (a) gives 
\begin{equation} \label{L53e}
\norm{\psi_{n,1}^{\omega,0}\circ\vartheta_n-\psi_{n,1}^{\omega,0}}_{\class^{m-1}(\cercle)}\leq C_m'n^{\gamma-1}\norm{\psi_{n,1}^{\omega,0}}_{\class^m(\cercle)}.
\end{equation} 
Hence, to complete the proof of \eqref{L53b} it remains to use \eqref{L53d} and to observe that the right-hand side of \eqref{L53e} can be estimated by $C_m''n^{2\gamma-1}$ due to \eqref{L53a}. 
\end{proof}

\subsection{A stationary phase estimate} \label{sec:94}

\begin{lemma} \label{lem:93} 
For $b\in\class^2(\cercle)$ and $\mu\in\D{R}$, $\mu\neq 0$ denote 
\[ 
\C{J}(b,\mu)\coloneqq\int_{-\pi}^{\pi}\eul^{\ii\mu\cos\eta}b(\eul^{\ii\eta})\,\dd\eta.  
\] 
Then there is a constant $C_0$ such that 
\[ 
\abs{\C{J}(b,\mu)}\leq\frac{C_0}{\abs{\mu}^{1/2}}\Bigl(\norm{b}_{\class^0(\cercle)}+\frac{1}{\abs{\mu}^{1/2}}\norm{b}_{\class^2(\cercle)}\Bigr). 
\] 
\end{lemma} 

\begin{proof}  
Let $\chi_+\in\class^{\infty}(\D{R})$ be real-valued with $\chi_+\equiv 1$ on $\croch{-\pi/2,\pi/2}$ and $\supp\chi_+\subset(-3\pi/4,\,3\pi/4)$. Let $\chi_-\in\class^{\infty}(\D{R})$ be such that $\chi_-(\xi)=1-\chi_+(\xi\mp\pi)$ if $0\leq\pm\xi\leq\pi$ and $0$ otherwise, so that $\supp\chi_-\subset(-\pi/2,\pi/2)$. Let $b_{\pm}(\xi)\coloneqq b(\pm\eul^{\ii\xi})$. Thus, 
\begin{align*}
\C{J}(b,\mu)&=\int_{-\pi}^{\pi}\eul^{\ii\mu\cos\eta}b_+(\eta)\chi_+(\eta)\,\dd\eta+\int_{-\pi}^{\pi}\eul^{\ii\mu\cos\eta}b_+(\eta)\croch{1-\chi_+(\eta)}\,\dd\eta\\
&=\int_{-\pi}^{\pi}\eul^{\ii\mu\cos\eta}b_+(\eta)\chi_+(\eta)\,\dd\eta+\int_{-\pi}^{\pi}\eul^{-\ii\mu\cos\xi}b_-(\xi)\chi_-(\xi)\,\dd\xi,
\end{align*}
where we perform the change of variable $\eta=\xi\pm\pi$ for $0\leq\pm\eta\leq\pi$ to get the last integral. We are thus reduced to the estimate
\[
\left\lvert\int_{-\infty}^{\infty}b_{\pm}(\xi)\eul^{\pm\ii\mu\cos\xi}\chi_{\pm}(\xi)\dd\xi\right\rvert\leq\frac{C}{\abs{\mu}^{1/2}}\norm{b_{\pm}}_{\class^0(\D{R})}+\frac{C}{\abs{\mu}}\norm{b_{\pm}}_{\class^2(\D{R})}
\]
with $b_{\pm}\in\class^2(\D{R})$. If $\abs{\xi}\leq 3\pi/4$ 
then we can write 
\[
b_{\pm}(\xi)=b_{\pm}(0)+q_{\pm}(\xi)\xi=b_{\pm}(0)+\tilde q_{\pm}(\xi)\sin\xi
\]
with $\tilde q_{\pm}(\xi)\coloneqq q_{\pm}(\xi)\xi/\sin\xi$. However the standard stationary phase method ensures 
\[
\left\lvert b_{\pm}(0)\int_{-\infty}^{\infty}\eul^{\pm\ii\mu\cos\xi}\chi_{\pm}(\xi)\,\dd\xi\right\rvert\leq\abs{b_{\pm}(0)}\,C_0\abs{\mu}^{-1/2}. 
\]
Writing $\eul^{\pm\ii\mu\cos\xi}\sin\xi=\pm\frac{\ii}{\mu}\partial_{\xi}\eul^{\pm\ii\mu\cos\xi}$, integration by parts gives
\begin{equation} \label{53}
\int_{-\infty}^{\infty}\tilde q_{\pm}(\xi)\sin\xi\eul^{\pm\ii\mu\cos\xi} \chi_{\pm}(\xi)\,\dd\xi=\pm\frac{\ii}{\mu}\int_{-\infty}^{\infty}\eul^{\pm\ii\mu\cos\xi}\partial_{\xi}(\tilde q_{\pm}\chi_{\pm})(\xi)\,\dd\xi.
\end{equation} 
We finally observe that the R.H.S.\ of \eqref{53} can be estimated by $\frac{C_1}{\abs{\mu}}\norm{b_{\pm}}_{\class^2(\D{R})}$.  
\end{proof} 

\subsection{End of proof of Proposition \ref{prop:6} (i)} \label{sec:95}

We observe that Lemma \ref{lem:91} ensures
\[ 
g_{1,n,j}(t)=\ii\!\sum_{\omega\in\Omega^*}\!c_{\omega}H_n^{\,\omega,t}(j,j)=\ii\!\sum_{\omega\in\Omega^*}\!c_{\omega}Q_n^{\,\omega,t}(j,j)+\ord(n^{\gamma-1}\ln n),
\] 
with $Q_n^{\,\omega,t}(j,j)$ given by \eqref{Q1}. It remains to show 
\[ 
\sup_{\substack{\abs{t}\leq t_0\\\abs{j-n}\leq n^{\gamma}}}\abs{Q_n^{\,\omega,t}(j,j)}\leq Cn^{-\gamma/2}.  
\] 
Using the decomposition \eqref{tpsinomt} of $\tilde\psi_n^{\,\omega,t}(j,\eul^{\ii\xi})$ and the value \eqref{psi1wt'} of $\psi_{n,1}^{\,\omega,t}(\eul^{\ii\xi})$ we can write 
\[ 
Q_n^{\,\omega,t}(j,j)=\eul^{\ii j\omega}\int_0^{2\pi}\eul^{\ii\mu_n^{\omega}\cos(\eta-t-\omega/2)}b_n^{\,\omega,t}(j,\eul^{\ii\eta})\,\frac{\dd\eta}{2\pi}
\] 
with $\mu_n^{\omega}\coloneqq -4a(n)\sin\frac{\omega}{2}$ and $b_n^{\,\omega,t}(j,\eul^{\ii\eta})\coloneqq\theta_{n,n}(j)^2\,\eul^{\ii\psi_{n,2}^{\,\omega,t}(\eul^{\ii\eta})+\ii(j-n)\varphi_n^{\,\omega,t}(\eul^{\ii\eta})}$. By \eqref{L53c} we have 
\begin{equation} \label{L53f}
\sup_{\abs{j-n}\leq n^{\gamma}}\norm{(j-n)\varphi_n^{\,\omega,t}}_{\class^2(\cercle)}\leq C_2n^{2\gamma-1}\leq C_2. 
\end{equation} 
Combining \eqref{L53f} with \eqref{L53b} we obtain 
\[ 
\sup_{\abs{j-n}\leq n^{\gamma}}\norm{b_n^{\,\omega,t}(j,\,\cdot\,)}_{\class^2(\cercle)}\leq C'.
\] 
Performing the change of variable $\xi=\eta-t-\omega/2$ and using Lemma \ref{lem:93} we find 
\[ 
Q_n^{\,\omega,t}(j,j)=\C{J}(b_n^{\,\omega,t}\circ\tilde\tau_{t+\omega/2}(j,\,\cdot\,),\,\mu_n^{\omega})=\ord(\abs{\mu_n^{\omega}}^{-1/2}),
\] 
uniformly with respect to $j\in\croch{n-n^{\gamma},\,n+n^{\gamma}}$. To complete the proof we observe that, using the assumption $a(n)\geq cn^{\gamma}$, $c>0$ from \eqref{H2ak} we can find $c_0>0$ such that 
\[ 
\abs{\mu_n^{\omega}}\geq c_0\,n^{\gamma} 
\] 
holds for any $\omega\in\Omega^*$. Hence $\abs{\mu_n^{\omega}}^{-1/2}=\ord(n^{-\gamma/2})$ and 
\[
Q_n^{\,\omega,t}(j,j)=\ord(n^{-\gamma/2}),
\]
uniformly with respect to $t\in\croch{-t_0,t_0}$ and $j\in\croch{n-n^{\gamma},n+n^{\gamma}}$.

\section{Proof of Proposition \ref{prop:6} \textup{(ii)}}\label{sec:10} 
\subsection{Plan of Section~\ref{sec:10}} \label{sec:101} 

We denote $H_n(\ut)\coloneqq H_n(t_1)\dots H_n(t_{\nu})$ where $\ut=(t_1,\dots,t_{\nu})\in\D{R}^{\nu}$. To prove the estimate 
\[
\int_{-t_0}^{t_0}\abs{H_n(\ut)(j,j)}\,\dd t_{\nu}=\ord(n^{-\gamma/2}),
\]
uniformly for $\abs{t_1},\dots,\abs{t_{\nu-1}}\leq t_0$ and $\abs{j-n}\leq n^{\gamma}$ we proceed as in Section \ref{sec:9}. In Section~\ref{sec:102} we first decompose $H_n(\ut)$ into a sum of components $H_n^{\,\uw,\ut}$, where $\uw=(\omega_1,\dots,\omega_{\nu})\in(\Omega^*)^{\nu}$, then we consider an approximation of $H_n(\ut)$ by operators $Q_n(\ut)$ whose diagonal entries $Q_n(\ut)(j,j)$ can be expressed by means of oscillatory integrals. Their phase functions are constructed in Section~\ref{sec:103} by induction on the number $\nu$ of factors. In Section~\ref{sec:104} we prove that we thus obtain good approximants $Q_n(\ut)$ of $H_n(\ut)$. Finally we complete the proof of Proposition \ref{prop:6} (ii) in Section~\ref{sec:105} estimating $Q_n^{\,\uw,\ut}(j,j)$ by the method of stationary phase.

\subsection{Approximation of $H_n(\ut)$} \label{sec:102} 
\subsubsection{Decomposition of $H_n(\ut)$ into components $H_n^{\,\uw,\ut}$}

For $\nu\geq 1$ and $\ut\in\D{R}^{\nu}$ we can write 
\[
H_n(\ut)=\sum_{\uw\in(\Omega^*)^{\nu}}\!c_{\omega_1}\dots c_{\omega_{\nu}}H_n^{\,\uw,\ut}
\] 
with  
\[
H_n^{\,\uw,\ut}\coloneqq H_n^{\,\omega_1,t_1}\dots H_n^{\,\omega_{\nu},t_{\nu}}.
\] 

\subsubsection{Approximants $Q_n^{\,\uw,\ut}$}

In Section \ref{sec:103} we approximate $H_n^{\,\uw,\ut}$ for large $n$ by
\begin{equation} \label{Qnuwut}
Q_n^{\,\uw,\ut}\coloneqq\,\eul^{\ii\abs{\uw}_1\Lambda}\bigl(\theta_{n,n}^{2\nu}\,\eul^{\ii\tilde\psi_n^{\,\uw,\ut}}\bigr)(\Lambda,S), 
\end{equation} 
where $\abs{\uw}_1\coloneqq\omega_1+\dots+\omega_{\nu}$. The phase $\tilde\psi_n^{\,\uw,\ut}$ will be defined below. Let us note that it suffices to know its values for $\abs{j-n}\leq n/3$.  

We prove Lemma \ref{lem:103} which gives the estimate $\norm{H_n^{\,\uw,\ut}-Q_n^{\,\uw,\ut}}\leq\nu\,n^{\gamma-1+3\varepsilon}$ for $0<\varepsilon\leq 1/8$, $\nu\leq n^{\varepsilon}$ and $n\geq\hat n$. Then it remains to prove that for $\nu\leq n^{\varepsilon}$ one has
\begin{equation} \label{1010} 
\sup_{\substack{\abs{t_1},\dots,\abs{t_{\nu-1}}\leq t_0\\\abs{j-n}\leq n^{\gamma}}}\int_{-t_0}^{t_0}\abs{Q_n^{\,\uw,\ut}(j,j)}\,\dd t_{\nu}\leq Cn^{-\gamma/2}
\end{equation} 
where  
\[ 
Q_n^{\,\uw,\ut}(j,j)=\eul^{\ii j\abs{\uw}_1}\theta_{n,n}(j)^{2\nu}\int_0^{2\pi}\eul^{\ii\tilde\psi_n^{\,\uw,\ut}(j,\eul^{\ii\eta})}\,\frac{\dd\eta}{2\pi}\,.
\] 
As in Section \ref{sec:9} we obtain \eqref{1010} using a stationary phase estimate.
 
\subsubsection{Construction of $Q_n^{\,\uw,\ut}$ by induction on $\nu$}

For $\nu=1$ the operators $Q_n^{\,\omega,t}$ are defined for large $n$ by \eqref{Qtpsiwt} as in Lemma~\ref{lem:91}, with $\tilde\psi_n^{\,\omega,t}$ is given by \eqref{tpsiwt}. For $\nu\geq 2$, writing 
\begin{subequations} \label{induction}
\begin{equation}  \label{inductom}
\uw=(\uw',\omega)\in(\Omega^*)^{\nu-1}\times\Omega^*,\quad\ut=(\ut'\!,t)\in\D{R}^{\nu-1}\times\D{R},
\end{equation}
we have the corresponding factorization
\begin{equation}  \label{inducth}
H_n^{\,\uw,\ut}=H_n^{\,\uw'\!,\ut'}(H_n^{-\omega,t})^*, 
\end{equation}
\end{subequations}
where $H_n^{-\omega,t}=H_n^{2\pi-\omega,t}$.

By \eqref{611'} from Lemma \ref{lem:101} below we have $\norm{\varphi_n^{-\omega,t}}_{\class^1(\cercle)}=\ord(n^{\gamma-1})$. Section~\ref{sec:751} then applies with $\accol{\varphi_n^{-\omega,t}}$ in place of $\accol{\varphi_n}$. We denote $\eta_n^{\,\omega,t}$, $\xi_n^{\,\omega,t}$, $\tilde\xi_n^{\,\omega,t}$, $\F{p}_n^{\,\omega,t}$, $\vartheta_n^{\,\omega,t}$, and $\tilde\vartheta_n^{\,\omega,t}$ the corresponding auxiliary functions. We have $\varphi_n^{-\omega,t}=\varphi_n^{-\omega,0}\circ\tau_t$. Then it is easy to see that $\tilde\xi_n^{\omega,t}=\tilde\xi_n^{\omega,0}\circ\tau_t$, $\F{p}_n^{\omega,t}=\F{p}_n^{\omega,0}\circ\tau_t$, and $\vartheta_n^{\omega,t}=\eul^{\ii t}\vartheta_n^{\omega,0}\circ\tau_t$. With $\vartheta_n^{\omega,t}$ in place of $\vartheta_n$ estimate \eqref{dcirc} from Lemma~\ref{lem:92} reads as
\begin{equation} \label{61circ}
\norm{f\circ\vartheta_n^{\,\omega,t}-f}_{\class^{m-1}(\cercle)}\leq C_m'n^{\gamma-1}\norm{f}_{\class^m(\cercle)}
\end{equation} 
for any $m\geq 0$, $f\in\class^m(\cercle)$, and with $C_m'\leq C_{m+1}'$.

The phase $\tilde\psi_n^{\,\uw,\ut}$ is chosen in $\C{Q}^0$ and such that for $\abs{j-n}\leq n/3$,
\begin{subequations}  \label{psi012}
\begin{equation} \label{psi}
\tilde\psi_n^{\,\uw,\ut}(j,\eul^{\ii\eta})=\psi_n^{\,\uw,\ut}(\eul^{\ii\eta})+(j-n)\varphi_n^{\,\uw,\ut}(\eul^{\ii\eta}),   
\end{equation} 
each component $\psi_n^{\,\uw,\ut}$ and $\varphi_n^{\,\uw,\ut}$ being the sum of two parts
\begin{align} \label{psi1+2}
\psi_n^{\,\uw,\ut}&=\psi_{n,1}^{\,\uw,\ut}+\psi_{n,2}^{\,\uw,\ut},\\ 
\label{f1+2}
\varphi_n^{\,\uw,\ut}&=\varphi_{n,1}^{\,\uw,\ut}+\varphi_{n,2}^{\,\uw,\ut}.   
\end{align}
\end{subequations} 
$\psi_{n,1}^{\,\uw,\ut}$ and $\varphi_{n,1}^{\,\uw,\ut}$ are the ``principal parts''. They are defined in Section~\ref{sec:103} by induction on $\nu$.  

\subsection{Construction of $\BS{\psi_n^{\,\uw,\ut}}$ and $\BS{\varphi_n^{\,\uw,\ut}}$} \label{sec:103}
\subsubsection{The case $\nu=1$}

For $\nu=1$, the principal part $\psi_{n,1}^{\,\omega,t}$ is defined by \eqref{psi1wt} and we choose $\varphi_{n,1}^{\,\omega,t}\coloneqq\varphi_n^{\,\omega,t}$ which is defined by \eqref{phiomt}, so that
\begin{alignat*}{2}
&\psi_n^{\,\omega,t}=(\psi_n\circ\tau_{\omega}-\psi_n)\circ\vartheta_n\circ\tau_t,&\qquad&\varphi_n^{\,\omega,t}=(\varphi_n\circ\tau_{\omega}-\varphi_n)\circ\vartheta_n\circ\tau_t,\\
&\psi_{n,1}^{\,\omega,t}=(\psi_{n,1}\circ\tau_{\omega}-\psi_{n,1})\circ\tau_t,&&\varphi_{n,1}^{\,\omega,t}=\varphi_n^{\,\omega,t},\\
&\psi_{n,2}^{\,\omega,t}=\psi_n^{\,\omega,t}-\psi_{n,1}^{\,\omega,t},&&\varphi_{n,2}^{\,\omega,t}=0,
\end{alignat*}
with $\psi_n$ and $\varphi_n$ given by \eqref{tpf'}, and $\psi_{n,1}(\eul^{\ii\xi})=2a(n)\sin\xi$.

\subsubsection{Principal parts $\psi_{n,1}^{\,\uw,\ut}$ and $\varphi_{n,1}^{\,\uw,\ut}$ for $\nu\geq 2$}

The relation \eqref{inducth} and the composition formula \eqref{com} from Lemma~\ref{lem:74} suggest the induction formulas  
\begin{equation} \label{psifn1uwut} 
\begin{split} 
\psi_{n,1}^{\,\uw,\ut}\coloneqq\psi_{0,n,1}^{\,\uw,\ut}\circ\tau_{\omega}&\quad\text{with}\quad\psi_{0,n,1}^{\,\uw,\ut}\coloneqq\psi_{n,1}^{\,\uw'\!,\ut'}-\psi_{n,1}^{\,-\omega,t}\\ 
\varphi_{n,1}^{\,\uw,\ut}\coloneqq\varphi_{0,n,1}^{\,\uw,\ut}\circ\tau_{\omega}&\quad\text{with}\quad\varphi_{0,n,1}^{\,\uw,\ut}\coloneqq\varphi_{n,1}^{\,\uw'\!,\ut'}-\varphi_{n,1}^{\,-\omega,t}.
\end{split}
\end{equation}  
Using \eqref{psi1wt'} we find that 
\[
-\psi_{n,1}^{\,-\omega,t}(\eul^{\ii(\xi-\omega)})=4a(n)\sin(\tfrac{-\omega}{2})\cos(\xi-\omega-t+\tfrac{\omega}{2}).
\] 
Hence $-\psi_{n,1}^{\,-\omega,t}\circ\tau_{\omega}=\psi_{n,1}^{\,\omega,t}$ and we can also write
\begin{equation} \label{psin1uwut} 
\psi_{n,1}^{\,\uw,\ut}=\psi_{n,1}^{\,\uw'\!,\ut'}\circ\tau_{\omega}+\psi_{n,1}^{\,\omega,t}.
\end{equation} 

\begin{lemma} \label{lem:101} 
Let $\nu\geq 1$ be fixed. Let $\psi_{n,1}^{\,\uw,\ut}$ and $\varphi_{n,1}^{\,\uw,\ut}$ be defined by \eqref{psifn1uwut}. Then: 
\begin{enumerate}[\rm(a)]
\item
For every $m\in\D{N}$ and any $t_0>0$ there exists a constant $C_m$ (independent of $\uw,\,\ut,\,n$) such that   
\begin{subequations} \label{6111'}
\begin{align} \label{611}
\norm{\psi_{n,1}^{\,\uw,\ut}}_{\class^m(\cercle)}&\leq C_m\nu\,n^{\gamma},\\ 
\label{611'}
\norm{\varphi_{n,1}^{\,\uw,\ut}}_{\class^m(\cercle)}&\leq C_m\nu\,n^{\gamma-1} 
\end{align}
\end{subequations}
hold for $\uw\in(\Omega^*)^{\nu}$, $\ut\in\croch{-t_0,\,t_0}^{\nu}$. Moreover, we can assume that $C_m\leq C_{m+1}$ for any $m\in\D{N}$. 
\item
For every $\uw\in(\Omega^*)^{\nu}$, $\ut\in\D{R}^{\nu}$ there exists $\Psi_{\uw,\ut}\in\class^{\infty}(\cercle)$ such that 
\begin{equation} \label{Psi}
\begin{split}
\psi_{n,1}^{\,\uw,\ut}(\eul^{\ii\xi})&=2a(n)\Im\Psi_{\uw,\ut}(\eul^{\ii\xi}),\\
\Psi_{\uw,\ut}(\eul^{\ii\xi})&=\Psi_{\uw,\ut}(1)\eul^{\ii\xi}.
\end{split} 
\end{equation}
\item
If $\nu\geq 2$ there exist real-valued Lebesgue measurable functions $\ut'\to\tau_{\uw,\ut'}$ defined on $\D{R}^{\nu-1}$ and such that for $\ut=(\ut'\!,t)\in\D{R}^{\nu-1}\times\D{R}$ one has  
\[
\abs{\Psi_{\uw,\ut}(1)}\geq\frac{2}{\pi}\,\sin\frac{\pi}{N}\times\abs{t-\tau_{\uw,\ut'}}_{2\pi},  
\] 
where $\abs{s}_{2\pi}=\dist(s,\,2\pi\D{Z})$. Moreover, for $\nu=1$ we have $\abs{\Psi_{\,\omega,t}(1)}\geq 2\sin\frac{\pi}{N}$.  
\end{enumerate}
\end{lemma}  

\begin{proof}  
(a) 
Estimates \eqref{6111'} hold clearly for $\nu=1$, see Section \ref{sec:9}. Let us make the induction assumption that 
\begin{align*}
\norm{\psi_{n,1}^{\,\uw'\!,\ut'}}_{\class^m(\cercle)}&\leq C_m(\nu-1)n^{\gamma}\\ 
\norm{\varphi_{n,1}^{\,\uw'\!,\ut'}}_{\class^m(\cercle)}&\leq C_m(\nu-1)n^{\gamma-1} 
\end{align*} 
hold for a certain $\nu\geq 2$. Then estimates \eqref{6111'} 
follow from 
\begin{align*}
\norm{\psi_{n,1}^{\,\uw,\ut}}_{\class^m(\cercle)}&\leq\norm{\psi_{n,1}^{\,\uw'\!,\ut'}}_{\class^m(\cercle)}+\norm{\psi_{n,1}^{-\omega,t}}_{\class^m(\cercle)},\\ 
\norm{\varphi_{n,1}^{\,\uw,\ut}}_{\class^m(\cercle)}&\leq\norm{\varphi_{n,1}^{\,\uw'\!,\ut'}}_{\class^m(\cercle)}+\norm{\varphi_{n,1}^{-\omega,t}}_{\class^m(\cercle)}. 
\end{align*} 

\smallskip\noindent
(b)
For $\nu=1$, if $\omega\in\Omega$, $t\in\D{R}$, then $\psi_{n,1}^{\,\omega,t}(\eul^{\ii\xi})=2a_n(j)\Im\Psi_{\,\omega,t}(\eul^{\ii\xi})$ holds with
\[
\Psi_{\,\omega,t}(\eul^{\ii\xi})\coloneqq\bigl(\eul^{-\ii\omega}-1\bigr)\eul^{\ii(\xi-t)}.
\]
Let now $\nu\geq 2$. By induction with respect to $\nu$, we assume that for any $\uw'\in(\Omega^*)^{\nu-1}$ and $\ut'\in\D{R}^{\nu-1}$ there exists $\Psi_{\,\uw'\!,\ut'}\in\class^{\infty}(\cercle)$ such that
\begin{equation}  \label{Psi0'}
\begin{split}
\psi_{n,1}^{\,\uw'\!,\ut'}(\eul^{\ii\xi})&=2a(n)\Im\Psi_{\,\uw'\!,\ut'}(\eul^{\ii\xi}),\\
\Psi_{\,\uw'\!,\ut'}(\eul^{\ii\xi})&=\Psi_{\,\uw'\!,\ut'}(1)\eul^{\ii\xi}.
\end{split} 
\end{equation}
If $\uw=(\uw',\omega)\in(\Omega^*)^{\nu-1}\times\Omega^*$ and $\ut=(\ut'\!,t)\in\D{R}^{\nu-1}\times\D{R}$, then \eqref{psin1uwut} ensures 
\[
\psi_{n,1}^{\,\uw,\ut}(\eul^{\ii\xi})=\psi_{n,1}^{\,\uw'\!,\ut'}(\eul^{\ii(\xi-\omega)})+\psi_{n,1}^{\,\omega,t}(\eul^{\ii\xi}),
\]  
and it is clear that relations \eqref{Psi} follow from \eqref{Psi0'} if we define $\Psi_{\uw,\ut}$ by   
\begin{equation} \label{recPsi}
\Psi_{\uw,\ut}(\eul^{\ii\xi})=\Psi_{\,\uw'\!,\ut'}(\eul^{\ii(\xi-\omega)})+\Psi_{\,\omega,t}(\eul^{\ii\xi}). 
\end{equation}

\smallskip\noindent  
(c)
Let $\omega\in\Omega^*$ and $z_{\omega}\coloneqq 1-\eul^{-\ii\omega}\neq 0$. Then $\abs{z_{\omega}}=2\sin\frac{\omega}{2}$. For $\nu=1$, using $\Psi_{\,\omega,t}(1)=-z_{\omega}\eul^{-\ii t}$ and $\omega\in\Omega^*$, we have the lower bound  
\[
\abs{\Psi_{\,\omega,t}(1)}=2\sin\tfrac{\omega}{2}\geq 2\sin\tfrac{\pi}{N}\,.
\]
For $\nu\geq 2$, if $\uw=(\uw',\omega)\in(\Omega^*)^{\nu-1}\times\Omega^*$ and $\ut=(\ut'\!,t)\in\D{R}^{\nu-1}\times\D{R}$, then \eqref{recPsi} ensures
\[
\Psi_{\uw,\ut}(1)=\Psi_{\,\uw'\!,\ut'}(1)\eul^{-\ii\omega}-z_{\omega}\eul^{-\ii t}.   
\]  
Let $\rho_{\uw,\ut'}\coloneqq\abs{\Psi_{\,\uw'\!,\ut'}(1)}^{1/2}\abs{z_{\omega}}^{-1/2}$ and $\tau_{\uw,\ut'}\in\lbrack 0,\,2\pi)$ be such that 
\[
z_{\omega}^{-1}\Psi_{\,\uw'\!,\ut'}(1)\eul^{-\ii\omega}=\rho_{\uw,\ut'}^2\eul^{\ii\tau_{\uw,\ut'}}.
\] 
Using $\Psi_{\uw,\ut}(1)=z_{\omega} \bigl(\rho_{\uw,\ut'}^2\eul^{\ii\tau_{\uw,\ut'}}-\eul^{-\ii t}\bigr)$ and $\abs{z_{\omega}}=2\sin\frac{\omega}{2}$ we can express 
\[
\abs{\Psi_{\uw,\ut}(1)}=2\sin\tfrac{\omega}{2}\times\abs{\rho_{\uw,\ut'}^2-\eul^{-\ii(t-\tau_{\uw,\ut'})}}.   
\]
Since $\sin\frac{\omega}{2}\geq\sin\frac{\pi}{N}$ for $\omega\in\Omega^*$, it remains to prove that 
\[
\abs{\rho^2-\eul^{\ii\tau}}\geq\tfrac{1}{\pi}\,\abs{\tau}_{2\pi}
\]
holds for any $\tau,\rho\in\D{R}$. We distinguish two cases. If $\cos\tau\leq 0$, then $\abs{\rho^2-\eul^{\ii\tau}}\geq 1\geq\abs{\tau}_{2\pi}/\pi$. If $\cos\tau\geq 0$, then $\abs{\rho^2-\eul^{\ii\tau}}\geq\abs{\sin\tau}\geq 2\,\abs{\tau}_{2\pi}/\pi$.
\end{proof}

\subsubsection{Remaining parts}

To define the remaining parts $\psi_{n,2}^{\,\uw,\ut}$ and $\varphi_{n,2}^{\,\uw,\ut}$ we proceed by induction on $\nu$ as in Section \ref{sec:102}. 

\smallskip\noindent
1) 
For $\nu=1$ we already defined $\psi_{n,2}^{\,\omega,t}\coloneqq\psi_n^{\,\omega,t}-\psi_{n,1}^{\,\omega,t}$ and $\varphi_{n,2}^{\,\omega,t}\coloneqq 0$. 

\smallskip\noindent
2) 
For $\nu\geq 2$ we still write $\uw=(\uw',\omega)\in(\Omega^*)^{\nu-1}\times\Omega^*$, $\ut=(\ut'\!,t)\in\D{R}^{\nu-1}\times\D{R}$, and define $\psi_{n,i}^{\,\uw,\ut},\varphi_{n,i}^{\,\uw,\ut}$, $i=1,2$ through $\psi_{0,n,i}^{\,\uw,\ut},\varphi_{0,n,i}^{\,\uw,\ut}$, according to the rules
\begin{subequations}  \label{psifn12uwut}
\begin{equation} \label{psifn11uwut} 
\begin{split} 
\psi_{n,i}^{\,\uw,\ut}&\coloneqq\psi_{0,n,i}^{\,\uw,\ut}\circ\tau_{\omega},\\ 
\varphi_{n,i}^{\,\uw,\ut}&\coloneqq\varphi_{0,n,i}^{\,\uw,\ut}\circ\tau_{\omega}.
\end{split}
\end{equation}
We define the principal parts $\psi_{n,1}^{\,\uw,\ut}$, $\psi_{n,1}^{\,\uw,\ut}$ using \eqref{psifn11uwut} with $i=1$ and 
\begin{equation} \label{psifn01uwut} 
\begin{split} 
\psi_{0,n,1}^{\,\uw,\ut}&\coloneqq\psi_{n,1}^{\,\uw'\!,\ut'}-\psi_{n,1}^{-\omega,t},\\ 
\varphi_{0,n,1}^{\,\uw,\ut}&\coloneqq\varphi_{n,1}^{\,\uw'\!,\ut'}-\varphi_{n,1}^{-\omega,t}.
\end{split}
\end{equation}
We define the remaining parts $\psi_{n,2}^{\,\uw,\ut}$, $\varphi_{n,2}^{\,\uw,\ut}$ using \eqref{psifn11uwut} with $i=2$ and
\begin{equation} \label{psifn2uwut} 
\begin{split} 
\psi_{0,n,2}^{\,\uw,\ut}&\coloneqq(\psi_{0,n,1}^{\,\uw,\ut}\circ\vartheta_n^{\,\omega,t}-\psi_{0,n,1}^{\,\uw,\ut})+(\psi_{n,2}^{\,\uw'\!,\ut'}-\psi_{n,2}^{-\omega,t}),\\ 
\varphi_{0,n,2}^{\,\uw,\ut}&\coloneqq(\varphi_{0,n,1}^{\,\uw,\ut}\circ\vartheta_n^{\,\omega,t}-\varphi_{0,n,1}^{\,\uw,\ut})+\varphi_{n,2}^{\,\uw'\!,\ut'}.
\end{split}
\end{equation}
\end{subequations}
The phase $\tilde\psi_n^{\,\uw,\ut}(j,\eul^{\ii\eta})$ is now defined for $\abs{j-n}\leq n/3$, according to \eqref{psi012}. For those values of $j$,
\[
\tilde\psi_n^{\,\uw,\ut}(j,\eul^{\ii\eta})\coloneqq\psi_n^{\,\uw,\ut}(\eul^{\ii\eta})+(j-n)\varphi_n^{\,\uw,\ut}(\eul^{\ii\eta})
\] 
with $\psi_n^{\,\uw,\ut}\coloneqq\psi_{n,1}^{\,\uw,\ut}+\psi_{n,2}^{\,\uw,\ut}$ and $\varphi_n^{\,\uw,\ut}\coloneqq\varphi_{n,1}^{\,\uw,\ut}+\varphi_{n,2}^{\,\uw,\ut}$.

\begin{lemma} \label{lem:102} 
Let $\psi_{n,i}^{\,\uw,\ut}$ and $\varphi_{n,i}^{\,\uw,\ut}$, $i=1,2$ be defined by \eqref{psifn12uwut} and let $0<\varepsilon\leq 1/8$ be fixed. Then there exist constants $\hat C$ and $\hat n$ (independent of $\uw,\,\ut,n$) such that the estimates  
\begin{subequations}  \label{psiphi2}
\begin{align}  \label{psi2}
\norm{\psi_{n,2}^{\,\uw,\ut}}_{\class^3(\cercle)}&\leq\hat C\nu\,n^{\varepsilon},\\
\label{vphi2}
\norm{\varphi_{n,2}^{\,\uw,\ut}}_{\class^3(\cercle)}&\leq\hat C\nu\,n^{2(\gamma-1)+\varepsilon}
\end{align} 
\end{subequations}
hold for $n\geq\max(\nu^{1/\varepsilon},\hat n)$. 
\end{lemma} 

\begin{proof} 
The proof is by induction on $\nu$. 

\smallskip\noindent
1) 
For $\nu=1$, the first estimate \eqref{psi2} follows from \eqref{L53b} in Lemma \ref{lem:92}. The second \eqref{vphi2} is straightforward since $\varphi_{n,2}^{\,\omega,t}=0$.

\smallskip\noindent
2) 
Let now $\nu\geq 2$ and $\uw=(\uw',\omega)\in(\Omega^*)^{\nu-1}\times\Omega^*$, $\ut=(\ut'\!,t)\in\D{R}^{\nu-1}\times\D{R}$. By induction assumption, 
\begin{equation} \label{62c}
\begin{split}
\norm{\psi_{n,2}^{\,\uw'\!,\ut'}}_{\class^3(\cercle)}&\leq\hat C(\nu-1)n^{\varepsilon},\\
\norm{\varphi_{n,2}^{\,\uw'\!,\ut'}}_{\class^3(\cercle)}&\leq\hat C(\nu-1)n^{2(\gamma-1)+\varepsilon}.
\end{split}
\end{equation}  
By \eqref{611} we have the estimate 
\begin{equation} \label{6s2} 
\norm{\psi_{0,n,1}^{\,\uw,\ut}}_{\class^m(\cercle)}=\norm{\psi_{n,1}^{\,\uw,\ut}}_{\class^m(\cercle)}\leq C_m\nu\,n^{\gamma}.
\end{equation} 
Further on, we assume $\hat C\geq 2C_3$. Therefore  
\[ 
\norm{\psi_{n,2}^{-\omega,t}}_{\class^3(\cercle)}\leq C_3\leq 
\tfrac{1}{2}\hat C\leq\tfrac{1}{2}\hat C\,n^{\varepsilon}
\]
and combining the last estimate with induction assumption \eqref{62c} we get 
\[ 
\norm{\psi_{n,2}^{\uw'\!,t'}-\psi_{n,2}^{-\omega,t}}_{\class^3(\cercle)}\leq\hat C(\nu-\tfrac{1}{2})n^{\varepsilon}. 
\] 
However, using \eqref{61circ}, \eqref{6s2}, and $\nu\leq n^{\varepsilon}$, we obtain \begin{align*}
\norm{\psi_{0,n,1}^{\,\uw,\ut}\circ\vartheta_n^{\,\omega,t}-\psi_{0,n,1}^{\,\uw,\ut}}_{\class^3(\cercle)}&\leq C_4'n^{\gamma-1}\norm{\psi_{0,n,1}^{\,\uw,\ut}}_{\class^4(\cercle)}\\ 
&\leq C_4'C_4\nu\,n^{2\gamma-1}\\
&\leq C_4'C_4n^{\varepsilon}\leq\tfrac{1}{2}\hat Cn^{\varepsilon}
\end{align*}
provided that $\hat C\geq 2C_4'C_4$. Thus $\hat C\geq 2C_4'C_4$ ensures 
\[   
\norm{\psi_{0,n,2}^{\,\uw,\ut}}_{\class^3(\cercle)}\leq\hat C\nu\,n^{\varepsilon}.
\]   
We proceed similarly to obtain the estimates 
\begin{align*} 
\norm{\varphi_{0,n,1}^{\,\uw,\ut}}_{\class^m(\cercle)}&\leq C_m\nu\,n^{\gamma-1},\\
\norm{\varphi_{0,n,2}^{\,\uw,\ut}}_{\class^3(\cercle)}&\leq\hat C\nu\,n^{\varepsilon+2(\gamma-1)}.\qedhere  
\end{align*}   
\end{proof}

In Steps 2 and 3 of the proof of Lemma~\ref{lem:103} below we need to use the following phase functions
\begin{align*}
\tilde\psi_{0,n}^{\,\uw,\ut}&\coloneqq(\tilde\psi_n^{\,\uw'\!,\ut'}-\tilde\psi_n^{-\omega,t})\circ\tilde\vartheta_n^{\,\omega,t},\\
\psi_{0,n}^{\,\uw,\ut}&\coloneqq(\psi_n^{\uw',\ut'}-\psi_n^{-\omega,t})\circ\vartheta_n^{\,\omega,t},\\
\varphi_{0,n}^{\,\uw,\ut}&\coloneqq(\varphi_n^{\uw',\ut'}-\varphi_n^{-\omega,t})\circ\vartheta_n^{\,\omega,t}.
\end{align*}
For $\abs{j-n}\leq n/3$, we clearly have by induction 
\[
\tilde\psi_{0,n}^{\,\uw,\ut}(j,\eul^{\ii\eta})=\psi_{0,n}^{\,\uw,\ut}(\eul^{\ii\eta})+(j-n)\varphi_{0,n}^{\,\uw,\ut}(\eul^{\ii\eta}).
\]
To compare $\tilde\psi_n^{\,\uw,\ut}$ with $\tilde\psi_{0,n}^{\,\uw,\ut}\circ\tilde\tau_{\omega}$ we define extra terms  
\begin{equation} \label{psifn3uwut} 
\begin{split} 
\psi_{0,n,3}^{\,\uw,\ut}&\coloneqq(\psi_{n,2}^{\,\uw'\!,\ut'}-\psi_{n,2}^{-\omega,t})\circ\vartheta_n^{\,\omega,t}-(\psi_{n,2}^{\,\uw'\!,\ut'}-\psi_{n,2}^{-\omega,t}),\\ 
\varphi_{0,n,3}^{\,\uw,\ut}&\coloneqq\varphi_{n,2}^{\,\uw'\!,\ut'}\circ\vartheta_n^{\,\omega,t}-\varphi_{n,2}^{\,\uw'\!,\ut'}.
\end{split}
\end{equation}
From this definition it is clear that \eqref{61circ} and \eqref{psiphi2} ensure
\begin{subequations}  \label{psiphi3}
\begin{align}  \label{psi3}
\norm{\psi_{0,n,3}^{\,\uw,\ut}}_{\class^2(\cercle)}&\leq C_3'\hat C\nu\,n^{\varepsilon+ \gamma-1},\\
\label{phi3}
\norm{\varphi_{0,n,3}^{\,\uw,\ut}}_{\class^2(\cercle)}&\leq C_3'\hat C\nu\,n^{\varepsilon+3(\gamma-1)}
\end{align}
\end{subequations} 

Bringing together the expressions \eqref{psifn01uwut}, \eqref{psifn2uwut} and \eqref{psifn3uwut}, and using the relations $\psi_n^{\,\uw'\!,\ut'}=\psi_{n,1}^{\,\uw'\!,\ut'}+\psi_{n,2}^{\,\uw'\!,\ut'}$, $\psi_n^{-\omega,t}=\psi_{n,1}^{-\omega,t}+\psi_{n,2}^{-\omega,t}$ and $\varphi_n^{\,\uw'\!,\ut'}=\varphi_{n,1}^{\,\uw'\!,\ut'}+\varphi_{n,2}^{\,\uw'\!,\ut'}$, we get
\begin{align*}
\psi_{0,n}^{\,\uw,\ut}=\psi_{0,n,1}^{\,\uw,\ut}+\psi_{0,n,2}^{\,\uw,\ut}+\psi_{0,n,3}^{\,\uw,\ut},\\
\varphi_{0,n}^{\,\uw,\ut}=\varphi_{0,n,1}^{\,\uw,\ut}+\varphi_{0,n,2}^{\,\uw,\ut}+\varphi_{0,n,3}^{\,\uw,\ut}.
\end{align*}
Thus, if $\tilde\psi_{n,3}^{\,\uw,\ut}\coloneqq\tilde\psi_{0,n,3}^{\,\uw,\ut}\circ\tilde\tau_{\omega}$ with $\tilde\psi_{0,n,3}^{\,\uw,\ut}(j,\eul^{\ii\eta})\coloneqq\psi_{0,n,3}^{\,\uw,\ut}(\eul^{\ii\eta})+(j-n)\varphi_{0,n,3}^{\,\uw,\ut}(\eul^{\ii\eta})$, then
\begin{equation}  \label{default}
\tilde\psi_{0,n}^{\,\uw,\ut}\circ\tilde\tau_{\omega}=\tilde\psi_n^{\,\uw,\ut}+\tilde\psi_{n,3}^{\,\uw,\ut}.
\end{equation}

\subsection{Estimate of $\BS{H_n^{\,\uw,\ut}-Q_n^{\,\uw,\ut}}$}  \label{sec:104}

\begin{lemma}[estimate of $H_n^{\,\uw,\ut}-Q_n^{\,\uw,\ut}$] \label{lem:103} 
Let $Q_n^{\,\uw,\ut}$ be the operator defined by \eqref{Qnuwut} for some $\nu\geq 1$, with $\tilde\psi_n^{\,\uw,\ut}$ defined by \eqref{psi012} and \eqref{psifn12uwut}. Let $0<\varepsilon\leq 1/8$ be fixed. The difference $R_n^{\,\uw,\ut}\coloneqq H_n^{\,\uw,\ut}-Q_n^{\,\uw,\ut}$ satisfies the estimate    
\begin{equation} \label{rn}
\norm{R_n^{\,\uw,\ut}}\leq\nu\,n^{\gamma-1+3\varepsilon} 
\end{equation}    
for $n\geq\max(\nu^{1/\varepsilon},\hat n)$. 
\end{lemma} 

\begin{proof} 
The proof is by induction on $\nu$. 

\smallskip\noindent
1) For $\nu=1$, Lemma \ref{lem:91} ensures 
\begin{equation} \label{62}
\norm{R_n^{\,\omega,t}}=\ord(n^{\gamma-1}\ln n).
\end{equation}  
Moreover, Lemma \ref{lem:75} ensures $\norm{Q_n^{\,\omega,t}}=\ord(\sqrt{\ln n})$.

\smallskip\noindent
2) 
Let now $\nu\geq 2$ and $\uw=(\uw'\!,\omega)\in(\Omega^*)^{\nu-1}\times\Omega^*$, $\ut=(\ut'\!,t)\in\D{R}^{\nu-1}\times\D{R}$. By induction assumption, the difference $R_n^{\,\uw'\!,\ut'}\coloneqq H_n^{\,\uw'\!,\ut'}-Q_n^{\,\uw'\!,\ut'}$ satisfies the estimate    
\begin{equation} \label{62r}
\norm{R_n^{\,\uw'\!,\ut'}}\leq(\nu-1)n^{\gamma-1+3\varepsilon}.
\end{equation}    
By \eqref{Qnuwut},
\begin{equation} \label{63}
Q_n^{\,\uw'\!,\ut'}=\eul^{\ii\abs{\uw'}_1\Lambda}\bigl(\theta_{n,n}^{2(\nu-1)}\,\eul^{\ii\tilde\psi_n^{\,\uw'\!,\ut'}}\bigr)(\Lambda,S),
\end{equation}
where $\tilde\psi_n^{\,\uw'\!,\ut'}(j,\eul^{\ii\eta})=(\psi_{n,1}^{\,\uw'\!,\ut'}+\psi_{n,2}^{\,\uw'\!,\ut'})(\eul^{\ii\eta})+(j-n)(\varphi_{n,1}^{\,\uw'\!,\ut'}+\varphi_{n,2}^{\,\uw'\!,\ut'})(\eul^{\ii\eta})$ for $\abs{j-n}\leq n/3$. Then using $\nu\leq n^{\varepsilon}$ and $\gamma-1+2\varepsilon\leq-1/4$ we can estimate 
\[ 
\norm{\varphi_{n,1}^{\,\uw'\!,\ut'}}_{\class^3(\cercle)}+\norm{\varphi_{n,2}^{\,\uw'\!,\ut'}}_{\class^3(\cercle)}\leq(C_3+\hat C)n^{\gamma-1+2\varepsilon}\leq 1/2
\]
for $n\geq\hat n$ so that Lemma \ref{lem:75} applies:
\begin{equation} \label{62q}
\norm{Q_n^{\,\uw'\!,\ut'}}\leq 4\sqrt{\ln n}.
\end{equation}
We will estimate the difference $H_n^{\,\uw,\ut}-Q_n^{\,\uw,\ut}$ assuming $\nu\leq n^{\varepsilon}$ and $n\geq\hat n$:
\begin{enumerate}[\textbullet]
\item
We first compare $H_n^{\,\uw,\ut}=H_n^{\,\uw'\!,\ut'}(H_n^{\,-\omega,t})^*$ with $Q_n^{\,\uw'\!,\ut'}(Q_n^{-\omega,t})^*$.
\item
Then we compare $Q_n^{\,\uw'\!,\ut'}(Q_n^{-\omega,t})^*$ with $Q_n^{\,\uw,\ut}$.
\end{enumerate}
 
\begin{stepone*}\sl
The difference $R_{n,1}^{\,\uw,\ut}\coloneqq H_n^{\,\uw,\ut}-Q_n^{\,\uw'\!,\ut'}(Q_n^{-\omega,t})^*$ satisfies 
\begin{equation} \label{6s1'}
\norm{R_{n,1}^{\,\uw,\ut}}\leq\norm{R_n^{\,\uw'\!,\ut'}}+\tilde C_1 n^{\gamma-1}\ln^{3/2}\!n, 
\end{equation} 
where $\tilde C_1$ is a constant independent of $\uw$, $\ut$, and $n$.  
\end{stepone*}

Using the factorization \eqref{inducth} and the definitions of $R_n^{\,\uw'\!,\ut'}$ and $R_n^{-\omega,t}$ we find that
\[ 
R_{n,1}^{\,\uw,\ut}=R_n^{\,\uw'\!,\ut'}(H_n^{-\omega,t})^*+Q_n^{\,\uw'\!,\ut'}(R_n^{-\omega,t})^*.  
\] 
Hence, using \eqref{62q} and $\norm{H_n^{-\omega,t}}=1$ we get the estimate 
\[ 
\norm{R_{n,1}^{\,\uw,\ut}}\leq\norm{R_n^{\,\uw'\!,\ut'}}+4\sqrt{\ln n}\,\norm{R_n^{-\omega,t}}. 
\] 
It is clear that \eqref{6s1'} follows since $\norm{R_n^{-\omega,t}}=\ord(n^{\gamma-1}\ln n)$ by Lemma~\ref{lem:91}. 
                                    
\begin{steptwo*}\sl
Computation of $Q_n^{\,\uw'\!,\ut'}(Q_n^{-\omega,t})^*$. We claim that
\[
Q_n^{\,\uw'\!,\ut'}(Q_n^{-\omega,t})^*=\eul^{\ii\abs{\uw}_1\Lambda}P_n^{\,\uw,\ut}\,\Theta_n^2\text{ with }P_n^{\,\uw,\ut}\coloneqq\bigl(\theta_{n,n}^{2(\nu-1)}\,\eul^{\ii\tilde\psi_{0,n}^{\,\uw,\ut}\circ\tilde\tau_{\omega}}\bigr)(\Lambda,S)\,\F{p}_n^{\,\omega,t}(S).  
\] 
Here, $\tilde\psi_{0,n}^{\,\uw,\ut}\coloneqq(\tilde\psi_n^{\,\uw'\!,\ut'}-\tilde\psi_n^{-\omega,t})\circ\tilde\vartheta_n^{\,\omega,t}$ as above, $\F{p}_n^{\,\omega,t}\coloneqq\F{p}_{0,n}^{\omega,t}\circ\tau_{\omega}$ and $\F{p}_{0,n}^{\omega,t}(\eul^{\ii\eta})\coloneqq 1+\partial_{\eta}\tilde\xi_n^{\,\omega,t}(\eul^{\ii\eta})$.
\end{steptwo*}

Using \eqref{63}, Lemma \ref{lem:74} at line two below, and Lemma \ref{lem:72} at line three, we indeed have
\begin{align*}
Q_n^{\,\uw'\!,\ut'}(Q_n^{-\omega,t})^*&=\eul^{\ii\abs{\uw'}_1\Lambda}\,\bigl(\theta_{n,n}^{2(\nu-1)}\,\eul^{\ii\tilde\psi_n^{\,\uw'\!,\ut'}}\bigr)(\Lambda,S)\,\bigl((\theta_{n,n}^2\,\eul^{\ii\tilde\psi_n^{-\omega,t}})(\Lambda,S)\bigr)^*\,\eul^{\ii\omega\Lambda}\\
&=\eul^{\ii\abs{\uw'}_1\Lambda}\,\bigl(\theta_{n,n}^{2(\nu-1)}\,\eul^{\ii\tilde\psi_{0,n}^{\,\uw,\ut}}\bigr)(\Lambda,S)\,\F{p}_{0,n}^{\omega,t}(S)\,\eul^{\ii\omega\Lambda}\,\Theta_n^2\\
&=\eul^{\ii\abs{\uw}_1\Lambda}\,\bigl(\theta_{n,n}^{2(\nu-1)}\,\eul^{\ii\tilde\psi_{0,n}^{\,\uw,\ut}\circ\tilde\tau_{\omega}}\bigr)(\Lambda,S)\,\F{p}_n^{\,\omega,t}(S)\,\Theta_n^2.
\end{align*}
Let us note that $\F{r}_{0,n}^{\,\omega,t}\coloneqq\F{p}_{0,n}^{\omega,t}-1$ satisfies
\begin{equation} \label{r0}
\abs{\F{r}_{0,n}^{\,\omega,t}(\eul^{\ii\xi})}\leq\tilde C_0n^{\gamma-1}.
\end{equation}  

\begin{stepthree*}\sl
Approximation of $P_n^{\,\uw,\ut}$ by $\tilde P_n^{\,\uw,\ut}\coloneqq\bigl(\theta_{n,n}^{2(\nu-1)}\,\eul^{\ii\tilde\psi_n^{\,\uw,\ut}}\bigr)(\Lambda,S)$. Estimate of $R_{n,3}^{\,\uw,\ut}\coloneqq P_n^{\,\uw,\ut}-\tilde P_n^{\,\uw,\ut}$.
\end{stepthree*}
 
Denote $r_{n,3}^{\,\uw,\ut}\coloneqq\eul^{\ii\tilde\psi_{n,3}^{\,\uw,\ut}}-1$ and $\F{r}_n^{\,\omega,t}\coloneqq\F{p}_n^{\,\omega,t}-1=\F{r}_{0,n}^{\,\omega,t}\circ\tau_{\omega}$. Using $\tilde\psi_{0,n}^{\,\uw,\ut}\circ\tilde\tau_{\omega}=\tilde\psi_n^{\,\uw,\ut}+\tilde\psi_{n,3}^{\,\uw,\ut}$ (see \eqref{default}), we get $\eul^{\ii\tilde\psi_{0,n}^{\,\uw,\ut}\circ\tilde\tau_{\omega}}=\eul^{\ii\tilde\psi_n^{\,\uw,\ut}}+\eul^{\ii\tilde\psi_n^{\,\uw,\ut}}r_{n,3}^{\,\uw,\ut}$. Thus,
\begin{align*} 
P_n^{\,\uw,\ut}&\coloneqq\bigl(\theta_{n,n}^{2(\nu-1)}\,\eul^{\ii\tilde\psi_{0,n}^{\,\uw,\ut}\circ\tilde\tau_{\omega}}\bigr)(\Lambda,S)\,\F{p}_n^{\,\omega,t}(S)\\
&=(\tilde P_n^{\,\uw,\ut}+\tilde R_n^{\,\uw,\ut})\bigl(I+\F{r}_n^{\,\omega,t}(S)\bigr)
\intertext{with}
\tilde R_n^{\,\uw,\ut}&\coloneqq\bigl(\theta_{n,n}^{2(\nu-1)}\,\eul^{\ii\tilde\psi_n^{\,\uw,\ut}}r_{n,3}^{\,\uw,\ut}\bigr)(\Lambda,S).
\end{align*}
However, using estimates \eqref{psiphi3} of $\norm{\psi_{0,n,3}^{\,\uw,\ut}}_{\class^2(\cercle)}$ and $\norm{\varphi_{0,n,3}^{\,\uw,\ut}}_{\class^2(\cercle)}$ in
\[
\tilde\psi_{n,3}^{\,\uw,\ut}(j,\,\cdot\,)=\psi_{0,n,3}^{\,\uw,\ut}\circ\tau_{\omega}+(j-n)\varphi_{0,n,3}^{\,\uw,\ut}\circ\tau_{\omega} 
\] 
we get 
\[
\norm{r_{n,3}^{\,\uw,\ut}(j,\,\cdot\,)}_{\class^1(\cercle)}\leq\norm{\tilde\psi_{n,3}^{\,\uw,\ut}(j,\,\cdot\,)}_{\class^1(\cercle)}\leq 2C_3'\hat C\nu\,n^{\varepsilon+\gamma-1} 
\] 
for any $j\in\D{Z}$ such that $\abs{j-n}\leq n/3$. Hence 
\[ 
\norm{\tilde R_n^{\,\uw,\ut}}\leq 8\sqrt{\ln n}\,C_3'\hat C\nu\,n^{\varepsilon+\gamma-1}.
\]  
Using moreover $\norm{\tilde P_n^{\,\uw,\ut}}\leq 4\sqrt{\ln n}$ and $\abs{\F{r}_n^{\,\omega,t}(\eul^{\ii\xi})}\leq\tilde C_0n^{\gamma-1}$, we obtain 
\[  
P_n^{\,\uw,\ut}=\tilde P_n^{\,\uw,\ut}+R_{n,3}^{\,\uw,\ut}
\] 
with 
\[
\norm{R_{n,3}^{\,\uw,\ut}}\leq 4\sqrt{\ln n}\,n^{\gamma-1}(\tilde C_0+2C_3'\hat C(1+\tilde C_0)n^{2\varepsilon}).
\] 
 
\begin{stepfour*}\sl
Estimate of $R_{n,4}^{\,\uw,\ut}\coloneqq H_n^{\,\uw,\ut}-\eul^{\ii\abs{\uw}_1\Lambda}\tilde P_n^{\,\uw,\ut}\,\Theta_n^2$.
\end{stepfour*}
 
Using Step 2 and definitions we have $R_{n,4}^{\,\uw,\ut}=R_{n,1}^{\,\uw,\ut}+\eul^{\ii\abs{\uw}_1\Lambda}R_{n,3}^{\,\uw,\ut}\,\Theta_n^2$. By Steps 1 and 3 we get 
\begin{align*} 
\norm{R_{n,4}^{\,\uw,\ut}}&\leq\norm{R_{n,1}^{\,\uw,\ut}}+\norm{R_{n,3}^{\,\uw,\ut}}\\ 
&\leq\norm{R_{n,1}^{\,\uw'\!,\ut'}}+\tilde C_2\sqrt{\ln n}\, n^{\gamma-1}(\ln n+n^{2\varepsilon})  
\end{align*}  
where the constant $\tilde C_2$ depends on the constants $\tilde C_0$ in \eqref{r0}, $\tilde C_1$ in \eqref{6s1'}, $C_4$, $C_4'$ and $\hat C$. Recall that $C_m\leq C_{m+1}$ and $C_m'\leq C_{m+1}'$ for all $m\geq 0$. 
 
\begin{stepfive*}\sl
End of proof of Lemma \ref{lem:103}.
\end{stepfive*}

We have $R_n^{\,\uw,\ut}=H_n^{\,\uw,\ut}-Q_n^{\,\uw,\ut}=R_{n,4}^{\,\uw,\ut}+\eul^{\ii\abs{\uw}_1\Lambda}\tilde P_n^{\,\uw,\ut}\,\Theta_n^2-\eul^{\ii\abs{\uw}_1\Lambda}\Theta_n^2\tilde P_n^{\,\uw,\ut}=R_{n,4}^{\,\uw,\ut}+\eul^{\ii\abs{\uw}_1\Lambda}\croch{\tilde P_n^{\,\uw,\ut},\Theta_n^2}$, then 
\[
\norm{R_n^{\,\uw,\ut}}\leq\norm{R_{n,4}^{\,\uw,\ut}}+\norm{\croch{\tilde P_n^{\,\uw,\ut},\Theta_n^2}} 
\]
Applying Lemma~\ref{lem:76} we obtain  
\[ 
\norm{\croch{\tilde P_n^{\,\uw,\ut},\Theta_n^2}}\leq\tilde C_3\sqrt{\ln n}\,n^{\gamma-1+2\varepsilon}  
\]  
with $\tilde C_3$ depending only on $C_2$ and $\hat C$. We can choose $\hat n$ depending on $\tilde C_2$, $\tilde C_3$, and $\varepsilon$ so that 
\begin{equation} \label{fin51} 
n\geq\hat n\implies\norm{R_n^{\,\uw,\ut}}\leq\norm{R_n^{\,\uw'\!,\ut'}}+n^{3\varepsilon+\gamma-1}. 
\end{equation} 
We complete the proof of \eqref{rn} using \eqref{fin51} and the induction assumption \eqref{62r}.
\end{proof}

\subsection{End of proof of Proposition \ref{prop:6} (ii)} \label{sec:105} 

Using Lemma \ref{lem:101} and taking $\eta_{\uw,\ut}\in[0,\,2\pi)$ such that 
\[
\Psi_{\uw,\ut}(1)=\abs{\Psi_{\uw,\ut}(1)}\,\eul^{\ii\eta_{\uw,\ut}}
\]
we can write 
\[ 
\psi_{n,1}^{\,\uw,\ut}(j,\eul^{\ii\eta})=2a(n)\abs{\Psi_{\uw,\ut}(1)}\,\sin(\eta+\eta_{\uw,\ut}). 
\]
Then the change of variable $\xi=\eta+\pi/2-\eta_{\uw,\ut}$ gives  
\[ 
\int_0^{2\pi}\eul^{\ii\tilde\psi_n^{\,\uw,\ut}(j,\eul^{\ii\eta})}\,\dd\eta=\C{J}(b_n^{\,\uw,\ut}(j,\,\cdot\,),\mu_n^{\,\uw,\ut}) 
\]
where $\C{J}$ is as in Lemma~\ref{lem:93} with
\begin{align*}
b_n^{\,\uw,\ut}&\coloneqq\bigl(\eul^{\ii\psi_{n,2}^{\,\uw,\ut}+\ii\tilde\varphi_n^{\,\uw,\ut}}\bigr)\circ\tilde\tau_{\eta_{\uw,\ut}-\pi/2},\\ 
\mu_n^{\,\uw,\ut}&\coloneqq 2a(n)\,\abs{\Psi_{\uw,\ut}(1)},   
\end{align*}
where $\tilde\varphi_n^{\,\uw,\ut}(j,\,\cdot\,)\coloneqq(j-n)\varphi_n^{\,\uw,\ut}(\,\cdot\,)$. However Lemma \ref{lem:102} ensures  
\[ 
\sup_{\abs{j-n}\leq n^{\gamma}}\norm{b_n^{\,\uw,\ut}(j,\,\cdot\,)}_{\class^2(\cercle)}\leq C'n^{4\varepsilon}  
\] 
and due to Lemma \ref{lem:101} there exists $c_0>0$ such that 
\[ 
\mu_n^{\,\uw,\ut}\geq c_0 n^{\gamma}\abs{t-\tau_{\uw,\ut'}}_{2\pi}\,. 
\]      
Further on, we abbreviate $\C{J}_{n,j}^{\,\uw,\ut}\coloneqq\C{J}(b_n^{\,\uw,\ut}(j,\,\cdot\,),\mu_n^{\,\uw,\ut})$. Since $\abs{\C{J}_{n,j}^{\,\uw,\ut}}\leq 2\pi$ we get 
\[
\left\lvert\int_{2k\pi+\tau_{\uw,\ut'}-n^{-\gamma/2}}^{2k\pi+\tau_{\uw,\ut'}+n^{-\gamma/2}}\C{J}_{n,j}^{\,\uw,\ut}\,\dd t\,\right\rvert\leq 4\pi n^{-\gamma/2}
\]
and it remains to integrate $\C{J}_{n,j}^{\,\uw,\ut}$ over 
\[
\Delta_n\coloneqq\croch{-t_0,t_0}\setminus\,\bigcup_{k\in\D{Z}}\,\croch{2k\pi+\tau_{\uw,\ut'}-n^{-\gamma/2},\,2k\pi+\tau_{\uw,\ut'}+n^{-\gamma/2}}.
\]
However combining Lemma \ref{lem:92} and Lemma \ref{lem:101} we find the estimate 
\[
\sup_{\abs{j-n}\leq n^{\gamma}}\abs{\C{J}_{n,j}^{\,\uw,\ut}}\leq\frac{C}{n^{\gamma/2}\abs{t-\tau_{\uw,\ut'}}_{2\pi}^{1/2}}\left(1+\frac{C'n^{4\varepsilon}}{n^{\gamma/2}\abs{t-\tau_{\uw,\ut'}}_{2\pi}^{1/2}}\right)
\]
and due to $4\varepsilon\leq\gamma/4$ we can estimate 
\[
t\in\Delta_n\implies\frac{n^{4\varepsilon}}{n^{\gamma/2}\abs{t-\tau_{\uw,\ut'}}_{2\pi}^{1/2}}\leq n^{4\varepsilon-\gamma/4}\leq 1. 
\]
Since $t\to\abs{t}^{-1/2}$ is locally integrable on $\D{R}$ we 
complete the proof writing
\[
\sup_{\abs{j-n}\leq n^{\gamma}}\int_{\Delta_n}\abs{\C{J}_{n,j}^{\,\uw,\ut}}\,\dd t\leq\frac{C(1+C')}{n^{\gamma/2}}\int_{-t_0}^{t_0} \frac{\dd t}{\abs{t-\tau_{\uw,\ut'}}_{2\pi}^{1/2}}\leq\frac{C''}{n^{\gamma/2}}\,.
\] 

\section{Proof of Theorem \ref{thm:2}}  \label{sec:11}
\subsection{Plan of Section~\ref{sec:11}} \label{sec:110} 

In Section \ref{sec:5} we introduced operators $L_n$ and explained that Proposition \ref{prop:4} implies Proposition \ref{cor5} (b) \& (c) whereas property (a) is still unproven.

In Section \ref{sec:111} we will prove Proposition \ref{prop:11} which is the basic tool to deduce the asymptotic estimate of Theorem \ref{thm:2} from the trace estimate of Proposition~\ref{prop:5}. More precisely, Proposition \ref{prop:11} allows us to deduce \eqref{cor5a} from \eqref{cor5b}, \eqref{cor5c} and from the trace estimate \eqref{tr0}.

We observe that writing $k=n+j$ in \eqref{G0n} we find 
\begin{equation} \label{G0n'}
\C{G}_n^0=\sum_{j\in\D{Z}}\bigl(\chi(\lambda_{n+j}(L_n)-l(n))-\chi(l_n(n+j)-l(n))\bigr)  
\end{equation}   
with $l(n)\coloneqq l_n(n)$ and in Section \ref{sec:111} we consider expressions of the form \eqref{G0n'} with $\lambda_{n+j}(L_n)$ replaced by $l_n(n+j)+r_n(j)$. 

The proof of Theorem \ref{thm:2} is completed in Section \ref{sec:113}.

\subsection{Comparison of two sequences} \label{sec:111}

In this section we consider two sequences $(l_n(n+j))_{j\in\D{Z}}$ and $(l_n(n+j)+r_n(j))_{j\in\D{Z}}$ where $l_n$ is defined in \eqref{lnj} and where $r_n\colon\D{Z}\to\D{R}$ has the following two properties:
\begin{subequations} \label{11abc}
\begin{equation} \label{11a}
\sup_{{j\in\D{Z}}}\abs{r_n(j+N)-r_n(j)}\leq Cn^{\gamma-1}  
\end{equation}
and 
\begin{equation} \label{11b}
\sup_{j\in\D{Z}}\abs{r_n(j)}\leq\rho_N'  
\end{equation} 
with
\begin{equation} \label{11c}
\begin{split}
\rho_2'&<\frac{1}{2}\,,\\
\rho_N'&<\frac{1}{\pi\sqrt{N}}\text{ when }N\geq 3. 
\end{split}   
\end{equation}
\end{subequations} 
For $\chi\in\C{S}(\D{R})$ we denote
\begin{equation} \label{11d}
\C{G}_n^{\chi}\coloneqq\sum_{j\in\D{Z}}\bigl(\chi(l_n(n+j)+r_n(j)-l(n))-\chi(l_n(n+j)-l(n))\bigr) 
\end{equation} 
where $l(n)=l_n(n)$.

\begin{proposition}\label{prop:11} 
Assume that $r_n\colon\D{Z}\to\D{R}$ satisfies \eqref{11abc} and that 
\begin{equation} \label{11e}
\C{G}_n^{\chi}=\ord(n^{-\gamma/2}\ln n)
\end{equation}   
holds for any $\chi\in\C{S}(\D{R})$ whose Fourier transform has compact support. Then 
\[
r_n(0)=\ord(n^{-\gamma/2}\ln n). 
\] 
\end{proposition} 

\begin{proof} 
Further on, $i=0,\,1$ and we denote 
\[
r_n^i(k)=
\begin{cases} 
r_n(k)&\text{if }i=1,\\
0&\text{if }i=0.
\end{cases} 
\] 
For $m\in\D{Z}$ and $k=0,\dots,N-1$ we denote 
\[ 
\lambda_{n,i}^{m,k}\coloneqq l_n(n+k+mN)+r_n^i(k+mN)  
\] 
and observe that writing $\D{Z}=\accol{k+mN:k=0,\dots,N-1,\,m\in\D{Z}}$ we can express
\[
\C{G}_n^{\chi}=\C{G}_{n,1}^{\chi}-\C{G}_{n,0}^{\chi}
\] 
with
\[
\C{G}_{n,i}^{\chi}\coloneqq\sum_{k=0}^{N-1}\sum_{m\in\D{Z}}\chi(\lambda_{n,i}^{m,k}-l(n)).    
\] 
Next we denote 
\[
\tilde{\C{G}}_{n,i}^{\chi}\coloneqq\sum_{k=0}^{N-1}\sum_{m\in\D{Z}}\chi(\tilde\lambda_{n,i}^{m,k}-l(n)). 
\] 
with
\[
\tilde\lambda_{n,i}^{m,k}\coloneqq l(n)+k+mN+r^{i}_n(k)  
\] 
and claim that for any $\varepsilon>0$ one can estimate 
\begin{equation} \label{1116}
\C{G}_{n,i}^{\chi}-\tilde{\C{G}}_{n,i}^{\chi}=\ord(n^{\gamma-1+\varepsilon}).
\end{equation} 
Indeed, we observe that using \eqref{lnk1} and \eqref{11a} we obtain
\[
\abs{l_n(n+k+mN)-l_n(n)-k-mN}\leq Cn^{\gamma-1}\abs{k+mN}
\]
and 
\[
\abs{r_n^i(k+mN)-r_n^i(k)}\leq Cn^{\gamma-1}\abs{m}
\] 
with 
\[
\lambda_{n,i}^{m,k}-\tilde\lambda_{n,i}^{m,k}=\bigl(l_n(n+k+mN)-l_n(n)-k-mN\bigr)+\bigl(r_n^i(k+mN)-r_n^i(k)\bigr) 
\] 
we obtain the estimate 
\begin{equation} \label{1117}
\sup_{\abs{m}\leq n^{\varepsilon/2}}\abs{\tilde\lambda_{n,i}^{m,k}-\lambda_{n,i}^{m,k}}=\ord(n^{\gamma-1+\varepsilon/2}). 
\end{equation}   
Since $\lambda_{n,i}^{m,k}-l(n)\sim mN$ as $\abs{m}\to\infty$, the fast decay of $\chi$ implies
\begin{align} \label{1118}
\sum_{\abs{m}\geq n^{\varepsilon/2}}\chi(\lambda_{n,i}^{m,k}-l(n))&=\ord(n^{-\infty})\\
\label{1118'}
\sum_{\abs{m}\geq n^{\varepsilon/2}}\chi(\tilde\lambda_{n,i}^{m,k}-l(n))&=\ord(n^{-\infty})  
\end{align}
and it is clear that \eqref{1116} follows from \eqref{1117}, 
\eqref{1118}, and \eqref{1118'}. 
  
For $j=0,\dots,N-1$ let $\chi_j\in\C{S}(\D{R})$ be a function whose Fourier transform has compact support and satisfies
\begin{equation} \label{1119}
\hat\chi_j(2\pi m/N)=N\delta_{m,j}\text{ for }m\in\D{Z}. 
\end{equation} 
Then we can express
\[ 
\tilde{\C{G}}_{n,i}^{\chi_j}=\sum_{k=0}^{N-1}\sum_{m\in\D{Z}}\chi_{n,i}^{j,k}(m)   
\] 
with $\chi_{n,i}^{j,k}(\lambda)\coloneqq\chi_j(\lambda N+k+r_n^i(k))$ and the Poisson summation formula gives
\[ 
\tilde{\C{G}}_{n,i}^{\chi_j}=\sum_{k=0}^{N-1}\sum_{m\in\D{Z}}\hat\chi_{n,i}^{j,k}(2\pi m) 
\]   
with 
\begin{equation}\label{608}
\hat\chi_{n,i}^{j,k}(t)=(2\pi)^{-1}\int_\D{R}\eul^{-\ii t\lambda}\,\chi_{n,i}^{j,k}(\lambda)\,\dd\lambda=\eul^{\ii(k+r_n^i(k))t/N}\,\hat\chi_j(t/N)/N.
\end{equation}  
Due to \eqref{608} and \eqref{1119} we have 
\begin{equation}\label{609}
\sum_{0\leq k\leq N-1}\sum_{m\in\D{Z}}\left(\hat\chi^{j,k}_{n,1}(2\pi m)-\hat\chi^{j,k}_{n,0}(2\pi m)\right)=\sum_{0\leq k\leq N-1}(z_{k+1}(n)^j-w_{k+1}^j) 
\end{equation} 
with $z_{k+1}(n)\coloneqq\eul^{2\pi\ii(k+r^1_n(k))/N}$ and $w_{k+1}(n)\coloneqq\eul^{2\pi\ii k/N}$. We observe that \eqref{11b} ensures 
\begin{equation}\label{69}
\abs{z_{k+1}(n)-w_{k+1}}\leq\frac{2\pi}{N}\abs{r_n(k)}\leq\frac{2\pi}{N}\rho_N'. 
\end{equation} 
Next we introduce $F_j\colon\D{C}^N\to\D{C}$ defined by $F_j(z)\coloneqq(z_1^j+\dots+z_N^j)/j$ where $z=(z_1,\dots,z_N)\in\D{C}^N$ and $j=1,\dots,N$. If $z(n)\coloneqq(z_1(n),\dots,z_N(n))$ and $w\coloneqq(w_1,\dots,w_N)$, combining \eqref{609} and \eqref{1116} with assumption \eqref{11e} we obtain
\begin{equation} \label{610}
j\bigl(F_j(z(n))-F_j(w)\bigr)=\C{G}_n^{\chi_j}+\ord(n^{\varepsilon+\gamma-1})=\ord\bigl(n^{-\gamma/2}\ln n\bigr). 
\end{equation} 
If $F(z)\coloneqq(F_1(z),\dots,F_N(z))\in\D{C}^N$ then $F'(z)=(z_l^{j-1})_{j,l=1}^N$. Introducing     
\[ 
G(z)\coloneqq\int_0^1\bigl(F'(w+t(z-w))-M\bigr)\,\dd t
\] 
with $M\coloneqq F'(w)$ we find $F(z)-F(w)-M(z-w)=G(z)(z-w)$ and
\[ 
z(n)-w=M^{-1}(F(z(n))-F(w))-M^{-1}G(z(n))(z(n)-w). 
\]
We denote $z(n,t)\coloneqq w+t(z(n)-w)$ and we want to estimate 
\[
F'(z(n,t))-M=\bigl(z_l(n,t)^{j-1}-w_l^{j-1}\bigr)_{j,l=1}^N\quad(0\leq t\leq 1).
\] 
However \eqref{69} ensures $\abs{z_l(n,t)^{j-1}-w_l^{j-1}}\leq N\abs{z_l(n,t)-w_l}\leq Nt\abs{z_l(n)-w_l}\leq 2\pi\rho_N't$ and $\norm{F'(z(n,t))-M}\leq 2\pi N\rho_N't$. Thus $\norm{G\bigl(z(n)\bigr)}\leq\pi N\rho_N'$ and
\[
\abs{z(n)-w}\leq\abs{M^{-1}(F(z(n))-F(w))}+\mu_N\abs{z(n)-w}
\] 
holds with $\mu_N\coloneqq\pi N\rho_N'\norm{M^{-1}}$. Since $M^*M=NI$ we find $\norm{M}=\sqrt{N}$ and $M^{-1}=M^*/N$, i.e.\ $\norm{M^{-1}}=1/\sqrt{N}$ and $\mu_N=\pi\sqrt{N}\rho_N'$. Therefore we can estimate
\begin{equation}\label{614}
(1-\mu_N)\abs{z(n)-w}\leq\abs{M^{-1}(F(z(n))-F(w))}\leq C\abs{F(z(n))-F(w)}
\end{equation}
and our choice of $\rho_N'$ ensures $\mu_N<1$ if $N\geq 3$. Hence, for $k=0,\dots,N-1$ we have
\begin{equation}\label{615}
r_n(k)=\ord(\abs{F(z(n))-F(w)}).
\end{equation}
Thus \eqref{615} and \eqref{610} complete the proof when $N\geq 3$. 

If $N=2$ then $(w_1,w_2)=(1,-1)$, 
\[
M=\begin{pmatrix} 1 & 1\\ 1 &-1 \end{pmatrix},\qquad G(z)=\begin{pmatrix} 0 & 0\\ (z_1-1)/2 & (z_2+1)/2 \end{pmatrix},  
\]   
and $\max\accol{\abs{r_n(0)},\abs{r_n(1)}}<\frac{1}{2}\implies\norm{G(z(n))}=\frac{1}{2}\left(\abs{z_1(n)-1}^2+\abs{z_2(n)+1}^2\right)^{1/2}<1$ ensures $\norm{M^{-1}G(z(n))}<\norm{M^{-1}}=1/\sqrt{2}$, i.e.\ \eqref{614} still holds for $N=2$ with $\mu_2<1$. 
\end{proof} 

\subsection{End of the proof of Theorem \ref{thm:2}} \label{sec:113} 

It remains to check that Proposition \ref{prop:11} allows us to deduce \eqref{cor5a} from \eqref{cor5b} and \eqref{cor5c}. More precisely it suffices to check that Proposition \ref{prop:11} applies with 
\begin{equation}  \label{rnj}
r_n(j)\coloneqq\lambda_{n+j}(L_n)-l_n(n+j). 
\end{equation}  
The assumption on $\rho_N$ allows us to choose $\rho_N'>\rho_N$ satisfying \eqref{11c} and \eqref{cor5b} implies 
\[
\sup_{j\in\D{Z}}\,\abs{r_n(j)}\leq\rho_N+Cn^{3\gamma-2},   
\]
hence \eqref{11b} holds for $n\geq n_0$ if $n_0$ is chosen such that  $Cn_0^{3\gamma-2}<\rho'_N-\rho_N$ and \eqref{cor5c} together with \eqref{lnk1} ensures the estimate \eqref{11a}. It remains to observe that in Sections \ref{sec:9} and \ref{sec:10} we proved Proposition \ref{prop:6} which implies Proposition \ref{prop:5}, hence \eqref{11e} holds if $r_n(j)$ is given by \eqref{rnj} and the Fourier transform of $\chi$ has compact support.
     
\section{Proof of Theorem \ref{thm:12}} \label{sec:120}
\subsection{Plan of Section \ref{sec:120}} \label{sec:121} 

In Section~\ref{sec:23} we gave an uncompleted proof that Theorem~\ref{thm:12} follows from Theorem \ref{thm:2}. It remains to complete parts (ii) and (iii) of this proof. In Section~\ref{sec:223} we prove Lemma \ref{lem:23} that states estimates for $a_n(k)$ and $a_n(k)-a(k)$ we used in part (iii) to get estimate \eqref{eq:121}:
\[
a_n(n-1)^2-a_n(n)^2=a(n-1)^2-a(n)^2+\ord(n^{2\gamma-2}). 
\]
Part (ii) of the proof given in Section~\ref{sec:23} is based on

\begin{proposition}[estimate of $\lambda_n(J)-\lambda_n(J_n)$] \label{prop:2}
Let $J$ be as in Theorem \ref{thm:12} with $\scal{v}=0$ and $J_n$ as in Theorem~\ref{thm:2}. Then one has the large $n$ estimate
\[
\lambda_n(J)=\lambda_n(J_n)+\ord(n^{3\gamma-2}).
\]  
\end{proposition}  

Its proof is given in the last three sections. Section~\ref{sec:123} introduces auxiliary operators $\tilde J_n^+$. Section~\ref{sec:124} states a simple form of the approximation result (\cite{BZ3}, Theorem 2.3). The proof is completed in Section~\ref{sec:125}.

\subsection{Estimates} \label{sec:223}

We prove large $n$ estimates of $a_n(k)$ and $a(k)$ for $k=k(n)$, e.g., $k=n-1$. 

\begin{lemma} \label{lem:21}
Under assumptions \emph{(H2)} on $\accol{a(k)}_{k=1}^{\infty}$ there exists a constant $\tilde C>0$ such that
\begin{equation} \label{21de}
\sup_{k\in\D{Z}}\,\abs{\delta^ma_n(k)}\leq\tilde Cn^{\gamma-m},\quad m=0,1,2.
\end{equation} 
\end{lemma}  

\begin{proof}  
By definition \eqref{an} we can write $a_n(k)=a_n^1(k)\theta_{2n,n}(k)$ with
\[
a_n^1(k)\coloneqq a(n)+(k-n)\delta a(n).
\]
Since $a_n(k)=0$ for $\abs{k-n}\geq 2n/5$ we can replace $\sup_{k\in\D{Z}}$ by $\sup_{\abs{k-n}\leq n/2}$. By assumptions (H2), more precisely, by $\abs{a(k)}\leq Ck^{\gamma}$ from \eqref{H2ak} and by \eqref{H2dak}, we get, for $\abs{k-n}\leq n/2$,
\[
\abs{a_n(k)}\leq\abs{a_n^1(k)}\leq\abs{a(n)}+\abs{k-n}\abs{\delta a(n)}\leq Cn^{\gamma}+nC'n^{\gamma-1}/2=(C+C'/2)n^{\gamma}.
\]
That proves \eqref{21de} for $m=0$. For $m=1,2$ we first observe that, for $\vartheta\in\class^2(\D{R})$ we have
\begin{subequations} \label{eq:XXX}
\begin{align} \label{eq:X}
\delta\vartheta(s)&\coloneqq\vartheta(s+1)-\vartheta(s)=\int_0^1\vartheta'(s+s_1)\dd s_1\\
\label{eq:XX}
\delta^2\vartheta(s)&\coloneqq\delta\vartheta(s+1)-\delta\vartheta(s)=\int_0^1\int_0^1\vartheta''(s+s_1+s_2)\dd s_1\dd s_2
\end{align}
\end{subequations}
For $\vartheta(s)=\theta_{2n,n}(s)=\theta_0\left(\frac{s}{2n}-\frac{1}{2}\right)$ we have $\vartheta^{(m)}(s)=(2n)^{-m}\theta_0^{(m)}\bigl(\frac{s}{2n}-\frac{1}{2}\bigr)$. Thus \eqref{eq:XXX} imply
\[
\abs{\delta^m\theta_{2n,n}(k)}\leq C_mn^{-m}
\]
for $m=1,2$, with $C_m\coloneqq 2^{-m}\norm{\theta_0^{(m)}}_{\infty}$. By using $\delta a_n^1(k)=\delta a(n)$ we get
\begin{align*}
\abs{\delta a_n(k)}&\leq\abs{\delta a(n)}+\abs{a_n^1(k)}\,\abs{\delta\theta_{2n,n}(k)}\\
&\leq C'n^{\gamma-1}+(C+C'/2)n^{\gamma}\hat C'n^{-1}\\
&=\bigl(C'+\hat C'(C+C'/2)\bigr)n^{\gamma-1}.
\end{align*}
Using $\delta^2a_n^1(k)=0$ we get
\begin{align*}
\abs{\delta^2a_n(k)}&\leq\abs{2\delta a(n)}\,\abs{\delta\theta_{2n,n}(k+1)}+\abs{a_n^1(k)}\,\abs{\delta^2\theta_{2n,n}(k)}\\
&\leq 2C'n^{\gamma-1}\hat C'n^{-1}+(C+C'/2)n^{\gamma}\hat C''n^{-2}\\
&=\bigl(2C'\hat C'+\hat C''(C+C'/2)\bigr)n^{\gamma-2}.\qedhere
\end{align*}
\end{proof}    

\begin{lemma} \label{lem:22}
Under assumptions \emph{(H2)} we have the estimates
\begin{subequations}  \label{eq:2112}
\begin{align} \label{eq:211}
\sup_{\abs{j}\leq n/2}\abs{\delta a(n+j)}&=\ord(n^{\gamma-1})\\ \label{eq:212}
\sup_{\abs{j}\leq n/2}\abs{\delta^2a(n+j)}&=\ord(n^{\gamma-2}).
\end{align}
\end{subequations}
\end{lemma}  

\begin{proof}  
Let $j\in\D{Z}$ be such that $\abs{j}\leq n/2$. By using \eqref{H2dak} and \eqref{H2ddak}, i.e., $\abs{\delta a(k)}\leq C'k^{\gamma-1}$ and $\abs{\delta^2a(k)}\leq C''k^{\gamma-2}$, respectively, we get 
\begin{align*}
\abs{\delta a(n+j)}&\leq C'(n+j)^{\gamma-1}\leq C'(n-n/2)^{\gamma-1}=C'n^{\gamma-1}/2^{\gamma-1}=\tilde C'n^{\gamma-1},\\
\abs{\delta^2a(n+j)}&\leq C''(n+j)^{\gamma-2}\leq C''(n-n/2)^{\gamma-2}=C''n^{\gamma-2}/2^{\gamma-2}=\tilde C''n^{\gamma-2}.\qedhere
\end{align*}
\end{proof}    

\begin{lemma} \label{lem:23}
Under assumptions \emph{(H2)} we have the estimates
\begin{subequations}  \label{eq:2167}
\begin{align} \label{eq:216}
\abs{k-n}\leq n/2&\implies\abs{a(k)-a(n)}\leq\tilde C\abs{k-n}n^{\gamma-1}\\ \label{eq:217}
\abs{k-n}\leq n/3&\implies\abs{a(k)-a_n(k)}\leq\tilde C\abs{k-n}^2n^{\gamma-2}.
\end{align}
\end{subequations}
\end{lemma}  

\begin{proof} 
It uses Lemma \ref{lem:22} together with the following two estimates:
\begin{subequations}  \label{eq:2189}
\begin{align} \label{eq:218}
\abs{a(k)-a(n)}&\leq\abs{k-n}\sup_{\abs{j}\leq\abs{k-n}}\abs{\delta a(n+j)}\\  \label{eq:219}
\abs{a(k)-a_n^1(k)}&\leq\abs{k-n}^2\sup_{\abs{j}\leq\abs{k-n}}\abs{\delta^2a(n+j)}.
\end{align} 
\end{subequations}
We get \eqref{eq:216} by using \eqref{eq:211} in \eqref{eq:218} for $\abs{k-n}\leq n/2$:
\begin{align*}
\abs{a(k)-a(n)}&\leq\abs{k-n}\sup_{\abs{j}\leq\abs{k-n}}\abs{\delta a(n+j)}\\
&\leq\abs{k-n}\sup_{\abs{j}\leq n/2}\abs{\delta a(n+j)}\\
&\leq\tilde C\abs{k-n}n^{\gamma-1}
\end{align*}
We get \eqref{eq:217} similarly, using \eqref{eq:212} in \eqref{eq:219} for $\abs{k-n}\leq n/3$. We then have $\theta_{2n,n}(k)=1$, hence $a_n(k)=a_n^1(k)$, and
\begin{align*}
\abs{a(k)-a_n(k)}&=\abs{a(k)-a_n^1(k)}\\
&\leq\abs{k-n}^2\sup_{\abs{j}\leq\abs{k-n}}\abs{\delta^2a(n+j)}\\
&\leq\abs{k-n}^2\sup_{\abs{j}\leq n/2}\abs{\delta^2a(n+j)}\\
&\leq\tilde C\abs{k-n}^2n^{\gamma-2}.\qedhere
\end{align*}
\end{proof}    

\subsection{Operators $\BS{\tilde J_n^+}$} \label{sec:123}

These auxiliary operators act on $l^2(\D{N}^*)$ by
\[  
(\tilde J_n^+x)(k)=d_n(k)x(k)+\tilde a_n(k)x(k+1)+\tilde a_n(k-1)x(k-1).  
\] 
for $x\in\C{D}$ and $k\geq 1$ with off-diagonal entries
\begin{equation} \label{240} 
\tilde a_n(k)\coloneqq  
\begin{cases} 
a(k)&\text{if }n-C_1n^{\gamma}\leq k\leq n+C_1n^{\gamma}\\ 
a_n(k)&\text{otherwise}
\end{cases}  
\end{equation}  
where $C_1$ is fixed large enough. We claim that 
\begin{equation}   \label{241}
\norm{\tilde J_n^+-J_n^+}_{\C{B}(l^2(\D{N}^*))}=\ord(n^{3\gamma-2}). 
\end{equation}
Indeed, $\abs{j}\leq n/3$ ensures $a_n(n+j)=a(n)+\delta a(n)j$ and $\abs{\delta^2a(n+j)}\leq\tilde Cn^{\gamma-2}$, hence we can estimate 
\begin{align*}
\sup_{k\in\D{N}^*}\abs{\tilde a_n(k)-a_n(k)}
&=\sup_{\abs{j}\leq C_1n^{\gamma}}\abs{a(n+j)-a(n)-\delta a(n)j}\\
&\leq\sup_{\abs{j}\leq C_1n^{\gamma}}j^2\,\tilde Cn^{\gamma-2}=\ord(n^{3\gamma-2}).
\end{align*}
However, \eqref{241} and the min-max principle give 
\begin{equation}  \label{eq:1232}
\sup_{k\in\D{N}^*}\abs{\lambda_k(\tilde J_n^+)-\lambda_k(J_n^+)}=\ord(n^{3\gamma-2}).
\end{equation}

\subsection{An approximation result of the spectrum of Jacobi matrices} \label{sec:124}

In the next section we apply Theorem 2.3 from \cite{BZ3} which is an approximation result of the spectrum of the operator $J$ defined by \eqref{1J} with real entries $\accol{d(k)}_{k=1}^{\infty}$ and $\accol{a(k)}_{k=1}^{\infty}$ such that
\begin{equation}  \label{eq:1240}
\begin{split}
d(k)&=ck^{\alpha}+\ord(k^{\beta}),\quad c>0,\\
a(k)&=\ord(k^{\beta}),\quad 0\leq\beta<\alpha<1+\beta.
\end{split}
\end{equation}
For simplicity we state this result assuming that \eqref{eq:1240} holds with $c=1$, $\alpha=1$, and $\beta=\gamma$, with $0<\gamma\leq 1/2$. These conditions are satisfied by the operator $J$ from Theorem~\ref{thm:12}. 

For $\lambda\geq 1$ and $\lambda'\leq\lambda$ we denote
\[
\C{N}(\lambda',\lambda,J)=\card\{n\in\D{N}^*:\lambda'<\lambda_n(J)\leq\lambda\}=\card\left(\spec(J)\cap(\lambda',\lambda\rbrack\right)
\] 
and we consider Jacobi operators $J_{\lambda',\lambda}$ defined like $J$ by
\begin{equation}  \label{eq:approx}
(J_{\lambda',\lambda}x)(k)=d(k)x(k)+a_{\lambda',\lambda}(k)x(k+1)+a_{\lambda',\lambda}(k-1)x(k-1)
\end{equation}
for $x\in\C{D}$, $k\geq 1$, with real off-diagonal entries $(a_{\lambda',\lambda}(k))_{k=1}^{\infty}$ satisfying $\abs{a_{\lambda',\lambda}(k)}\leq\abs{a(k)}$. Then Theorem 2.3 from \cite{BZ3} takes the form:

\begin{proposition}[\cite{BZ3}, Theorem 2.3] \label{prop:approx}
Let $J$ be given by \eqref{1J}. Its entries are assumed to satisfy \eqref{eq:1240} with $c=1$, $\alpha=1$, and $\beta=\gamma$. Let $C_0>0$ be large enough. For $\lambda\geq 1$ and $\lambda'\leq\lambda$ we denote
\begin{align*}
\kappa(\lambda)&\coloneqq\lambda+C_0\lambda^{\gamma},\\
\kappa(\lambda',\lambda)&\coloneqq\lambda'-C_0\lambda^{\gamma},
\end{align*}
and $J_{\lambda',\lambda}$ an operator as in \eqref{eq:approx} with $a_{\lambda',\lambda}(k)=a(k)$ if $\kappa(\lambda',\lambda)\leq k\leq\kappa(\lambda)$.

Then for any $\nu>0$ there exists $\lambda(\nu)>0$ such that 
\[
\C{N}(\lambda'+\lambda^{-\nu},\lambda-\lambda^{-\nu},J_{\lambda',\lambda})\leq\C{N}(\lambda',\lambda,J)\leq\C{N}(\lambda'-\lambda^{-\nu},\lambda+\lambda^{-\nu},J_{\lambda',\lambda})
\]
for any $\lambda\geq\lambda(\nu)$ and any $\lambda'$ such that $(C_0+1)\lambda^{\gamma}\leq\lambda'\leq\lambda$.
\end{proposition}  

\subsection{Proof of Proposition~\ref{prop:2}} \label{sec:125}

Due to \eqref{eq:1232} it suffices to show that 
\[
\lambda_n(J)=\lambda_n(\tilde J_n^+)+\ord(n^{-\nu}) 
\] 
holds for any $\nu>0$ provided that $C_1$ is chosen large enough in \eqref{240}.
 
Let $\accol{\lambda_n}_{n=1}^\infty$ be a real sequence satisfying $\lambda_n=n+\ord(n^{\gamma})$. If $\kappa(\lambda)$ and $\kappa(\lambda',\lambda)$ are as in Proposition~\ref{prop:approx} then choosing $C_1$ large enough in \eqref{240} we get that 
\[ 
\kappa(\lambda_n,\lambda_n-\lambda_n^{-\nu})\leq k\leq\kappa(\lambda_n)\implies J\vece_k=\tilde J_n^+\vece_k   
\] 
for $n\geq n_0$. Proposition~\ref{prop:approx} applied with $\lambda=\lambda_n$, $\lambda'=\lambda_n-\lambda_n^{-\nu}$, $J_{\lambda,\lambda'}=\tilde J_n^+$:  
\begin{equation} \label{692} 
\card\left(\spec(J)\cap(\lambda_n-\lambda_n^{-\nu},\lambda_n]\right)\leq\card\left(\spec(\tilde J_n^+)\cap\bigl(\lambda_n-2\lambda_n^{-\nu},\,\lambda_n+\lambda_n^{-\nu}\bigr]\right)
\end{equation} 
for $n\geq n_0$. However, $\lambda_n(\tilde J_n^+)=\lambda_n(J_n^+)+\ord(n^{3\gamma-2})$ implies 
\begin{equation} \label{693} 
n\geq n_0\implies\abs{\lambda_n(\tilde J_n^+)-l(n)}\leq\rho', 
\end{equation}  
where $\rho'\in(\rho_N,\frac{1}{2})$, hence the cardinal in the right-hand side of \eqref{692} is at most $1$. If now $\lambda_n=\lambda_n(J)$ then both cardinals are equal to $1$ and there is an eigenvalue $\lambda_{k(n)}(\tilde J_n^+)$ such that 
\begin{equation} \label{694} 
\lambda_{k(n)}(\tilde J_n^+)\in\bigl(\lambda_n(J)-2\lambda_n(J)^{-\nu},\lambda_n(J)+\lambda_n(J)^{-\nu}\bigr].
\end{equation} 
It remains to check that $k(n)=n$. Due to \eqref{693} and \eqref{694} it suffices to know that 
\begin{equation} \label{695} 
n\geq n_0\implies\lambda_n(J)\in(l(n)-\rho'',\,l(n)+\rho'') 
\end{equation} 
for some $\rho''<1/2$. However, the operator $J^0 \coloneqq\Lambda^++2\Re\left(S^+a(\Lambda^+)\right)$ was investigated in \cite{BZ1} where we proved the large $n$ asymptotic formula 
\begin{equation} \label{696} 
\lambda_n(J^0)=l(n)+\ord(n^{3\gamma-2}). 
\end{equation} 
Since \eqref{695} follows from \eqref{696} and $\abs{\lambda_n(J)-\lambda_n(J^0)}\leq\rho_N$, the proof of Proposition \ref{prop:2} is complete.
   


\begin{bibdiv}
\begin{biblist}
\bib{BNS}{article}{
   author={Boutet de Monvel, Anne},
   author={Naboko, Serguei},
   author={Silva, Luis O.},
   title={The asymptotic behavior of eigenvalues of a modified
   Jaynes-Cummings model},
   journal={Asymptot. Anal.},
   volume={47},
   date={2006},
   number={3-4},
   pages={291--315},
} 
\bib{BZ1}{article}{
   author={Boutet de Monvel, Anne},
   author={Zielinski, Lech},
   title={Eigenvalue asymptotics for Jaynes--Cummings type models without modulations}, 
   journal={BiBoS preprint},
   number={08-03-278},
   eprint={http://www.physik.uni-bielefeld.de/bibos/},
   year={2008},
}
\bib{BZ2}{article}{
   author={Boutet de Monvel, Anne},
   author={Zielinski, Lech},
   title={Explicit error estimates for eigenvalues of some unbounded Jacobi matrices},
   conference={
   title={Spectral Theory, Mathematical System Theory, Evolution Equations,
   Differential and Difference Equations: IWOTA10},},
   book={
      series={Oper. Theory Adv. Appl.},
      volume={221},
      publisher={Birkh\"auser Verlag},
      place={Basel},},
   date={2012},
   pages={187--215},
}
\bib{BZ3}{article}{
   author={Boutet de Monvel, Anne},
   author={Zielinski, Lech},
   title={Approximation of eigenvalues for unbounded Jacobi matrices 
   using finite submatrices},  
   journal={Cent. Eur. J. Math.}, 
   volume={12},
   date={2014},
   number={3},
   pages={445--463},
} 
\bib{BZ4}{article}{
   author={Boutet de Monvel, Anne},
   author={Zielinski, Lech},
   title={Asymptotic behavior of large eigenvalues of a modified  
   Jaynes--Cummings model}, 
   conference={
      title={Spectral Theory and Differential Equations},
   },
   book={
      series={Amer. Math. Soc. Transl. Ser. 2},
      volume={233},
      publisher={Amer. Math. Soc., Providence, RI},
   },
   date={2014},
   pages={77--93},
}  
\bib{CJ}{article}{
   author={Cojuhari, Petru A.},
   author={Janas, Jan},
   title={Discreteness of the spectrum for some unbounded Jacobi matrices},
   journal={Acta Sci. Math. (Szeged)},
   volume={73},
   date={2007},
   number={3-4},
   pages={649--667},
}
\bib{JN}{article}{
   author={Janas, Jan},
   author={Naboko, Serguei},
   title={Infinite Jacobi matrices with unbounded entries: asymptotics of
   eigenvalues and the transformation operator approach},
   journal={SIAM J. Math. Anal.},
   volume={36},
   date={2004},
   number={2},
   pages={643--658},
}
\bib{Ma}{article}{
   author={Malejki, Maria},
   title={Asymptotics of large eigenvalues for some discrete unbounded
   Jacobi matrices},
   journal={Linear Algebra Appl.},
   volume={431},
   date={2009},
   number={10},
   pages={1952--1970},
}
\bib{Tur}{article}{
   author={Tur, {\`E}. A.},
   title={Jaynes--Cummings model: solution without rotating wave approximation},
   journal={Optics and Spectroscopy},
   volume={89},
   date={2000},
   number={4},
   pages={574--588},
}
\bib{Yan}{article}{
   author={Yanovich, Eduard A.},
   title={Asymptotics of eigenvalues of an energy operator in a problem of quantum physics},
   conference={
      title={Operator Methods in Mathematical Physics},
   },
   book={
      series={Oper. Theory Adv. Appl.},
      volume={227},
      publisher={Birkh\"auser/Springer Basel AG, Basel},
   },
   date={2013},
   pages={165--177},
}
\end{biblist}
\end{bibdiv}
\end{document}